\documentclass{article}
\pdfoutput=1
\usepackage{amsthm}
\usepackage[T1]{fontenc}

\usepackage{
	amsmath,
	amsthm,
	amssymb,
	wrapfig,
	cases,
	mathtools,
	thmtools,
	array,
	bbm,
	booktabs,
	upgreek,
	xkvltxp,
	fixme,
	listings,
	multirow,
	xargs,
	xstring,
	multicol,
	graphicx,
	float,
	color,
	indentfirst,
	accents,
	ifthen,
	tikz,
	etoolbox,
 	caption,    
 	subcaption,
	lscape
	}

\usepackage{algorithmic}
\usepackage[linesnumbered,ruled,vlined]{algorithm2e}
\DontPrintSemicolon
\usepackage[colorlinks=true, linkcolor=blue!70!black, citecolor=blue!70!black,urlcolor=black,breaklinks=true]{hyperref}
\usepackage{natbib}
\usetikzlibrary{shapes,arrows,trees,automata,positioning}

\makeatletter
\renewcommand{\maketag@@@}[1]{\hbox{\m@th\normalsize\normalfont#1}}%
\makeatother

\makeatletter
\let\reftagform@=\tagform@
\def\tagform@#1{\maketag@@@{\ignorespaces\textcolor{gray}{(#1)}\unskip\@@italiccorr}}
\renewcommand{\eqref}[1]{\textup{\reftagform@{\ref{#1}}}}
\makeatother

\usepackage[capitalize]{cleveref}

\hypersetup{%
    bookmarksnumbered, bookmarksopen=true, bookmarksopenlevel=1,%
}

\crefname{lemma}{Lemma}{Lemmas}
\crefname{corollary}{Corollary}{Corollaries}
\crefname{theorem}{Theorem}{Theorems}
\crefname{assumption}{Assumption}{Assumptions}

\declaretheorem[name=Remark]{rmk} %
\declaretheorem[name=Theorem]{theorem}

\declaretheorem[name=Lemma]{lemma}
\declaretheorem[name=Proposition]{proposition}
\declaretheorem[name=Corollary]{corollary}

\declaretheorem[name=Definition]{definition}
\declaretheorem[name=Assumption]{assumption}
\declaretheorem[name=Example]{example} %

\usepackage{crossreftools}
\newcommand{\CC}{\mathbb{C}}

\newcommand{\EE}{\mathbb{E}}

\newcommand{\II}{\mathbb{I}}

\newcommand{\MM}{\mathbb{M}}
\newcommand{\NN}{\mathbb{N}}

\newcommand{\PP}{\mathbb{P}}

\newcommand{\RR}{\mathbb{R}}

\newcommand{\Aa}{\mathcal{A}}
\newcommand{\Bb}{\mathcal{B}}

\newcommand{\Ff}{\mathcal{F}}

\newcommand{\Ll}{\mathcal{L}}
\newcommand{\Mm}{\mathcal{M}}

\newcommand{\Oo}{\mathcal{O}}
\newcommand{\Pp}{\mathcal{P}}

\newcommand{\Xx}{\mathcal{X}}

\newcommand{\one}{\mathbf{1}}

\def\[#1\]{\begin{equation}\begin{aligned}#1\end{aligned}\end{equation}}
\def\*[#1\]{\begin{equation*}\begin{aligned}#1\end{aligned}\end{equation*}}

\def\s*[#1\s]{\small\begin{align*}#1\end{align*}\normalsize}

\newcommand{\lcr}[3]{\left #1 #2 \right #3} %
\newcommand{\lcrx}[4][{-1}]{
	\IfEq{#1}{-1}{\left #2 {{{{#3}}}} \right #4}{
   	\IfEq{#1}{0}{#2 {{{{#3}}}} #4}{
	\IfEq{#1}{1}{\bigl #2 {{{{#3}}}} \bigr #4}{
	\IfEq{#1}{2}{\Bigl #2 {{{{#3}}}} \Bigr #4}{
	\IfEq{#1}{3}{\biggl #2 {{{{#3}}}} \biggr #4}{
	\IfEq{#1}{4}{\Biggl #2 {{{{#3}}}} \Biggr #4}{
    \GenericWarning{"4th argument to lcrx must be -1, 0, 1, 2, 3, or 4"}
    }}}}}}} %

\newcommand{\stk}[2]{\ensuremath{\stackrel{\text{#2}}{#1}}}

\newbool{hideupper}

\newcommand{\upper}[1]{\ifbool{hideupper}{}{^{(#1)}}}

\makeatletter
\newcommand{\subalign}[1]{%
  \vcenter{%
    \Let@ \restore@math@cr \default@tag
    \baselineskip\fontdimen10 \scriptfont\tw@
    \advance\baselineskip\fontdimen12 \scriptfont\tw@
    \lineskip\thr@@\fontdimen8 \scriptfont\thr@@
    \lineskiplimit\lineskip
    \ialign{\hfil$\m@th\scriptstyle##$&$\m@th\scriptstyle{}##$\crcr
      #1\crcr
    }%
  }
}
\makeatother

\renewcommand{\Pr}{\mathbb{P}} %
\def\EE{\mathbb{E}} %

\newcommand{\cPr}[1]{
 \Pr^{#1}}  %

\newcommand{\cEE}[1]{
 \EE^{#1}} %

\newcommand{\EEE}[1]{\underset{#1}{\EE}}

\DeclareMathOperator*{\Cov}{Cov} %
\DeclareMathOperator*{\Var}{Var} %

\newcommand{\KLname}{KL}
\DeclareMathOperator{\KLop}{\KLname}
\newcommand{\KL}[2]{\KLop\lcr({#1\vphantom{#2}\hspace{1pt}}\Vert\hspace{-3pt}\lcr.{#2\vphantom{#1}})}
\newcommand{\stT}{\ \text{s.t.}\ }
\newcommand{\andT}{\ \text{and}\ }
\newcommand{\orT}{\ \text{or}\ }

\newcommand{\withT}{\ \text{with}\ }

\newcommand{\ind}[1]{\II_{\sbra{{#1}}}} %

\def\multiset#1#2{\ensuremath{\left(\kern-.3em\left(\genfrac{}{}{0pt}{}{#1}{#2}\right)\kern-.3em\right)}}

\DeclareMathOperator*{\argmin}{\arg\min} %
\DeclareMathOperator*{\argmax}{\arg\max} %
\DeclareMathOperator*{\newlim}{\mathrm{lim}\vphantom{\mathrm{infsup}}}
\DeclareMathOperator*{\newmin}{\mathrm{min}\vphantom{\mathrm{infsup}}}
\DeclareMathOperator*{\newmax}{\mathrm{max}\vphantom{\mathrm{infsup}}}
\DeclareMathOperator*{\newinf}{\mathrm{inf}\vphantom{\mathrm{infsup}}}
\DeclareMathOperator*{\newsup}{\mathrm{sup}\vphantom{\mathrm{infsup}}}
\renewcommand{\lim}{\newlim}
\renewcommand{\min}{\newmin}
\renewcommand{\max}{\newmax}
\renewcommand{\inf}{\newinf}
\renewcommand{\sup}{\newsup}

\newcommand{\tr}{\mathrm{Tr}} %
\newcommand{\adj}{^{\dag}} %
\DeclareMathOperator{\diag}{diag} %
\newcommand{\dee}{\mathrm{d}} %
\newcommand{\grad}{\nabla} %
\newcommand{\hess}{\grad^{\otimes 2}} %
\newcommand{\rnderiv}[3][{-1}]{
\IfEq{#1}{inline}{\dee #2 / \dee #3}{\frac{\dee #2}{\dee #3}}}

\newcommand{\distto}{\rightsquigarrow}

\newcommand{\distas}{\sim}
\newcommand{\distiidas}{\stk{\distas}{iid}}

\newcommand{\pushfwdmeas}[2]{{{#1}_{\sharp} #2}}
\newcommand{\Law}{\Ll}
\newcommand{\cLaw}[1]{\Law^{#1}} %

\newcommand{\IID}{i.i.d.}

\newcommand{\normaldist}{\mathrm{N}}

\newcommand{\unifdist}{\mathrm{Unif}}

\newcommand{\rbra}[2][{-1}]{\lcrx[#1] ( {#2} ) }
\newcommand{\cbra}[2][{-1}]{\lcrx[#1] \{ {#2} \} }
\newcommand{\sbra}[2][{-1}]{\lcrx[#1] [ {#2} ] }

\newcommand{\abs}[2][{-1}]{\lcrx[#1] \vert {#2} \vert }
\newcommand{\set}[2][{-1}]{\lcrx[#1] \{ {#2} \}}
\newcommand{\floor}[2][{-1}]{\lcrx \lfloor {#2} \rfloor}

\newcommand{\norm}[2][{-1}]{\lcrx[#1] \Vert {#2} \Vert}
\newcommand{\inner}[3][{-1}]{\lcrx[#1] \langle {{#2},\ {#3}} \rangle}

\newcommand{\defas}{:=}  %
\newcommand{\revdefas}{=:}  %

\newcommand{\Nats}{{\NN}}
\newcommand{\NatsO}{\Nats\cup\set{0}}

\newcommand{\Reals}{\RR}
\newcommand{\Complex}{\CC}

\newcommand{\PosReals}{\Reals_+}

\newcommand{\range}[2][{1}]{
	\IfEq{#1}{1}{\sbra{#2}}{\sbra{#2}_{#1}}}
\newcommand{\rangeO}[2][{0}]{
	\IfEq{#1}{0}{\sbra{#2}_0}{\sbra{#2}_{#1}}}

\newcommand{\ointer}[2][{-1}]{\lcrx[#1] ( {#2} ) }

\newcommand{\cointer}[2][{-1}]{\lcrx[#1] [ {#2} ) }

\newcommand{\union}{\cup}

\newcommand{\Intersect}{\bigcap}

\DeclareMathOperator{\domain}{dom}

\DeclareMathOperator{\interior}{interior}

\DeclareMathOperator{\convhull}{Conv}

\definecolor{WowColor}{rgb}{.75,0,.75}
\definecolor{SubtleColor}{rgb}{0,0,.50}

\newcommand{\NA}[1]{\textcolor{SubtleColor}{ {\tiny \bf ($\star$)} #1}}
\newcommand{\LATER}[1]{\textcolor{SubtleColor}{ {\tiny \bf ($\dagger$)} #1}}
\newcommand{\TBD}[1]{\textcolor{SubtleColor}{ {\tiny \bf (!)} #1}}
\newcommand{\PROBLEM}[1]{\textcolor{WowColor}{ {\bf (!!)} {\bf #1}}}

\newcounter{margincounter}
\newcommand{\displaycounter}{{\arabic{margincounter}}}
\newcommand{\incdisplaycounter}{{\stepcounter{margincounter}\arabic{margincounter}}}

\newcommand{\fTBD}[1]{\textcolor{SubtleColor}{$\,^{(\incdisplaycounter)}$}\marginpar{\tiny\textcolor{SubtleColor}{ {\tiny $(\displaycounter)$} #1}}}

\newcommand{\fPROBLEM}[1]{\textcolor{WowColor}{$\,^{((\incdisplaycounter))}$}\marginpar{\tiny\textcolor{WowColor}{ {\bf $\mathbf{((\displaycounter))}$} {\bf #1}}}}

\newcommand{\fLATER}[1]{\textcolor{SubtleColor}{$\,^{(\incdisplaycounter\dagger)}$}\marginpar{\tiny\textcolor{SubtleColor}{ {\tiny $(\displaycounter\dagger)$} #1}}}

\DeclareRobustCommand{\suppresscomments}{
\renewcommand{\LATER}[1]{}
\renewcommand{\fLATER}[1]{}
\renewcommand{\TBD}[1]{}
\renewcommand{\fTBD}[1]{}
\renewcommand{\PROBLEM}[1]{}
\renewcommand{\fPROBLEM}[1]{}
\renewcommand{\NA}[1]{##1}  %
}

\newcommand{\transposesym}{\top}
\newcommand{\transpose}{^{\transposesym}}

\DeclareMathOperator*{\Sym}{Sym} %

\newcommand{\mineigval}{\lambda_{\min}}
\newcommand{\iact}{\tau}  %

\newcommand{\driftMat}{B}
\newcommand{\diffusionMat}{A}

\def\cadlag{c\`adl\`ag\xspace}

\newcommand{\iid}{I.I.D.}
\newcommand{\innov}{\xi}
\newcommand{\innoviter}[1]{{\innov_{#1}}}
\newcommand{\aniso}{\Lambda}
\newcommand{\precon}{\Gamma}
\newcommand{\massmatrix}{M}
\newcommand{\stochgrad}[2]{\hat G\upper{#1}_{#2}}

\newcommand{\stepsize}{h}

\newcommand{\dataidx}{i}

\newcommand{\mbname}{I}
\newcommand{\minibatch}[2]{{\mbname\upper{#1}_{#2}}}
\newcommand{\minibatchseqlong}[1]{{\rbra[0]{\minibatch{#1}{\iternum}}_{\iternum\in\Nats}}}

\newcommand{\mbidx}{j}
\newcommand{\mbdataidx}[3]{{\mbname\upper{#1}_{#2}(#3)}}

\newcommand{\priorposteriordist}{\pi}
\newcommand{\prior}{{\priorposteriordist^{(0)}}}
\newcommand{\loglik}{\ell}
\newcommand{\reg}{r}
\newcommand{\posterior}[1]{\priorposteriordist\upper{#1}}
\newcommand{\potential}[1]{\mathcal{U}\upper{#1}}

\newcommand{\priorposteriormeas}{\Pi}
\newcommand{\priormeas}{{\priorposteriormeas\upper{0}}}
\newcommand{\posteriormeas}[1]{\priorposteriormeas\upper{#1}}

\newcommand{\samplesize}{n}
\newcommand{\batchsize}{b}
\newcommand{\iternum}{k}
\newcommand{\iternumdum}{k'}
\newcommand{\numpasses}{m}

\newcommand{\testfun}{f}
\newcommand{\data}{X}
\newcommand{\dumdata}{x}
\newcommand{\dataspace}{\Xx}

\newcommand{\dataset}[1]{{\mathbf{\data}\upper{#1}}}
\newcommand{\datadist}{P}

\newcommand{\parm}{\theta}
\newcommand{\parmdum}{\theta}
\newcommand{\parmiter}[2]{{{\parm\upper{#1}_{#2}}}}
\newcommand{\iteravg}[2]{{\bar\parm\upper{#1}_{#2}}}
\newcommand{\parmdim}{d}
\newcommand{\parmspace}{{\Theta}}
\newcommand{\trueparm}{\parm_\star}
\newcommand{\limparm}{\vartheta}
\newcommand{\scaleparmiter}[2]{{\limparm\upper{#1}_{#2}}} %
\newcommand{\parmMLE}[1]{{\widehat \parm}\upper{#1}}
\newcommand{\parmLMLE}[1]{{\widehat \parm}\upper{#1}}

\newcommand{\localparm}{{\limparm}}

\newcommand{\limparmavg}{\bar\limparm}

\newcommand{\probmeasures}[1]{{\Mm_{1,+}\rbra{#1}}}
\newcommand{\modeldist}[1]{{Q_{#1}}}
\newcommand{\modeldens}[1]{{q_{#1}}}

\newcommand{\basemeasure}{{\mu}}

\newcommand{\frobprod}{:}

\newcommand{\generator}{A}

\newcommand{\contfuns}{C}
\newcommand{\contbddfuns}{{\overline C}}
\newcommand{\contvanfuns}{\hat \contfuns}

\newcommand{\smoothcompactsuppfuns}{C^\infty_c}
\newcommand{\kdiffuns}[1]{C^{#1}}
\newcommand{\bddmeasfuns}{B}

\newcommand{\normTV}[1]{\norm{{#1}}_{\textup{TV}}}
\newcommand{\normInf}[1]{\norm{{#1}}_{{\infty}}}

\newcommand{\invtemp}{\beta}

\newcommand{\SA}[1]{{\Sigma_{#1}}}
\newcommand{\lebesgue}{\lambda}
\newcommand{\model}{\MM}
\newcommand{\localmodel}[1]{{\MM\upper{#1}}}

\newcommand{\jointPP}{{\PP_{\priormeas,\model}}}
\newcommand{\jointEE}{{\EE_{\priormeas,\model}}}
\newcommand{\jointCEE}[1]{{\EE_{\priormeas,\model}^{#1}}}
\newcommand{\finfo}{\mathcal{J}}
\newcommand{\finfoOpt}{\finfo_\star}
\newcommand{\finfoEmp}[1]{{\widehat \finfo}\upper{#1}}
\newcommand{\vinfo}{\mathcal{I}}
\newcommand{\vinfoOpt}{\vinfo_\star}
\newcommand{\vinfoEmp}[1]{{\widehat \vinfo}\upper{#1}}
\newcommand{\loc}[1]{\textup{loc}_{#1}}
\newcommand{\semigroup}{T}
\newcommand{\mkernel}{U}
\newcommand{\bddlinear}{\Bb}
\newcommand{\banachspace}{L}
\newcommand{\dualbanachspace}{L'}
\newcommand{\kernel}{k}
\newcommand{\skspace}{D}

\newcommand{\timescale}{\alpha}

\newcommand{\const}[1]{c_{#1}}

\newcommand{\weiner}{W}

\newcommand{\subseqidx}{m}
\newcommand{\subsubseqidx}{k}
\newcommand{\subsubsubseqidx}{j}

\newcommand{\subseqsamplesize}{\samplesize_\subseqidx}
\newcommand{\subsubseqsamplesize}{\samplesize_{\subseqidx_\subsubseqidx}}
\newcommand{\subsubsubseqsamplesize}{\samplesize_{\subseqidx_{\subsubseqidx_\subsubsubseqidx}}}

\newcommand{\compactset}{K}
\newcommand{\radius}{R}
\newcommand{\lipschitzconst}{L}

\newcommand{\equicontinuityfun}{\rho}
\newcommand{\ball}{B}
\newcommand{\samplespace}{\Omega}

\crefname{premise}{premise}{premises}

\newcounter{termcounter}
\renewcommand{\thetermcounter}{\Alph{termcounter}}
\crefname{term}{term}{terms}
\creflabelformat{term}{#2\textup{(#1)}#3}

\makeatletter
\def\term{\@ifnextchar[\term@optarg\term@noarg}%
\def\term@optarg[#1]#2{%
  \textup{#1}%
  \def\@currentlabel{#1}%
  \def\cref@currentlabel{[][2147483647][]#1}%
  \cref@label[term]{#2}}
\def\term@noarg#1{%
  \refstepcounter{termcounter}%
  \textup{\thetermcounter}%
  \cref@label[term]{#1}}
\makeatother

\newcommand{\incrInnov}[1]{\Delta\upper{#1}_{\innov}}
\newcommand{\incrPrior}[1]{\Delta\upper{#1}_{\prior}}
\newcommand{\incrLoglik}[1]{\Delta\upper{#1}_{\loglik}}
\newcommand{\incrTotal}[1]{\Delta\upper{#1}}

\newcommand{\incrInnovNup}{\Delta_{\innov}}
\newcommand{\incrPriorNup}{\Delta_{\prior}}
\newcommand{\incrLoglikNup}{\Delta_{\loglik}}

\newcommand{\statcov}{Q_\infty}
\newcommand{\margcov}[1]{Q_{{#1}}}
\newcommand{\avgcov}[1]{\bar Q_{{#1}}}

\newcommand{\statmeas}{\nu}

\newcommand{\localscale}{w}

\newcommand{\stepsizePow}{\mathfrak{h}}
\newcommand{\invtempPow}{\mathfrak{t}}
\newcommand{\localscalePow}{\mathfrak{w}}
\newcommand{\timescalePow}{\mathfrak{a}}
\newcommand{\batchsizePow}{\mathfrak{b}}

\newcommand{\assmoment}[1]{{p_{\text{#1}}}}
\newcommand{\asspower}[1]{{q_{\text{#1}}}}

\newcommand{\boundarymap}{\Pp}

\newcommand{\driftName}{\texttt{d}}
\newcommand{\langNoiseName}{\texttt{g}}
\newcommand{\mbNoiseName}{\texttt{mb}}

\usepackage{graphicx} 
\usepackage{grffile}

\suppresscomments

\allowdisplaybreaks[3]

\usepackage[shortlabels]{enumitem}
\usepackage[letterpaper, margin=1in]{geometry} %

\usepackage{color}
\newcommand{\add}[1]{#1}
\newcommand{\mydelete}[1]{}
\newcommand{\delete}[1]{}
\newcommand{\edit}[2]{#2}
\newcommand{\mnote}[1]{}

\title{Tuning Stochastic Gradient Algorithms for \\ Statistical Inference via Large-Sample Asymptotics}

\usepackage{authblk}
\author[1,2]{Jeffrey Negrea}
\author[3]{Jun Yang}
\author[4]{Haoyue Feng}
\author[5,2]{Daniel M. Roy}
\author[4,6]{Jonathan H. Huggins}

\affil[1]{Department of Statistics and Actuarial Science, University of Waterloo, Canada}
\affil[2]{Vector Institute, Canada}
\affil[3]{Department of Mathematical Sciences, University of Copenhagen, Denmark}
\affil[4]{Department of Mathematics \& Statistics, Boston University, USA}
\affil[5]{Department of Statistical Sciences, University of Toronto, Canada}
\affil[6]{Faculty of Computing \& Data Sciences, Boston University, USA}

\begin{document}

\maketitle

\begin{abstract}
\add{The} tuning of stochastic gradient algorithms (SGAs) for optimization and sampling is often based on heuristics and trial-and-error rather than generalizable theory.
We address this theory--practice gap by characterizing the large-sample statistical asymptotics of SGAs via a joint step-size--sample-size scaling limit.
We show that iterate averaging with a large fixed step size is robust to the choice of tuning parameters and asymptotically has covariance proportional to that of the MLE sampling distribution.
We also prove a Bernstein--von Mises-like theorem to guide tuning, including for generalized posteriors that are robust to model misspecification. 
Numerical experiments validate our results and recommendations in realistic finite-sample regimes.
\add{Our work lays the foundation for a systematic analysis of other stochastic gradient Markov chain Monte Carlo algorithms for a wide range of models.}

\end{abstract}

\section{Introduction}\label{sec:introduction}
Stochastic gradient algorithms, \add{which} were originally proposed as optimization and root finding methods by \citet{robbins1951stochastic}, \delete{and} have become
the standard approach to large-scale optimization in statistics and machine learning.
\delete{This}\add{Their} success \delete{is because}\add{can be attributed to} the reduction in per-iteration computational complexity from subsampling outweigh\delete{s}\add{ing} the accuracy loss from stochastic approximation for empirical objectives.
Hence, stochastic \delete{optimization methods}\add{gradient algorithms} scale more favourably with the sample size and model complexity than their deterministic counterparts \citep{bach2011non-asymptotic,Goodfellow:2016:DL-book}.
Over the past decade, this scalability has also lead to tremendous growth in the use \add{of} stochastic gradient Markov chain Monte Carlo sampling algorithms, particularly in machine learning \citep{welling2011bayesian,Nemeth:2021:SGLD}.

Most analyses of stochastic gradient optimization procedures such as stochastic gradient descent (SGD) focus on the parameter error
or the optimality gap \citep[e.g.,][]{bach2011non-asymptotic,kushner2003stochastic,nemirovski2009robust,Reddi:2018:Adam},
while analyses of stochastic gradient sampling procedures such as stochastic gradient Langevin dynamics (SGLD)
focus on how well the empirical distribution of the iterates approximates the posterior
\citep[e.g.,][]{teh2016consistency,vollmer2016exploration,brosse2018sgld,Baker.2019,Nemeth:2021:SGLD,raginsky2017non,durmus2017nonasymptotic,durmus2019high}.
\add{These results often rely on settings for tuning parameters that fall outside of standard practice.}
\add{It is an important challenge to explain why, empirically, stochastic gradient algorithms appear successful with previously unvalidated tunings (e.g., large step size and small batch size).}
\delete{Yet\add{,} in practice\add{,} stochastic gradient algorithms are successful even when used with tuning parameter combinations
insufficient to result in accurate approximations according to the standard theory (e.g., large step size and small batch size).}
The lack of an explanatory theory forced users to rely on heuristic and problem-specific approaches to tuning parameters.

We take a step toward closing this gap between theory and practice when the step size is fixed across iterations but decreases with the sample size.
The fixed--step-size setting proves to be practically relevant for optimization because convergence to a near-optimum is rapid and robust to the precise
step size choice \citep{bach2011non-asymptotic,dieuleveut2020bridging} while
for sampling, using a fixed--step-size leads to better mixing time behaviour:
the number of iterations until the next approximately independent sample is constant, unlike in the decreasing-step size regime where the number of
iterations until the next approximately independent sample increases without bound \citep{teh2016consistency,vollmer2016exploration}.

Our main result characterizes the statistical scaling limits of stochastic gradient algorithms as the sample size tends to infinity.
We show that the sample paths of a very general class of preconditioned stochastic gradient algorithms\delete{---including stochastic gradient descent with and without additional Gaussian noise, momentum, and/or acceleration---} converge
to the sample paths of an Ornstein--Uhlenbeck process under relatively mild conditions. \add{The class of algorithms includes stochastic gradient descent with and without additional Gaussian noise, momentum, and/or acceleration.}
Notably, however, while the asymptotic guarantees in the decreasing step size case often require an impractically large number of iterates,
numerical experiments show that our constant step size averaging result can hold after a small number of passes over the dataset.
\delete{For sampling, we show that it is more computationally efficient (in terms of mixing) to use SGLD to sample from a distribution that is far from the
posterior but has better robustness to model misspecification.}
\add{For sampling, we show that it is even possible to leverage stochastic gradients to sample modifications to the posterior that have better robustness to model misspecification.}
\delete{This result suggests that stochastic gradients, rather than being a problem in need of solution, can be beneficial---or at the very least are harmless.}\add{This result suggests that stochastic gradients have a potentially beneficial (or at least benign) role to play, rather than one that creates accuracy problems in exchange for computational efficiency.}

\add{Because the guarantees we provide are asymptotic in the sample size, it is possible that they may not be representative for a particular dataset. 
Therefore, we complement our asymptotic results with three numerical experiments to demonstrate that the limiting behaviour often predicts actual performance. 
These include}
\delete{Therefore, we confirm our theoretical findings through three numerical experiments:}
a simulation study with a Gaussian location model, and two real-data experiments
(a logistic regression example with 1 million observations, and a misspecified Poisson regression example with 150,000 observations). 

\subsection{Implications for sampling}

\add{
\citet{Nemeth:2021:SGLD} recently identified several key areas for stochastic gradient MCMC (SG-MCMC) research. 
Our work makes significant strides in two of these areas for fixed–step-size variants of SGLD through our analysis of their large-sample asymptotics.
One 
key area
they identify is the need for general theoretical results beyond the log-concave regime that are not asymptotic in the number of iterations. 
We move beyond the log-concave regime by using large-sample asymptotics, analogous to the applicability of the Bernstein-von Mises theorem regardless of the convexity of the likelihood. 
Under tuning regimes relevant to statistical inference, our results apply after a constant number of epochs (i.e., passes over the full dataset). 
Another 
key area
identified by \citet{Nemeth:2021:SGLD} is the need for methods for robust and/or adaptive tunings. Ideally, tunings ought to be automatable for non-experts to use.
We use our results to make recommendations on the tuning of these methods in the large-sample setting (see \cref{tab:tuning-combos}), 
which is especially relevant in practice since stochastic gradient MCMC algorithms are typically used when the sample size is large. 
In particular, a large class of bad tunings whose large-sample asymptotics do not match the large-sample asymptotics of the target, or whose large-sample asymptotic local mixing is very slow, can be immediately identified and ruled out.
Good tunings with the correct large-sample asymptotics and rapid asymptotic local mixing can also be identified as candidates for use and possibly fine-tuned using
other methods \citep[e.g.,][]{Coullon2023}. 
The guidance we derive in this way does not require additional expertise to use and could be implemented in an automated way.
Moreover, our statistical perspective on the large-sample asymptotics of these methods leads to the insight that other benefits can be obtained by targeting statistically robust modifications of the posterior distribution,
a direction not foreseen by \citeauthor{Nemeth:2021:SGLD}.}

\add{
We illustrate the implications of our results with two recent applications of SG-MCMC in the statistics literature. 
\begin{example} \label{exa:pollock}
\citet{pollock2020quasi} benchmark their subsampling-based MCMC algorithm against SGLD. They tuned SGLD using the best-available-at-the-time theoretical guidance \citep{teh2016consistency} and other best practices, 
including variance-reduced stochastic gradients \citep{Baker.2019}. 
However, their implementation of SGLD mixes slowly and does not appear to be sampling from the posterior. 
This can be attributed to two causes. 
First, because they use a decreasing step size, each nearly independent sample takes an increasing number of epochs to reach. 
Second, even if they were to use a fixed--step-size---or to run the Langevin diffusion for the posterior directly---it would have mixed slowly due to ill-conditioning of the posterior distribution. 
Based on the scaling limit we derive, the poor approximation quality due to ill-conditioning would have been foreseen, and a fixed--step-size sampler could have been appropriately tuned. Our theory predicts that, because SGLD was not preconditioned adequately,  mixing would be very slow---and how much slower it is relative to the optimal preconditioner. Furthermore, since the optimal preconditioner according to our theory was used in the implementation of their method, the numerical comparison overstates the relative benefits of their proposed method versus SGLD. We demonstrate this on the same data as used by \citet{pollock2020quasi} in our \cref{experiment2}. 
In short, this example exhibits how our results are directly exploitable: our tuning recommendations would have resolved the mixing problems of SGLD seen by \citet{pollock2020quasi}, and led to a more meaningful comparator to their method. 
\end{example}}%
\add{\begin{example} \label{exa:nemeth}
\citet[Section 6.3]{Nemeth:2021:SGLD} compare various SG-MCMC algorithms on a challenging matrix factorization problem. 
Due to the lack of actionable tuning advice in the literature, they use the kernel Stein discrepancy (KSD) to select the step size. 
They initialise the variance-reduced SG-MCMC algorithms at the maximum a posteriori solution.
Due to a pathology of the KSD \citep{Coullon2023}, this results in selection of the smallest possible step size of $10^{-10}$,
which leads to the variance-reduced chains essentially remaining at their initialization.
As a result, \citet[Figure 6]{Nemeth:2021:SGLD} incorrectly suggests the that the variance-reduced algorithms had much lower predictive error than other SG-MCMC algorithms, 
when in fact it illustrates that the maximum a posteriori solution provides small test error---but of course no uncertainty quantification. 
Our theory predicts the observed poor approximation to the posterior and lack of meaningful uncertainty quantification from the ``stuck'' chains. 
Using our recommendations would have avoided the undetected pathological slow-mixing behaviour resulting from the use of the KSD. 
\end{example}
}

\subsection{Implications for optimization and frequentist inference}

\add{
Our theory provide rigorous foundations and new insights into the use of iterate averaging with fixed--step-size SGAs.
Our main result differs from the seminal works on scaling limits in stochastic approximations \citep{kushner1981asymptotic, pflug1986stochastic, walk1977invariance, Kushner:1993, kushner2003stochastic} in both the nature of our analysis and the required regularity conditions. 
We analyze the setting where the source of the stochastic gradients is itself random and undergoes stochastic convergence. 
This is an important distinction because this joint limit is the pivotal object that we study in the present work and is required to address our research questions.
In further contrast to our work, \citet{kushner1981asymptotic,pflug1986stochastic,Kushner:1993, kushner2003stochastic} require restrictive assumptions that are not readily lifted to this ``doubly stochastic convergence'' case. 
The assumptions required by our analysis, on the other hand, are quite weak. We allow the batch size used to compute the stochastic gradient to be constant or depend on the dataset size, and allow the batches to be sampled with or without replacement. 
We only require the local maximizer to converge in probability and we do not assume the model is correctly specified.
At the same time, our results are stronger than those achieved by previous analyses since we characterize both the sample paths of the iterates and the complete stationary distribution. %
For example, \citet{walk1977invariance} only considers decreasing step sizes and demonstrates asymptotic normality of the marginal distribution, unlike our a functional/path limit results. 
Characterising the sample-path distribution is critical to analysing not just iterate averages but also the mixing time.
}

\add{ Despite not directly applying to the statistical setting with fixed step size, the work of \citet{kushner1981asymptotic, pflug1986stochastic, walk1977invariance, Kushner:1993, kushner2003stochastic}
has been an important source of motivation for more recent methodological developments. 
We highlight two examples.}
\add{
\begin{example}
\citet{li2018statistical} propose a method for constructing samples from a local asymptotic fiducial distribution (one whose credible regions are asymptotic confidence sets at the same significance) by magnifying the deviations of SGD from the mode. They point out that the intuition underlying their results aligns with the Ornstein--Uhlenbeck scaling limit of stochastic gradient algorithms (including references to \citep[e.g.,][]{kushner1981asymptotic, pflug1986stochastic}).
However, the rigorous proof of their results does not leverage the continuous-time limit, and uses stronger assumptions  -- such as weak strong convexity (equivalently, strong convexity of the composite objective) -- than we require. The gap between the seminal work on Ornstein--Uhlenbeck limits of SGAs and the intuition they formed is, again, the need for a joint stochastic limit of the decreasing step size and stochastically varying objective function. Using our results, their heuristic explanation based on the continuous-time limit could be made rigorous and their findings and methods could be extended beyond weak strong convexity.
\end{example}
}
\add{
\begin{example}
A special case of the scaling limit we derive, and some of resulting the tuning recommendations, was conjectured by \citet{mandt2017stochastic} based on the heuristic combination of statistical asymptotics of the posterior distribution (Bernstein--von Mises) and results from the stochastic approximation literature.
However, being led mainly by heuristics, \citet{mandt2017stochastic} arrive at an erroneous conclusion on the effect of iterate averaging that. 
In particular, they claim that ``there exist conditions [on the data-generating process and sample size] under which iterate averaging generates one true posterior sample per pass over the data.'' However, their heuristic approach fails to account for ``low order'' terms, which our analyses in \cref{sec:optimization-iteravg} reveal are \emph{not} actually low-order when only making a single pass over the dataset. 
This difference is observed empirically in the numerical experiments we present in \cref{experiment1}.
Our results also go far beyond those conjectured by \citet{mandt2017stochastic}, broadening the applicability of this category of result, including fixed and growing batch-size regimes, non-traditional spatial scalings and concentration rates, and incorporating comparison to other asymptotic target distributions of interest. 
\end{example}
}
Another practical benefit of our rigour for the stochastic limit of stochastic processes is that we can clearly distinguish previous heuristics which can be turned into precise claims (the functional Bernstein--von Mises results, for instance), and those which require additional qualification or restricted application (such as mixing time results).

\delete{
As such, our results can be viewed a generalization and formalization of \mbox{\citet{mandt2017stochastic}}'s  heuristic arguments, and open the way
for further generalizations to situations where heuristics provide minimal insight such as infinite- dimensional models and models where the number of
parameters scales with the sample size.
Our results also complement those of \mbox{\citet{Kushner:1993}} and \mbox{\citet{kushner2003stochastic}}, which provide the basis for the conjectures in \mbox{\citet{mandt2017stochastic}}, and establish weak convergence of stochastic gradient algorithms to an OU process in a variety of settings.
Notably, however, they do not cover cases where either (i) the objective function is itself stochastic (in particular arising as the random loss function based on a finite \iid sample) or
(ii) the joint scaling of the objective function with the tuning parameters of the algorithm, both of which are required to obtain our asymptotic statistical results.
}

\subsection{Notation}

Let $\probmeasures{\mathcal{A}}$ denote the set of probability measures on the measure space $\mathcal{A}$
and let $\Nats \defas \set[0]{1, 2, \dots}$ denote the natural numbers.
For $\samplesize\in\Nats$, define $\range{\samplesize} \defas \set{1,\dots,\samplesize}$.
For $\parmdim \in \Nats$, denote the $\parmdim$-dimensional Gaussian distribution with mean $\mu \in \Reals^{\parmdim}$
and (positive semi-definite) covariance matrix $\Sigma \in \Reals^{\parmdim \times \parmdim}$ by $\normaldist_{\parmdim}(\mu, \Sigma)$.
For vectors $a, b \in \Reals^{\parmdim}$, define the outer product $a \otimes b \in \Reals^{\parmdim \times \parmdim}$ given by
$(a \otimes b)_{ij} = a_{i}b_{j}$ and write $a^{\otimes 2} \defas a \otimes a$.
Let $\grad \otimes \grad = \grad^{\otimes 2}$ denote the Hessian operator.
For random elements $(\xi_{k})_{k \in \Nats}$ and $\xi$, we write $\xi_{k} \distto \xi$ to denote convergence in distribution;
that is, $\xi_{k} \distto \xi$ if and only if for every bounded continuous function $f$, $\EE\{f(\xi_{k})\} \to \EE\{f(\xi)\}$ as $k\to\infty$.
We write $\Law(\xi)$ for the distribution (law) of a random element $\xi$, and $\cLaw{\nu}(\xi)$ for the conditional distribution of $\xi$ given another random element $\nu$.
For a square matrix $M$, define the symmetrization operator as $\Sym(M) \defas (M + M\transpose)/2$.
For a function $f:\Aa\to \banachspace$ with $\Aa$ a set and $(\banachspace,\norm{\cdot})$ a normed linear space, define $\normInf{f} \defas \sup_{a\in \Aa} \norm{f(a)}$.

\section{Stochastic Gradient Optimization and Sampling} \label{sec:sgas}

Let $\dataset{\samplesize} = (\data_{\dataidx})_{\dataidx=1}^{\samplesize} \in \dataspace^{\samplesize}$ denote a dataset
with observations $\data_{\dataidx}$ independently and identically distributed (\IID) from an unknown distribution $\datadist$.
For parameter $\parm \in \parmspace \subseteq \Reals^{\parmdim}$, consider the potential
$\potential{\samplesize}(\parm) \defas \reg\rbra[0]{\parm} + \sum_{\dataidx = 1}^{\samplesize} \loglik\rbra[0]{\parm; \data_\dataidx}$,
where typically $\loglik$ represents a log-likelihood or a negative loss function, and $\reg\rbra[0]{\parm}$ represents a regularizer or
a (possibly improper) log prior density $\log \prior(\parm)$ that is everywhere positive on $\parmspace$.

If $-\potential{\samplesize}(\parm)$ is interpreted as a (possibly regularized) loss, perhaps the most popular estimator for the (locally) optimal population parameter $\trueparm$
satisfying $\EE\cbra[0]{\grad\loglik\rbra{ \trueparm; \data_{1}}} = 0$, is the M-estimator
$\parmLMLE{{\samplesize}}$ satisfying the first-order optimality condition $\grad \potential{\samplesize}(\parmLMLE{{\samplesize}}) = 0$.
If $-\potential{\samplesize}(\parm)$ is interpreted as the negative log of the joint model density or as a generalized Bayesian loss \citep{bissiri2016general},
the quantity of interest is (usually) an expectation with respect to the (generalized) posterior density $\posterior{{\samplesize}}(\parm) \propto \exp\{-\potential{\samplesize}(\parm)\}$
of a function $\testfun : \parmspace \to \Reals^{\ell}$, which we denote $\posterior{{\samplesize}}(\testfun)$.
In either case, when $\samplesize$ is large relative to the computational cost of evaluating $\loglik\rbra[0]{\parm; \data_\dataidx}$,
classical optimization methods for approximating $\parmLMLE{{\samplesize}}$ (e.g., gradient descent or Newton--Raphson) and sampling methods for estimating $\posterior{{\samplesize}}(\testfun)$ (e.g., Metropolis--Hastings algorithms)
become computationally prohibitive.

Stochastic gradient algorithms provide a means of reducing the per-iteration computational cost of optimization and sampling methods.
To generate a sequence of iterates $\parmiter{\samplesize}{1}, \dots, \parmiter{\samplesize}{\iternum}, \ldots \in \parmspace$,
rather than computing exact gradients of $\samplesize^{-1}\potential{\samplesize}$ using the full dataset,
at iteration $\iternum$ a small batch of subsampled data is used instead to compute an unbiased gradient estimate
\[ \label{eq:stochastic-gradient-formula}
\textstyle\stochgrad{\samplesize}{\iternum}
\defas \frac{1}{\samplesize}\grad\reg\rbra{\parmiter{\samplesize}{\iternum}} + \frac{1}{\batchsize\upper{\samplesize}}\sum_{\mbidx=1}^{\batchsize\upper{\samplesize}} \grad\loglik\rbra[3]{ \parmiter{\samplesize}{\iternum};\ \data_{\mbdataidx{\samplesize}{\iternum}{\mbidx}}},
\]
where $\minibatchseqlong{\samplesize} \in \rbra[0]{\range{\samplesize}^\batchsize}^\Nats$ are an independent and identically distributed (\IID{})
sequence of uniform random samples  from $\set{1,\dots,\samplesize}$ of size $\batchsize\upper{\samplesize}$, which are
formed either with or without replacement.\footnote{``With replacement'' means $\minibatchseqlong{\samplesize} \distiidas \unifdist\rbra[0]{\range{\samplesize}^\batchsize}$,
and ``without replacement'' means $\minibatchseqlong{\samplesize} \distiidas \unifdist \rbra[0]{\set[0]{\mbname\in \range{\samplesize}^\batchsize: \rbra{j_1\neq j_2 \Rightarrow \mbname(j_1)\neq \mbname(j_2)}}}$.}

For optimization, the canonical approach is stochastic gradient descent (SGD), which has one-step update
\[\label{eq:sgd-iter-update}
	\parmiter{\samplesize}{\iternum+1}
		& = \parmiter{\samplesize}{\iternum} + \frac{\stepsize\upper{\samplesize}_{\iternum}}{2} \stochgrad{\samplesize}{\iternum},
\]
where $(\stepsize\upper{\samplesize}_{\iternum})_{\iternum \in \Nats}$ is a sequence of positive step sizes.
While optimal tuning of the last-iterate error is challenging, averaging the iterates can provide automatic optimal uncertainty quantification \citep{polyak1992acceleration,Kushner:1993,kushner2003stochastic}. %
More precisely, when $\stepsize_{\iternum} \propto \iternum^{-\varsigma}$ for $\varsigma \in (0,1)$,
the iterate average $\iteravg{\samplesize}{\iternum} \defas \frac{1}{\iternum}\sum_{\iternumdum=1}^{\iternum} \parmiter{\samplesize}{\iternumdum}$
satisfies
\[ \label{eq:iterate-averaging-optimality}
\lim_{\samplesize\to\infty}\lim_{\iternum \to \infty}\iternum\Cov(\iteravg{\samplesize}{\iternum})
	= \finfoOpt^{-1}\vinfoOpt\finfoOpt^{-1}
	= \lim_{\samplesize \to \infty}\samplesize \Cov( \parmLMLE{{\samplesize}}),
\]
where
$\vinfoOpt %
\defas \EE\cbra{\grad_\parm\loglik(\trueparm; \data)\otimes \grad_\parm\loglik(\trueparm; \data) }$
and
$\finfoOpt  %
 \defas -\EE\cbra[0]{\hess_\parm\loglik(\trueparm; \data)}$
are, respectively, the first- and second-order Fisher information matrices.
Such results are, however, very sensitive to the choice of step size schedule, leading to impractically slow convergence rates \citep{bach2011non-asymptotic,toulis2021proximal}.

For sampling, the canonical approach is stochastic gradient Langevin dynamics \citep[SGLD;\ ][]{welling2011bayesian},
with one-step update 
\[\label{eq:sgld-iter-update-simple}
\parmiter{\samplesize}{\iternum+1}
		 = \textstyle\parmiter{\samplesize}{\iternum}
			+ \frac{\stepsize\upper{\samplesize}_{\iternum}}{2}\stochgrad{\samplesize}{\iternum}
			+ \sqrt{\frac{\stepsize\upper{\samplesize}_{\iternum}}{\invtemp}}\ \innoviter{\iternum},
\]
where $\innoviter{\iternum} \distas \normaldist_{\parmdim}(0,I)$ is independent standard Gaussian noise
and $\invtemp \in (0, \infty]$ is the inverse temperature, which is usually taken to be
$\samplesize$.\footnote{We take $\invtemp^{-1}$ to mean $0$ when $\invtemp = +\infty$, in which case we recover SGD from \cref{eq:sgd-iter-update}.}
The benefits of introducing stochastic gradients into an MCMC procedure are less clear than for optimization since retaining exactness
would require an accept/reject step using the full-sample likelihood in the Metropolis--Hastings adjustment.
While SGLD can be asymptotically exact when run with a decreasing step size,
the optimal choice of step sizes results in a slow $\iternum^{-1/3}$ convergence rate \citep{teh2016consistency,vollmer2016exploration}.
Further, these results do not directly guarantee finite-time accuracy \citep{brosse2018sgld}.
Despite these limitations, variants of SGLD has been an active area of methods development and seen adoption in practice \citep{Ahn:2012,Chen:2014,Ma:2015:recipe,Baker.2019,Nemeth:2021:SGLD}. 

\section{Stochastic gradient algorithms and their scaling limits} \label{sec:main-results}

In this section we develop a comprehensive framework that accurately predicts the large-sample behaviour of stochastic gradient algorithms with fixed step sizes
for inference and parameter estimation, including in cases where the model is misspecified.
We develop our methods and theory in the framework of a \emph{stochastic gradient meta-algorithm} with one-step update
\[\label{eq:sgld-iter-update}
	\parmiter{\samplesize}{\iternum+1}
		& = \parmiter{\samplesize}{\iternum}
			+ \frac{\stepsize\upper{\samplesize} \precon}{2} \stochgrad{\samplesize}{\iternum}
			+ \sqrt{\frac{\stepsize\upper{\samplesize}\aniso}{\invtemp\upper{\samplesize}}}\ \innoviter{\iternum},
\]
where
$\precon \in \Reals^{\parmdim \times \parmdim}$ is the (not necessarily positive semi-definite) \emph{gradient preconditioner},
$\aniso \in \Reals^{\parmdim \times \parmdim}$ is the positive semi-definite \emph{diffusion anisotropy matrix},
$\innoviter{\iternum}$ are \IID{} $\normaldist_d(0,I_d)$, and
$\stochgrad{\samplesize}{\iternum}$ implicitly depends on the batch size $\batchsize\upper{\samplesize}$ (which in turn may vary with the sample size $\samplesize$).
Unless otherwise noted take the parameter space $\parmspace = \Reals^\parmdim$.
The meta-algorithm subsumes the SGD and SGLD algorithms described in \cref{sec:sgas}.
It also includes momentum-based methods;
see \cref{sec:example-underdamped} for details in the case of the underdamped stochastic Langevin dynamics.

\subsection{Scaling limit of the stochastic gradient meta-algorithm}

We now characterize the behaviour of the sample path of the iterates of \cref{eq:sgld-iter-update} in the region about $\parmLMLE{\samplesize}$, which
will enable us to determine the limiting distribution of the iterate average (for optimization), the asymptotic stationary distribution of
the iterates (for optimization and sampling), and the mixing speed (for sampling).
Our approach is to obtain a functional central limit theorem by taking the scaling limit of the
piecewise-constant, continuous-time process
\[  \label{eq:local-scaled-process}
\localparm\upper{\samplesize}_{t} \defas \localscale\upper{\samplesize}\rbra{\parmiter{\samplesize}{\floor{\timescale\upper{\samplesize} t}} - \parmLMLE{{\samplesize}}},
\]
where $\localscale\upper{\samplesize} \to \infty$  determines the spatial scaling
and $\timescale\upper{\samplesize} \to \infty$ determines the temporal scaling.
Since it suffices for practical application, we assume polynomial scaling of all tuning parameters as a function of sample size:
$\stepsize\upper{\samplesize} = \const{\stepsize}\samplesize^{-\stepsizePow}$ for $\stepsizePow > 0$,
$\batchsize\upper{\samplesize} = \floor{\const{b}\samplesize^\batchsizePow}$ for $\batchsizePow \ge 0$,
and $\invtemp\upper{\samplesize} = \const{\invtemp}\samplesize^{\invtempPow}$ for $\invtempPow \in \Reals$.
Given these tuning parameters, in order to have a stable and non-trivial\footnote{By non-trivial here, we mean that the limiting SDE should have both non-zero drift and non-zero diffusion terms if possible.} limit, we must take the time scaling to be $\timescale\upper{\samplesize} = \samplesize^{\stepsizePow}$ and the spatial scaling to be $\localscale\upper{\samplesize} = \samplesize^{\localscalePow}$ for $\localscalePow = \min\cbra{\batchsizePow + \stepsizePow, \invtempPow}/2$. In this setting we have the following result, under \crefrange{ass:basic-reg}{ass:V} discussed in \cref{sec:assumptions}.
All proofs are deferred to the Supplementary Materials. 

\begin{theorem}[Scaling limit of the meta-algorithm]
    \label{thm:sgld-scaling-limit-multiple-scalings-boundary-simple}
    \label{THM:SGLD-SCALING-LIMIT-MULTIPLE-SCALINGS-BOUNDARY-SIMPLE}

    If \crefrange{ass:basic-reg}{ass:V} hold, there exists $\trueparm \in \parmspace$ such that $\parmLMLE{{\samplesize}} \stk\to{p} \trueparm$,
    and there exists $\localparm_{0} \in \probmeasures{\parmspace}$ such that $\localparm\upper{\samplesize}_{0} \distto \localparm_{0}$, then
        $
            (\localparm\upper{\samplesize}_{t})_{t\in\PosReals} \distto (\localparm_{t})_{t\in\PosReals}
        $
        in the Skorohod topology\footnote{See \cref{sec:semigroups-etc} for further discussion.} in probability,
        where
        $(\localparm_{t})_{t\in\Reals}$ is an Ornstein--Uhlenbeck process given by
        \[ \label{eq:generic-OU}
        \dee \localparm_{t}
        		&= -\frac{1}{2} \driftMat \localparm_{t} \,\dee t + \vphantom{\frac{1}{2}}\sqrt{\diffusionMat} \, \dee\weiner_{t},
\]
	with $\weiner_{t}$ a $\parmdim$-dimensional standard Brownian motion,
	$\driftMat \defas \const{\stepsize}\precon \finfoOpt$ the drift matrix,
	$\diffusionMat \defas \ind{\batchsizePow + \stepsizePow \le \invtempPow}\frac{\const{\stepsize}^2\overline{\const{\batchsize}}}{4\const{\batchsize}} \precon \vinfoOpt 	\precon\transpose + \ind{\invtempPow \le \batchsizePow + \stepsizePow}\frac{\const{\stepsize}}{\const{\invtemp}} \aniso$ the positive semi-definite diffusion matrix,
	and $\overline{\const{\batchsize}} \defas 1 - \const{\batchsize}\ind{\batchsizePow = 1 \andT \text{``no replacement''}}$ the batch constant.
\end{theorem}

\add{
\begin{rmk}[Assumptions]
\Crefrange{ass:basic-reg}{ass:V} are quite weak and notably do not require convexity or bounded gradients. 
We require that the sequence of empirical critical points of the log-likelihood converges to a critical point of the expected log-likelihood. 
The critical point does not even need to be a minimizer, though in the case of a limiting critical point where the hessian is not positive definite,
the paths of the process will move away from the critical point instead of towards it (hence the need for the Hurwitz condition for existence of the stationary distribution below).
Further, we do not require a specific rate of convergence for the empirical critical point to the limiting one, 
but there is a trade-off determined by our proof strategy between this rate and the number of moments we must assume exist for various derivatives of the likelihood.
\end{rmk}
}

\begin{rmk}[Effects of stochastic gradient noise]
As expected, the mini-batch noise contributes in the large-sample regime when $\stepsizePow+\batchsizePow\leq\invtempPow$.
This exactly corresponds to when the mini-batch noise in a single step is on the same order ($=$) or dominates ($<$) the noise from the Gaussian innovations, $\innoviter{\iternum}$.
We can interpret the phase transition as occurring because the variance of the mini-batch gradient scales as $\samplesize^{-2\stepsizePow-\batchsizePow}$ while the variance of update due to the Gaussian innovations scale as $\samplesize^{-\stepsizePow-\invtempPow}$.
The spatial scaling is chosen as $\localscalePow = \min\cbra{\batchsizePow + \stepsizePow, \invtempPow}/2$ to ensure that at least one of (a) the mini-batch noise or (b) the Gaussian innovations contribute to the limit, as otherwise the limit would be a gradient flow instead of Ornstein--Uhlenbeck process, and hence fail to capture the asymptotically dominant local stochastic behaviour around $\parmLMLE{\samplesize}$.
\end{rmk}

\begin{rmk}[SGLD with control variates]
	Modifications to SGLD that include control variates can be analyzed using similar techniques. 
	These methods seek to reduce the variance of stochastic gradients using a control variate. 
	In \cref{sec:control-variates,apx:sketch-cv} we examine the SGLD-FP algorithm \citep{Baker.2019,nagapetyan2017true}, 
	where the control variate is given by the random gradient function evaluated at (the current estimate of) the MLE. 
	Formally, in an idealized setting where the MLE is known, it modifies the meta-algorithm by replacing \cref{eq:stochastic-gradient-formula} with
	\[ \label{eq:stochastic-gradient-formula-SGLDFP}
	\stochgrad{\samplesize}{\iternum}
	\defas \frac{1}{\samplesize}\grad\reg\rbra{\parmiter{\samplesize}{\iternum}} + \frac{1}{\batchsize\upper{\samplesize}}\sum_{\mbidx=1}^{\batchsize\upper{\samplesize}} \cbra{\grad\loglik\rbra[3]{ \parmiter{\samplesize}{\iternum};\ \data_{\mbdataidx{\samplesize}{\iternum}{\mbidx}}}
		-
		\grad\loglik\rbra[3]{ \parmLMLE{\samplesize};\ \data_{\mbdataidx{\samplesize}{\iternum}{\mbidx}}}}.
	\]
	We find that, when a non-trivial amount of additional Gaussian noise is included ($\invtemp>0$), the use of control variates is sufficient to reduce the variance in minibatch gradients so much that at all non-trivial scalings the gradient noise is $0$.
	Hence, \cref{thm:sgld-scaling-limit-multiple-scalings-boundary-simple} holds for SGLD-FP except with	$\localscalePow = \invtempPow/2$, and
	$\diffusionMat = \frac{\const{\stepsize}}{\const{\invtemp}} \aniso$. 
\end{rmk}

Based on \cref{thm:sgld-scaling-limit-multiple-scalings-boundary-simple}, we can establish the following corollaries which we will further leverage to explain the empirical behaviour of stochastic gradient methods and to make recommendations for how these methods could be best tuned.
First, we have a characterization of the marginal and (when it exists) the stationary covariance
of the limiting process, including conditions under which simplified forms are possible.
A square matrix $M$ is said to be \emph{Hurwitz} (or \emph{stable}) if every eigenvalue of $M$ has negative real part.

\begin{corollary}[Marginal and stationary covariances] \label{cor:sgld-scaling-limit-properties}
In the setting of \cref{thm:sgld-scaling-limit-multiple-scalings-boundary-simple}, the following hold:
\begin{enumerate}
\item For any initial parameter $\localparm_{0}$, at time $t$ the marginal distribution is
$
	 \cLaw{\localparm_{0}}(\localparm_{t}) = \normaldist_{\parmdim}\rbra[1]{e^{-s\driftMat/2} \limparm_0,\margcov{t}}
$,
with
    \*[
         \margcov{t}
         		&\defas \Cov(\localparm_{t}\vert \localparm_{0})
		= \int_0^t e^{-s\driftMat/2} \diffusionMat e^{-s\driftMat\transpose/2} \dee s.
    \]
\item If $-\precon\finfoOpt$ is Hurwitz, then $\statcov \defas \lim_{t\to\infty}\margcov{t}$ exists
and the stationary distribution of $(\localparm_{t})_{t\in\Reals}$ is $\statmeas \defas \normaldist_{\parmdim}(0,\statcov)$. In this case, $\statcov$ solves the equation
\[\label{eq:stat-cov-implicit}
	\frac{1}{2}\driftMat \statcov  + \frac{1}{2}\statcov \driftMat \transpose = \diffusionMat.
\]
\end{enumerate}
\end{corollary}

Let $\statmeas\upper{\samplesize}$ denote the stationary measure of the stochastic gradient algorithm when the sample size is $\samplesize$, if it exists.
The previous corollary leads to conditions for a Bernstein--von Mises-type result for these stationary measures. 

\begin{corollary}[Bernstein--von Mises-type theorem] \label{cor:sgld-bvm}\label{COR:SGLD-BVM}
    In the setting of \cref{thm:sgld-scaling-limit-multiple-scalings-boundary-simple},
    if  $-\precon\finfoOpt$ is Hurwitz and the collection $\{\statmeas\upper{\samplesize}\}_{\samplesize\in\Nats}$
    is uniformly tight, 
    the stationary-distributed parameters $\parm\upper{\samplesize} \distas \statmeas\upper{\samplesize}$ satisfy
$\samplesize^{\localscalePow}(\parm\upper{\samplesize} - \parmLMLE{{\samplesize}}) \distto \normaldist_{\parmdim}(0, \statcov)$ in probability.
\end{corollary}

We can interpret Corollary \ref{cor:sgld-bvm} as saying that if there is a subsequence of the stationary measures where no probability mass ``escapes to infinity'' along that subsequence, then that subsequence converges weakly to the stationary distribution of the limiting process.

\subsection{Discussion of assumptions} \label{sec:assumptions}

\crefrange{ass:basic-reg}{ass:V} are fairly mild. %
\cref{ass:basic-reg} requires that the likelihood has a minimal number of continuous derivatives, and that the regularizer is smooth in the optimization theory sense
of having Lipschitz gradients.
\begin{assumption}
	\label{ass:basic-reg}
	$\grad\reg$ is $\lipschitzconst_0$-Lipschitz, and $\loglik(\cdot;x) \in \kdiffuns{2}(\parmspace)$ for each $\dumdata\in\dataspace$.
\end{assumption}
\cref{ass:grad-norm-mom} ensures that the gradient value of the log-likelihood at the limiting parameter is not too volatile via a moment condition.
\begin{assumption}\label{ass:grad-norm-mom}
	$\stepsizePow-\localscalePow - \timescalePow/3 >0$ and
	$\EE\sbra{\norm{\grad\loglik(\trueparm;\data_1)}^{\assmoment{2}}}<\infty$
	for some $\assmoment{2}> \frac{1}{\stepsizePow-\localscalePow - \timescalePow/3}$.
\end{assumption}
 \cref{ass:gll} ensures that the random likelihood functions from each data sample are sufficiently smooth via a moment condition on the random smoothness parameter.
\begin{assumption}
	\label{ass:gll}
	For some $\asspower{3}\in\cointer{0,\localscalePow}$ and $\assmoment{3} \defas \frac{1}{\stepsizePow+\asspower{3}-\localscalePow - \timescalePow/3}$,
	$
		\norm[0]{\parmLMLE{{\samplesize}} - \trueparm}  \in o_p(1 / \samplesize^{\asspower{3}}),
		 \andT \ 
		\EE \sbra[1]{ \normInf{\hess\loglik(\cdot;\data_1)}^{\assmoment{3}}}  <\infty.
	$
\end{assumption}
\cref{ass:J,ass:V} require convergence of the  first-and second-order empirical Fisher information matrices
$
	\vinfoEmp{\samplesize}(\parm)
		= \textstyle \frac{1}{\samplesize}\sum_{\dataidx\in\range{\samplesize}} [\grad\loglik(\parm;\, \data_\dataidx)]^{\otimes 2}
$
and
$
	\finfoEmp{\samplesize}(\parm)
		 = \textstyle \frac{1}{\samplesize}\sum_{\dataidx\in\range{\samplesize}} [-\hess\loglik(\parm;\, \data_\dataidx)].
$
For any $r > 0$, define the ball $\ball\upper{\samplesize}\rbra[0]{r} \defas \{ \parm \in \parmspace : \norm{\parm - \parmMLE{\samplesize}} \le r/\samplesize^\localscalePow\}$.
\begin{assumption}
	\label{ass:J}\label{ASS:J}
	There is a non-decreasing sequence $r_{\finfo,\samplesize} \overset{n \to \infty}{\longrightarrow} \infty$  such that
	$\sup_{\parm\in\ball\upper{\samplesize}\rbra[0]{r_{\finfo,\samplesize}}} \norm[1]{\finfoEmp{\samplesize}(\parm) - \finfo(\trueparm)} \stk\to{p} 0.
	$
\end{assumption}
\begin{assumption}
	\label{ass:V}\label{ASS:V}
There is a non-decreasing sequence $r_{\vinfo,\samplesize} \overset{n \to \infty}{\longrightarrow} \infty$  such that
$
\sup_{\parm\in\ball\upper{\samplesize}\rbra[1]{r_{\vinfo,\samplesize}}} \norm[1]{\vinfoEmp{\samplesize}(\parm) - \vinfo(\trueparm)} \stk\to{p} 0
.$
\end{assumption}
The assumptions all hold, for example, for generalized linear models with bounded covariates and either Lipschitz inverse-link functions, or suitably constrained parameter domains. %
Several sufficient conditions for each of \cref{ass:J,ass:V} are given in \cref{sec:assVJ-suff}.

\section{Practical implications of the scaling limit}

We now turn to assessing the implications of our statistical scaling limit on the large-sample behavior of stochastic gradient
algorithms used for optimization and sampling.

\subsection{Mixing time} \label{sec:implications-mixingtime}

Because we characterize the full-path behavior of the meta-algorithm, we can obtain insights into its mixing speed.
Let $\hat\statmeas_{\iternum}\upper{\samplesize}(f) \defas \iternum^{-1}\sum_{\iternumdum=1}^{\iternum} f(\parmiter{\samplesize}{\iternumdum})$ denote the
Monte Carlo estimate of $\statmeas\upper{\samplesize}(f)$.
We can use the \emph{mixing time} (or worst-case integrated autocorrelation time) $\iact\upper{\samplesize} \defas \sup_{f} \inf \{ \iternum : \Var_{\hat\statmeas_{\iternum}\upper{\samplesize}}(f) / \Var_{\statmeas\upper{\samplesize}}(f) \le 1 \} $ to characterize the efficiency of MCMC algorithms.
For the limiting process, define the ``Monte Carlo average'' $\hat\statmeas_{t}(f) \defas t^{-1} \int_{0}^{t} f(\localparm_{s})\,\dee s$
and the mixing time $\iact \defas \sup_{f} \inf \{ t : \Var_{\hat\statmeas_{t}}(f) / \Var_{\statmeas}(f) \le 1 \}$.
When the limiting process is reversible, standard results\footnote{Apply the spectral theorem for self-adjoint operators \citep{rudin1991functional} to the Poincar\'e inequality \citep{bakry2014analysis}} allow us to upper bound $\tau$ by the reciprocal of the spectral gap of the limiting process.
Since the spectral gap of the Ornstein--Uhlenbeck process is $\mineigval(\driftMat)/2$, where  $\mineigval(\driftMat)$ denotes its minimum eigenvalue of $\driftMat$,
we may \emph{heuristically} conclude then that the limiting mixing time is $\iact\upper{\samplesize} = 2\timescale\upper{\samplesize}/\mineigval(\driftMat)$ iterations. 
\add{This mixing time corresponds to $2\timescale\upper{\samplesize}\batchsize\upper{\samplesize}/\mineigval(\driftMat) = 2\batchsize\upper{\samplesize}/\{\stepsize\upper{\samplesize}\mineigval(\precon\finfoOpt)\}$ likelihood evaluations, or equivalently $ 2\batchsize\upper{\samplesize}/\{\stepsize\upper{\samplesize}\mineigval(\precon\finfoOpt)\}$ epochs.}
Even when the limiting process is not reversible, the spectral gap is still a useful metric for the large-time rate of mixing of the process, and is given by the same formula, while the integrated autocorrelation time becomes intractable.

The reason these arguments are heuristic is because weak converge of the processes and stationary distributions
 is insufficient to conclude that the mixing times converge.
In \cref{apx:mixing-asympt}, we provide further details and describe a possible approach to making the mixing result rigorous.
\add{We note, however, that comparison between the mixing time of a scaling limit and the mixing time of the corresponding pre-limiting processes is standard is MCMC tuning, even though it is technically only a heuristic. This is, for example, the nature of widely celebrated results in the optimal scaling literature \citep[e.g.,][]{gelman1997weak, roberts2001optimal}. 
Thus, as a practical matter, a user with a dataset of size $\samplesize$ can conclude that using a step size $\stepsize$ and batch size $\batchsize$, will result in a mixing time of roughly
\[\label{eq:mixing-equation}
\frac{2\batchsize}{\stepsize\,\mineigval\{\precon\finfoEmp{\samplesize}(\parmLMLE{{\samplesize}})\}}
\]
epochs, thereby providing a valuable constraint when tuning $\batchsize$, $\stepsize$, and $\precon$.} 
\add{Some example tuning parameter combinations that lead to limiting stationary distributions of interest, and the corresponding mixing times of the limit process, are given in \cref{tab:tuning-combos}.}

\subsection{Optimization}
\label{sec:optimization-iteravg}
The key implication of our results for optimization concern the average $\smash{\iteravg{\samplesize}{\iternum} = \frac{1}{\iternum}\sum_{j=1}^\iternum \parmiter{\samplesize}{j}}$ of the first $\iternum$ iterations of the algorithm.
The accuracy of the iterate average is characterized by its covariance $\avgcov{\iternum}\upper{\samplesize} \defas \Cov(\iteravg{\samplesize}{\iternum})$.
We can approximate $\avgcov{\iternum}\upper{\samplesize}$ in terms of the covariance of the
averaged limiting process, which is defined as $\limparmavg_{t} \defas t^{-1} \int_{0}^{t} \limparm_{s} \, \dee s$.
The following result is similar in spirit to Theorem 2.1 of \citet{Kushner:1993}.

\begin{proposition}[Path averaging] \label{prop:averaged-covariance}\label{PROP:AVERAGED-COVARIANCE}
For $(\localparm_{t})_{t\in\PosReals}$ defined by \cref{eq:generic-OU}, assuming $-\driftMat$ is Hurwitz and $\localparm_{0}\distas \normaldist(0,\statcov)$,
the covariance of the averaged limiting process is
\begin{align}
\avgcov{t}
\defas \Cov\rbra{\limparmavg_{t}}
	&=  \frac{4}{t}\driftMat^{-1}\diffusionMat\driftMat^{-\top} - \frac{8}{t^{2}}\Sym\rbra{\driftMat^{-2}\cbra{I - e^{-t \driftMat/2}}\statcov} \label{eq:continuous-time-IA} \\
	&=
	\begin{cases}
		\statcov - \frac{t}{6}\diffusionMat + O(t^{2}) & \text{if } t \ll 7\norm{\driftMat}^2\norm{\driftMat^{-2}\statcov}^{1/2} \\
		\frac{4}{t}\driftMat^{-1}\diffusionMat\driftMat^{-\top} + O(t^{-2}) & \text{if }  t \gg 3\norm{\driftMat^{-2}\statcov}^{1/2}.
	\end{cases} \label{eq:continuous-time-IA-estimates}
\end{align}
If either (i) $\batchsizePow + \stepsizePow < \invtempPow$ or (ii) $\batchsizePow + \stepsizePow = \invtempPow$ and $\const{\invtemp} = +\infty$, then
$
	\frac{4}{t}\driftMat^{-1}\diffusionMat\driftMat^{-\top}
		 = \frac{\overline{\const{\batchsize}}}{t\const{\batchsize}} \finfoOpt^{-1} \vinfoOpt \finfoOpt^{-1}.
$
\end{proposition}
The proof of this result is in \cref{proof:prop:averaged-covariance}.
Using Proposition \ref{prop:averaged-covariance}, we can characterize large-sample behaviour of $\iteravg{\samplesize}{\iternum}$ for $\iternum = \iternum\upper{\samplesize} \defas \floor{\numpasses \timescale\upper{\samplesize} /  \const{\batchsize}} = \floor{\numpasses \samplesize^{\stepsizePow}  /  \const{\batchsize}}$,
which corresponds to making $\numpasses$ passes over the dataset.

\begin{corollary}[Bernstein--von Mises-type theorem for iterate averaging]\label{cor:iteravg}\label{COR:ITERAVG}
Suppose \crefrange{ass:basic-reg}{ass:V} all hold.
If $\batchsizePow + \stepsizePow \le \invtempPow \le 1$
and $\Law(\localparm\upper{\samplesize}_{0}) \distto \normaldist(0,\statcov)$,
then $\samplesize^{\batchsizePow + \stepsizePow}(\iteravg{\samplesize}{\iternum\upper{\samplesize}} - \parmLMLE{{\samplesize}})$ converges in distribution to a zero-mean Gaussian and 
\[\label{eq:asymptotic-iteravg-cov}
\samplesize^{\batchsizePow + \stepsizePow} \Cov\rbra{\iteravg{\samplesize}{\iternum\upper{\samplesize}}}
	\to  \frac{4 \const{\batchsize}}{ \const{\stepsize}\numpasses}\Sym\rbra{{\rbra{\precon \finfoOpt}}^{-1}\statcov}
	- \frac{8 \const{\batchsize}^{2}}{\const{\stepsize}^{2}\numpasses^{2}}\Sym\rbra{\rbra{\precon \finfoOpt}^{-2}\cbra{I - e^{-\numpasses \precon \finfoOpt/\rbra{2\const{\batchsize}}}}\statcov}
\]
in probability. 
If in addition $\batchsizePow + \stepsizePow = 1$ and $\const{\invtemp} = +\infty$, then %
\[\label{eq:asymptotic-iteravg-cov-simple}
\samplesize \Cov\rbra{\iteravg{\samplesize}{\iternum\upper{\samplesize}}}
 	\to \frac{1}{\numpasses} \finfoOpt^{-1} \vinfoOpt \finfoOpt^{-1} + R(\numpasses) \text{ in probability},
\qquad\text{where  $\norm{R(\numpasses)} \le  \frac{8\const{\batchsize}^{2}}{\const{\stepsize}^{2}\numpasses^{2}} \norm{\rbra{\precon \finfoOpt}^{-2}\statcov}$.}
\]
\end{corollary}
It follows from \cref{eq:asymptotic-iteravg-cov-simple,eq:iterate-averaging-optimality} that for $\numpasses$ sufficiently large, iterate averaging with potentially large step size
of order $n^{-\stepsizePow}$ (\edit{with $\stepsizePow \le 1$}{``large'' constant-in-time step sizes here means that $\stepsizePow \ll 1$, while the results apply for any $\stepsizePow \leq 1$}) and batch size of order $n^{1-\stepsizePow}$ has numerical error
$\Cov\rbra[1]{\iteravg{\samplesize}{\iternum\upper{\samplesize}}}  \approx \frac{1}{\numpasses} \Cov( \parmLMLE{{\samplesize}})$
for estimation of $\parmLMLE{{\samplesize}}$.
This error is optimal, in the sense that after $\numpasses \gg 1$ passes over the dataset,
it is small compared to the statistical error $\Cov( \parmLMLE{{\samplesize}})$.
It is instructive to consider two idealized cases: 
\begin{enumerate}
\item Take $\precon = \aniso = I$ and assume that $\finfoOpt$ and $\vinfoOpt$ commute. The bound on the remainder term simplifies to
$\norm{R(\numpasses)} \le  \frac{2\const{\batchsize}}{\const{\stepsize}\numpasses^{2}} \norm{\finfoOpt^{-1}}$.
Hence, we should only expect the remainder term to be small when $\numpasses^{2} \gg \frac{2\const{\batchsize}}{\const{\stepsize}} \norm{\finfoOpt^{-1} \vinfoOpt \finfoOpt^{-2} }$.
\item Take $\precon = \aniso = \finfoOpt^{-1}$. The bound on the remainder term simplifies to
$\norm{R(\numpasses)} \le  \frac{2\const{\batchsize}}{\const{\stepsize}\numpasses^{2}} \norm{\finfoOpt^{-1} \vinfoOpt \finfoOpt^{-1} }$.
Hence, we should only expect the remainder term to be small when $\numpasses^{2} \gg \frac{2\const{\batchsize}}{\const{\stepsize}}$.
\end{enumerate}
In either case, a large step size constant $\const{\stepsize}$ relative to the batch size constant $\const{\batchsize}$ leads to the remainder term being small even for small $\numpasses$.
However, particularly without preconditioning, this regime may lead to numerical instability.

\subsection{Sampling}
\label{sec:sampling}
\textbf{Sampling from the posterior.}
The Bernstein-von Mises theorem states that the posterior-distributed parameter $\parm\upper{\samplesize} \distas \posterior{{\samplesize}}$ satisfies
$\samplesize^{1/2}(\parm\upper{\samplesize} - \parmLMLE{{\samplesize}}) \distto \normaldist_{\parmdim}\rbra[1]{0, \finfoOpt^{-1}}$ in probability.
In order for the large-sample stationary distribution of \cref{eq:sgld-iter-update} to match the Bernstein--von Mises limit of the posterior, we must first enforce that $\localscalePow= 1/2$. \add{Then, there are several ways to ensure that the limiting process has the same distribution as the limiting, six of which using various forms of SGD, SGLD, and SGLD-FP are shown in \cref{tab:tuning-combos}.}

In terms of the number of gradient queries per unit \add{mixing} time \delete{in the scaling limit scales, options \ref{item:2}\ref{item:2b} and \ref{item:3}}\add{the cases where $\stepsize\not\in o(b/n)$} are the most efficient as the query-count scales linearly
with the dataset size (since $\stepsizePow+\batchsizePow=1$), while for \delete{options \ref{item:1} and \ref{item:2}\ref{item:2a}}\add{the cases where $\stepsize\in o(b/n)$} it scales super-linearly ($\stepsizePow+\batchsizePow>1$).
In practice, options \delete{\ref{item:2}\ref{item:2b} and \ref{item:3} both}\add{involving preconditioning matrices ($\Gamma\neq I$) or control variates (SGLD-FP)} first require an estimate of $\parmLMLE{{\samplesize}}$ to, respectively,
construct \add{estimates of }the preconditioner $\precon = \finfoEmp{\samplesize}(\parmLMLE{{\samplesize}})^{-1} \approx \finfoOpt^{-1}$
\add{and/}or construct the control variates $\grad\loglik(\parmLMLE{{\samplesize}};\data_i)$.
The latter option is more appealing, particularly if $\parmdim$ is large, as no matrix inversion or per-iteration multiplication is required.
In either case, however, preconditioning with $\finfoEmp{\samplesize}(\parmLMLE{{\samplesize}})^{-1}$ will minimize the mixing time.
\add{Methods for estimating $\finfoOpt,\ \vinfoOpt$, their inverses, and sparse approximations have been explored extensively in other work \citep[e.g.,][]{haario2001,ahn2012,mandt2017stochastic,pollock2020quasi,chen2020statistical}. }

\textbf{Alternative uncertainty quantification.}
\add{When the model is misspecified or when generalized Baye\-sian inference based on a loss function is used \citep{bissiri2016general}, the (generalized) posterior distribution may provide less-than-robust uncertainty quantification 
because the (local) M-estimator $\parmLMLE{\samplesize}$ is itself asymptotically normal, centered at the true parameter $\trueparm$, with covariance equal to the ``sandwich'' covariance matrix, $\finfoOpt^{-1}\vinfoOpt\finfoOpt^{-1}$ \citep{kleijn2012bernstein,muller2013risk}.
If the model is well-specified (i.e., $\datadist = \modeldist{\parm}$ for some $\parm\in\parmspace$), then $\finfoOpt=\vinfoOpt$, and so $\finfoOpt^{-1}\vinfoOpt\finfoOpt^{-1}=\finfoOpt^{-1}$. However, if the model is misspecified (i.e., $\datadist \ne \modeldist{\parm}$ for any $\parm\in\parmspace$), then the sandwich may differ from $\finfoOpt^{-1}$ \citep{huber1967,white1982maximum}.
In this case, posterior credible sets are not asymptotically well-calibrated frequentist confidence sets \citep{kleijn2012bernstein,muller2013risk}
and predictions (or other decision-theoretic quantities) can become unstable \citep{huggins2019baggedposterior,huggins2023modelselection}.}

The question of how to account for misspecification in the Bayesian setting has been addressed in a number of ways \citep[e.g.,][]{royall2003interpreting,muller2013risk,stafford1996robust,grunwald2017inconsistency,huggins2019baggedposterior}.
For example, we may want to match the sandwich covariance, as prescribed by \citet{muller2013risk}, which by definition is robust to model misspecification in a frequentist sense.
Or we may want to combine the sandwich and Bernstein--von Mises covariances, as in the bagged posterior \citep{huggins2019baggedposterior}.
Either of these desiderata can be obtained by setting $\precon = \aniso = \finfoOpt^{-1}$,
and any valid $\stepsizePow+\batchsizePow=1=\invtempPow$.
With this tuning, for any $w_1, w_2>0$, taking $\const{\stepsize} = 4w_1\const{\batchsize}$
and $\const{\invtemp} = w_2^{-1}$, gives $\statcov = w_1\finfoOpt^{-1}\vinfoOpt\finfoOpt^{-1} +  w_2\finfoOpt^{-1}$.
This matches the asymptotic distribution of the bagged posterior with re-sampling rate $w_1$ when $w_1 = w_2$ \citep{huggins2019baggedposterior}. \add{This is summarized in the ``Bagged Posterior'' row of \cref{tab:tuning-combos}.}
Moreover, we can obtain any convex combinations of the uncertainty quantification from the posterior and the asymptotics of the M-estimator by taking $w_1+w_2=1$.
This enables interpolation between frequentist-like and Bayesian-like forms of inference and results in a mixing time of $1/w_{1}$ epochs.
Hence in principle we can use SGD (by setting $w_{1} = 1$, $w_{2} = 0$, and $\const{\invtemp}=+\infty$) to obtain the sandwich covariance and minimize the mixing
time to be a single pass over the dataset. \add{This can be interpreted as an asymptotic local fiducial distribution for the parameter, as it has credible sets which match frequentist confidence sets asymptotically. This tuning is summarized in the ``Local Fiducial'' row of \cref{tab:tuning-combos}.}

\add{
\begin{table}[t]
\begin{tabular}{llllllll}
Target           & Target Asymp. Cov.                                            & Algo. & $\precon$        & $\aniso$         & $\invtemp$                     & $\stepsize$                                                               & Mix. Time (Epochs)                                                                \\ \hline
Posterior        & $\finfoOpt^{-1}$                                              & \smaller SGD   & $\vinfoOpt^{-1}$ & n.a.             & n.a.                           & $4\batchsize / \samplesize$                                               & $\lambda_{\min{}}^{-1}(\vinfoOpt^{-1}\finfoOpt)$                                      \\
Posterior        & $\finfoOpt^{-1}$                                              & \smaller SGLD  & $\finfoOpt^{-1}$ & $\finfoOpt^{-1}$ & $\const{\invtemp} \samplesize$ & $\frac{4 \batchsize (1-\const{\invtemp})} {\samplesize \const{\invtemp}}$ & $\frac{\const{\invtemp}}{(1-\const{\invtemp})}$                                       \\
Posterior        & $\finfoOpt^{-1}$                                              & \smaller SGLD  & $\finfoOpt^{-1}$ & $\finfoOpt^{-1}$ & $\samplesize$ & $o(\batchsize / \samplesize)$                                             & $\frac{\batchsize}{\samplesize \stepsize} \in \omega(1) $                             \\
Posterior        & $\finfoOpt^{-1}$                                              & \smaller SGLD  & $I$              & $I$              & $\samplesize$ & $o(\batchsize / \samplesize)$                                             & $\frac{\batchsize\lambda_{\min{}}^{-1}(\finfoOpt)}{\samplesize \stepsize} \in \omega(1) $ \\

Posterior        & $\finfoOpt^{-1}$                                              & \smaller SGLD-FP  & $I$              & $I$              & $\samplesize$ & $4\batchsize / \samplesize$                                              & $\lambda_{\min{}}^{-1}(\finfoOpt)$ \\
Posterior        & $\finfoOpt^{-1}$                                              & \smaller SGLD-FP  & $\finfoOpt^{-1}$             & $\finfoOpt^{-1}$              & $\samplesize$ & $4\batchsize / \samplesize$                                              & $1$ \\

Bagged Posterior & $w_1 \finfoOpt^{-1}+w_2\finfoOpt^{-1}\vinfoOpt\finfoOpt^{-1}$ & \smaller SGLD  & $\finfoOpt^{-1}$ & $\finfoOpt^{-1}$ & $\frac{\samplesize}{w_2}$         & $\frac{4 w_1 \batchsize}{\samplesize}$                                          & $1/w_1$                                                                               \\
Local Fiducial   & $\finfoOpt^{-1}\vinfoOpt\finfoOpt^{-1}$                       & \smaller SGD   & $\finfoOpt^{-1}$ & n.a.             & n.a.                           &$\frac{4 \batchsize}{\samplesize}$                                          & $1$                                                                                  
\end{tabular}
\caption{\add{Tuning parameter combinations for various target distributions, and their corresponding mixing times {in epochs}. 
If the mixing time is $\omega(1)$ (in $\samplesize$), then in the limit the process does not mix in a constant number of epochs.}}
\label{tab:tuning-combos}
\end{table}

}

\section{Numerical Experiments} \label{sec:experiments}

We present results for three experiments using both simulated and real data that show our theory closely reflects finite-sample behavior. Source code for experiments is at \url{https://github.com/jnegrea/stat-infr-sgas}.

\subsection{Experiment 1: Gaussian simulation study}
\label{experiment1}

First we demonstrate the effect of model misspecification on tuning, highlighting both the sampling implications (\cref{sec:sampling,tab:tuning-combos})
and the iterate averaging behaviour relevant to optimization (\cref{sec:optimization-iteravg}). 
We choose the combination of the data-generating distribution and likelihood function specifically to ensure that $\finfoOpt \neq \vinfoOpt$, so that the effect of misspecification would be apparent.
We run SGD with no preconditioning and with preconditioning by $\finfoOpt$, and SGLD with preconditioning by $\finfoOpt$. 
For SGLD we use the inverse temperature $\invtemp\upper{\samplesize} = \samplesize$, which corresponds to the canonical choice that would be made when not using stochastic gradients. %
We also compute iterate averages for SGD with no preconditioning \add{over 1 epoch and 8 epochs}. 
Exact specifications for the experiment are in \cref{tab:exp1-2-3-settings} in \delete{\cref{apx:experiment-supp} of} the supplemental material.

We interpret our results using our scaling limit with parameters
$\localscalePow = 1/2$,
$\stepsizePow = 1$,
$\batchsizePow = 0$, \add{which}
\delete{This combination of scaling parameters} corresponds to the standard statistical local scaling and a fixed batch size.
\Cref{fig:exp1-results} shows plots for the joint density of the first and last coordinates of the parameter vector for each of the five tunings.
The density for the empirical run of the algorithms is given by a 2D kernel density estimate. The density for the predicted behaviour is given by the stationary distribution of the limiting process.
As predicted by our results \add{in \cref{sec:sampling}, specifically in the ``Local Fiducial'' row of \cref{tab:tuning-combos}}, preconditioning by $\finfoOpt$ leads to an empirical distribution for the iterates of the algorithm matching the covariance of the MLE \add{(an asymptotic locally fiducial distribution)}, 
not preconditioning leads to behaviour that matches neither (but is still predictable using our results), 
and preconditioning by $\finfoOpt$ for SGLD leads to an empirical distribution for the iterates of the algorithm matching the asymptotics of a {bagged posterior}, which is given by a linear combination of the covariance of the MLE and the covariance of the posterior. 
\add{Furthermore, as predicted by the results in \cref{sec:optimization-iteravg} (in particular \cref{eq:asymptotic-iteravg-cov-simple}), and in contrast to the predictions in \citet{mandt2017stochastic}, iterate averaging for a ``large'' number of epochs (8) is closely approximated by the scaled sandwich covariance, while iterate averaging over a ``small'' number of epochs (1) is not sufficient for the approximation by the scaled sandwich covariance to be accurate.} 
\add{Finally, \cref{tab:exp1-mixing-times} shows that the mixing times predicted by our theory closely match their empirical counterparts.}

\begin{figure}[t]
  \centering
  \begin{subfigure}[b]{0.3\textwidth}
    \includegraphics[width=\textwidth]{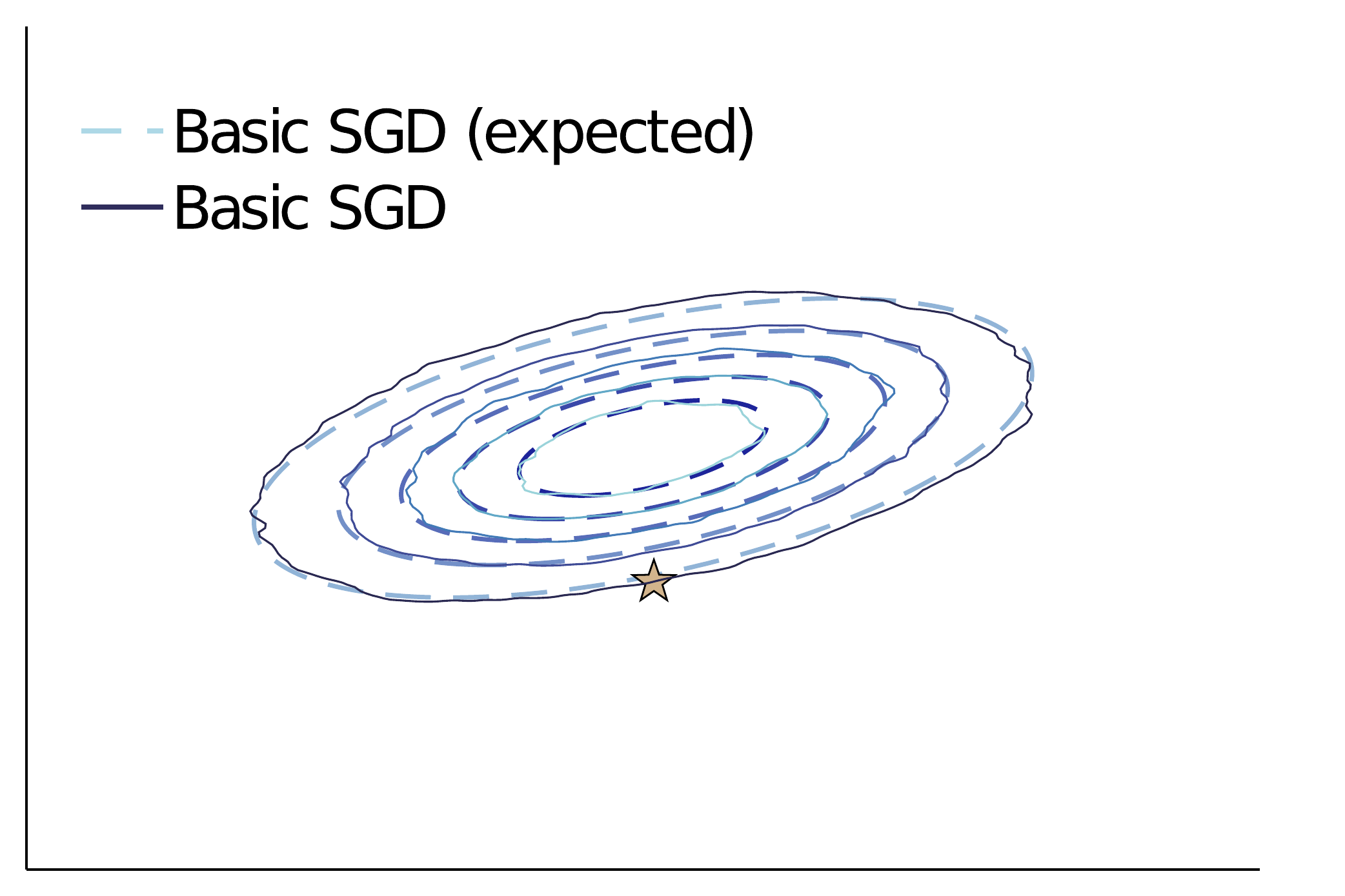}
    \caption{No preconditioning}
  \end{subfigure}
  \begin{subfigure}[b]{0.3\textwidth}
    \includegraphics[width=\textwidth]{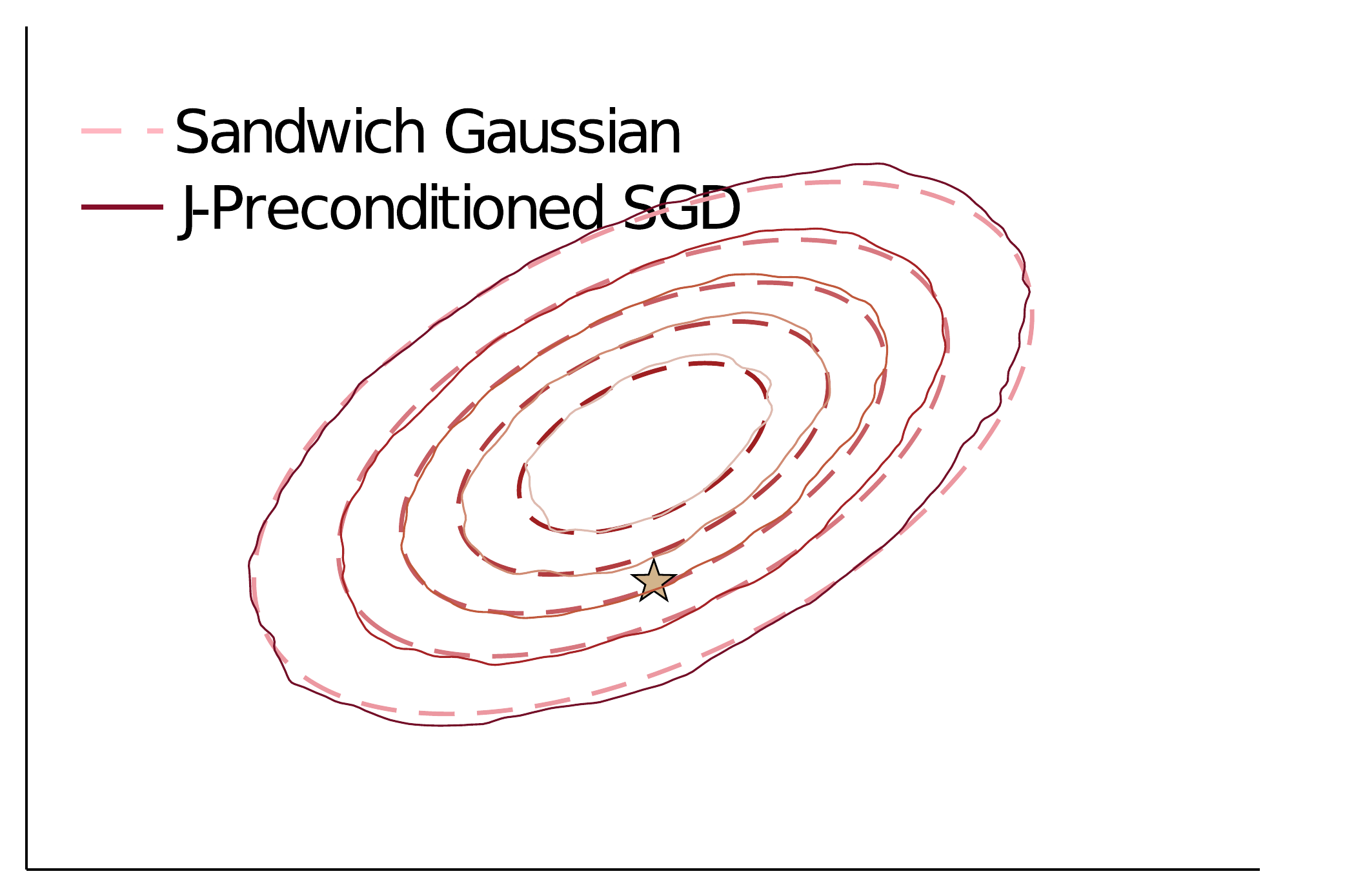}
    \caption{$\finfoOpt^{-1}$-precond. SGD}
  \end{subfigure}
  \begin{subfigure}[b]{0.3\textwidth}
    \includegraphics[width=\textwidth]{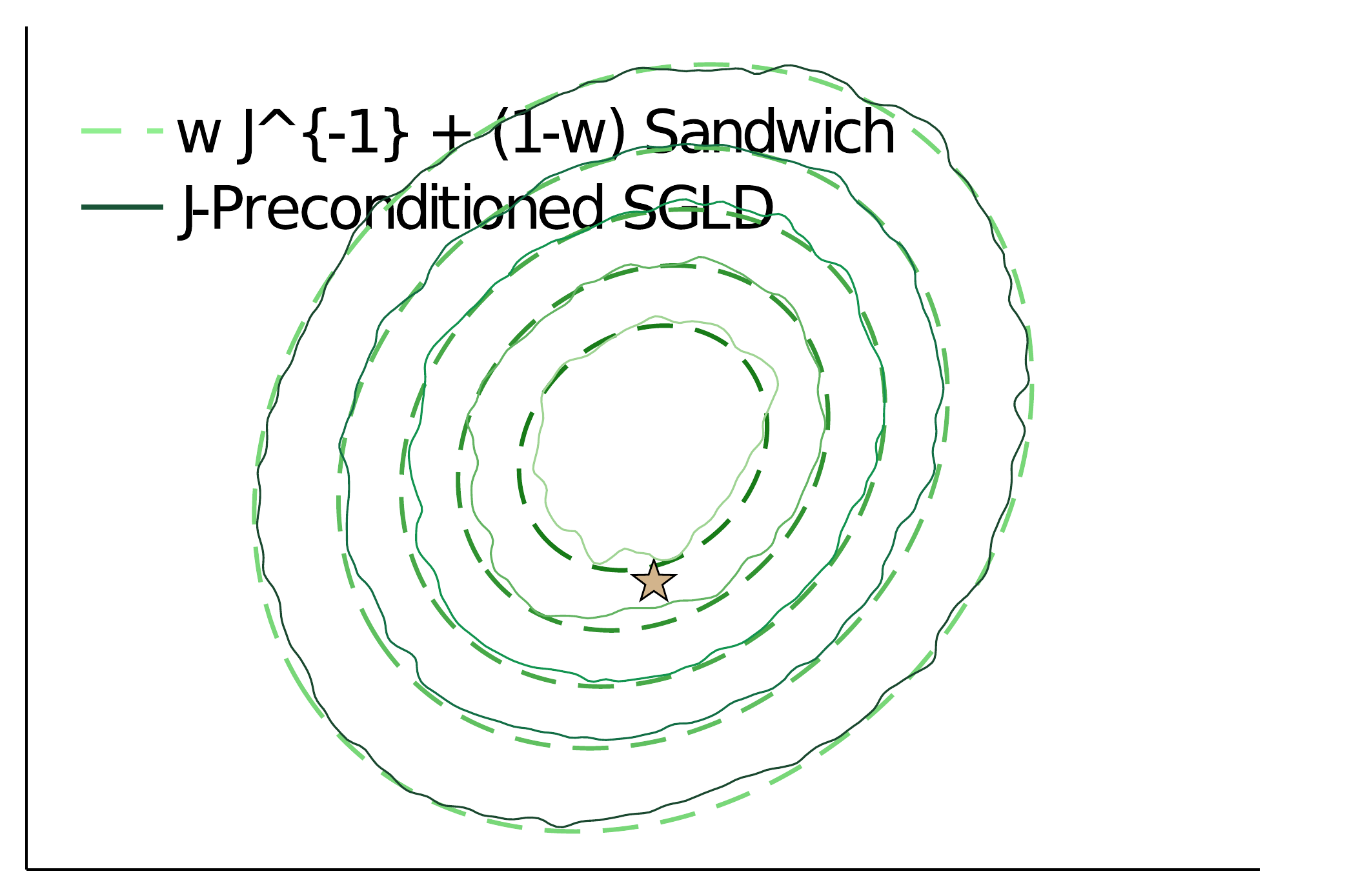}
    \caption{$\finfoOpt^{-1}$-precond. SGLD}
  \end{subfigure} \\
  \begin{subfigure}[b]{0.3\textwidth}
    \includegraphics[width=\textwidth]{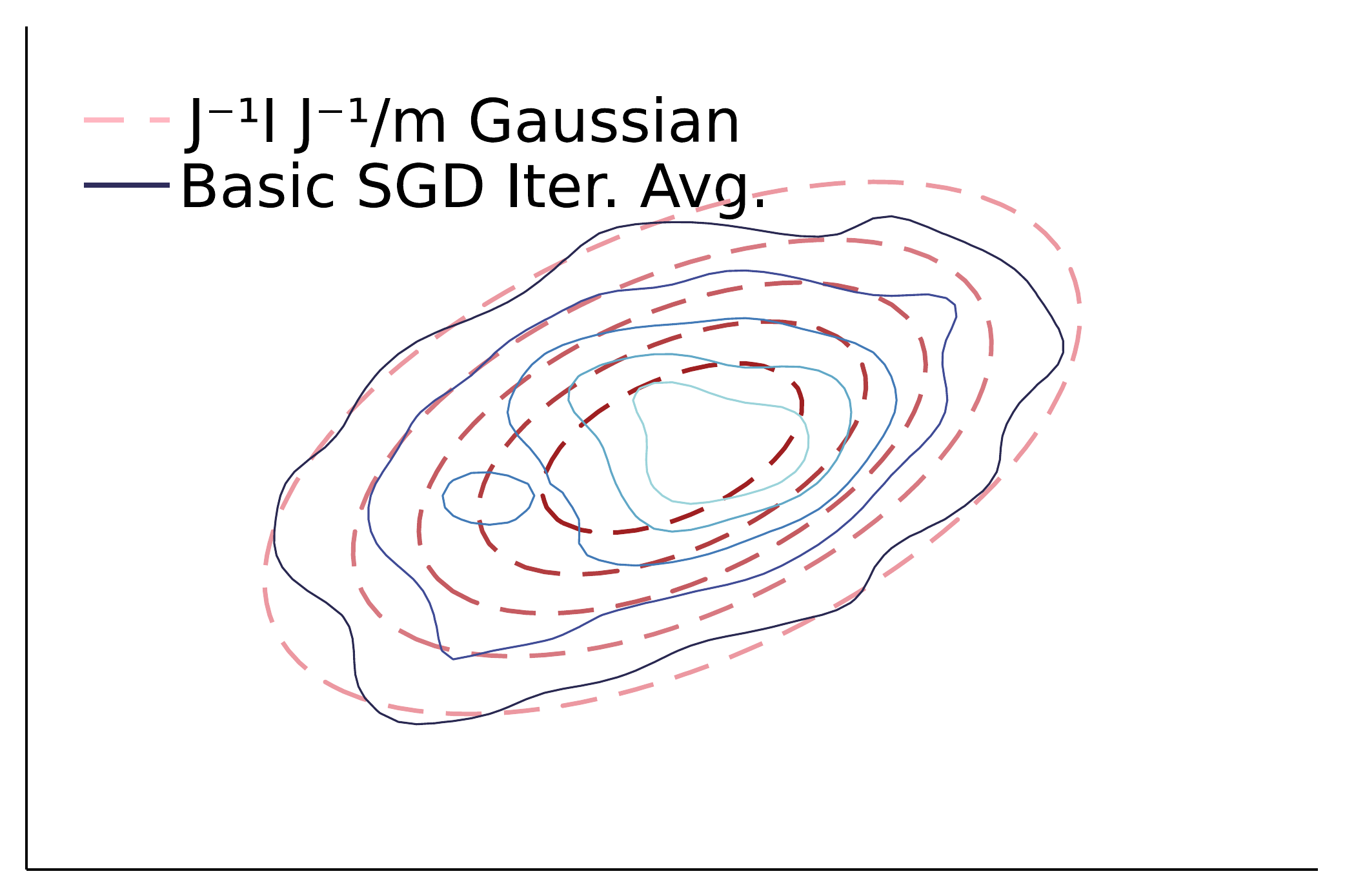}
    \caption{SGD Iter. Avg., $8$ epochs} 
  \end{subfigure} 
  \begin{subfigure}[b]{0.3\textwidth}
    \includegraphics[width=\textwidth]{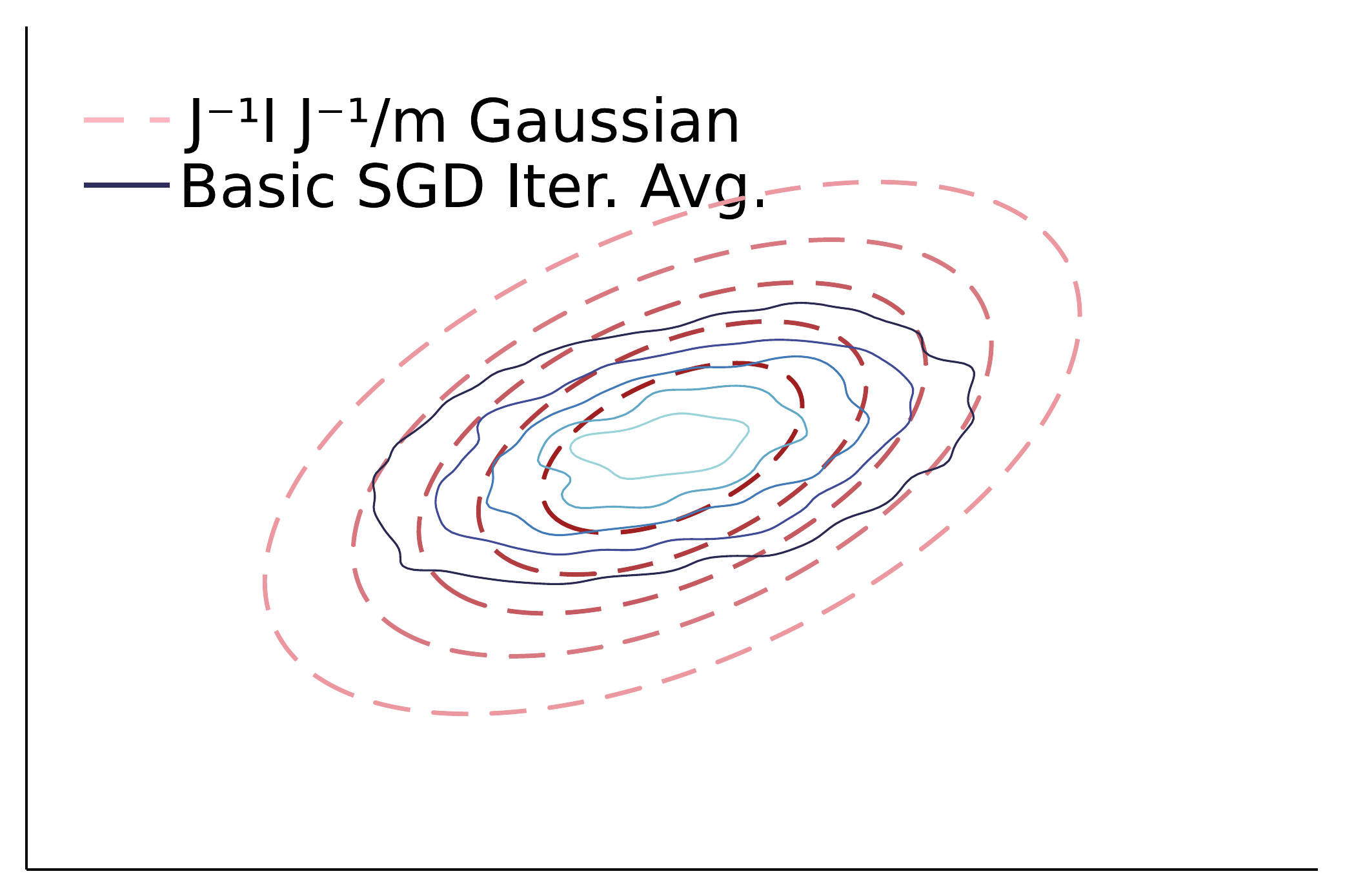}
    \caption{SGD Iter. Avg., $1$ epoch} 
  \end{subfigure} 
  \vspace{-1em}
  \caption{Results of experiment 1. The empirical results follow the theoretical predictions based on scaling limits. \delete{The KDE for the 8 epoch iterate average is more volatile than the others because its effective sample size is $\approx8$ times smaller.} For preconditioned SGLD, $w=1/2$.}
  \label{fig:exp1-results}
\end{figure}

\begin{table}[b]
  \centering
  \begin{tabular}{l|ll|ll|ll}
    &\multicolumn{2}{c|}{Experiment 1} &\multicolumn{2}{c|}{Experiment 2} &\multicolumn{2}{c}{Experiment 3} \\
    Method & Emp. & Pred. & Emp. & Pred. & Emp. & Pred. \\ \hline
      SGD, no preconditioning
      & 3.2
      & 3.2
      & 150
      & 480
      & -
      & -\\
      $\finfoOpt^{-1}$-preconditioned SGD
      & 1.1
      & 1.0 
      & 1.2
      & 1.0
      & 1.5
      & 1.0 \\
      $\vinfoOpt^{-1}$-preconditioned SGD
      & 2.3
      & 2.8  
      & 1.0
      & 1.0
      & - 
      & - \\
      $\finfoOpt^{-1}$-preconditioned SGLD
      & 2.2
      & 2.0
      & 2.3
      & 2.0
      & 3.0
      & 2.0 
  \end{tabular}
  \caption{Comparison of empirical and predicted mixing times (in epochs) for all experiments measured by integrated autocorrelation times (IACT). The empirical value is computed numerically from the run. The predicted value is computed following \cref{tab:tuning-combos}.}
  \label{tab:exp1-mixing-times}\label{tab:exp2-mixing-times}\label{tab:exp4-mixing-times}
\end{table}

\subsection{Experiment 2: Large-scale inference for airline delay data---logistic regression}
\label{experiment2}

Next we examine the same airline dataset and model as in \citet{pollock2020quasi} using their pre-processed data \add{so we can directly compare our recommended settings to 
the results they obtained with suboptimal tuning parameters (see \cref{exa:pollock}).} 
 The responses are binary and there are 3 covariates. We use the full dataset ($\approx$ 120 million observations) to estimate the ``ground truth'' quantities $(\trueparm, \finfoOpt, \vinfoOpt)$, and we apply the stochastic gradient algorithms using \delete{as a dataset} a random subsample of size 1 million from the full dataset. 

For the results regarding the marginal distribution of the iterates, we compare SGLD without preconditioning to SGD preconditioned by $\vinfoOpt$. For this example, the matrices $\finfoOpt$ and $\vinfoOpt$ are numerically indistinguishable, and hence all three preconditioned methods we examined yield essentially identical results, and all are materially different from not preconditioning. Again, we interpret this using our scaling limit with parameters
$\localscalePow = 1/2$,
$\stepsizePow = 1$,
$\batchsizePow = 0$.
An experimental finding of \citet{pollock2020quasi} was that (non-preconditioned) SGLD had relatively poor mixing performance as compared with the ScaLE algorithm they introduce.
\Cref{fig:exp2-results,fig:exp2-results-joint,tab:exp2-mixing-times} similarly show that, without preconditioning, SGLD fails to properly quantify uncertainty in the true parameter (marginally for coordinate 4, and jointly) and mixes slowly.
Furthermore, SGLD without preconditioning mixes materially more slowly than preconditioned methods, as evidenced by the jagged histogram from its run (\cref{fig:exp2-results}) and the contour plot (\cref{fig:exp2-results-joint}).
Our numerical results also show that their findings would have been significantly different had they used the appropriate preconditioning as predicted by our results showing that preconditioning accelerates the mixing of SGLD considerably and leads to more accurate uncertainty quantification. 
\add{These findings are consistent with our theoretical developments in \cref{sec:implications-mixingtime,sec:sampling} (in particular, \cref{eq:mixing-equation}, and the preconditioned SGD and non-preconditioned SGLD and SGLD-FP rows of \cref{tab:tuning-combos}).
Thus, we can conclude that the poor relative mixing of non-preconditioned SGLD-FP observed in \citet{pollock2020quasi}, and the fast mixing with preconditioning could both have been predicted
using our results. 
In particular, fixed--step-size preconditioned SGD would have been much more competitive with that work's proposed method than the non-preconditioned decreasing step size SGLD that was used.}

To further explore the value of our tuning guidance, we also consider the behavior of iterate averaging when using the \delete{more practically realistic}preconditioner $\diag(\finfoOpt)^{-1}$, \add{which is less computationally demanding in high dimensions,}
and examine different combinations of the step size and batch size scaling powers that both lead to statistically relevant scaling limits (in particular $(\stepsizePow,\batchsizePow)\in\set{(1,0), (1/2,1/2)}$). In both cases the iterate averages are computed for one epoch. 
As shown in \cref{fig:exp2-results,fig:exp2-results-supp}, since this is a ``small'' number of epochs, the higher order approximation from \cref{eq:asymptotic-iteravg-cov} is required to have an accurate approximation. 
In particular, \cref{fig:exp2-results-supp} confirms that one epoch is not sufficient for \cref{eq:asymptotic-iteravg-cov-simple} to be accurate in this case. \add{This is consistent with our theoretical developments in \cref{sec:optimization-iteravg}.}

\begin{figure}[t]
  \centering
  \begin{tabular}{cccc}
  \begin{subfigure}[b]{0.24\textwidth}
    \includegraphics[width=\textwidth]{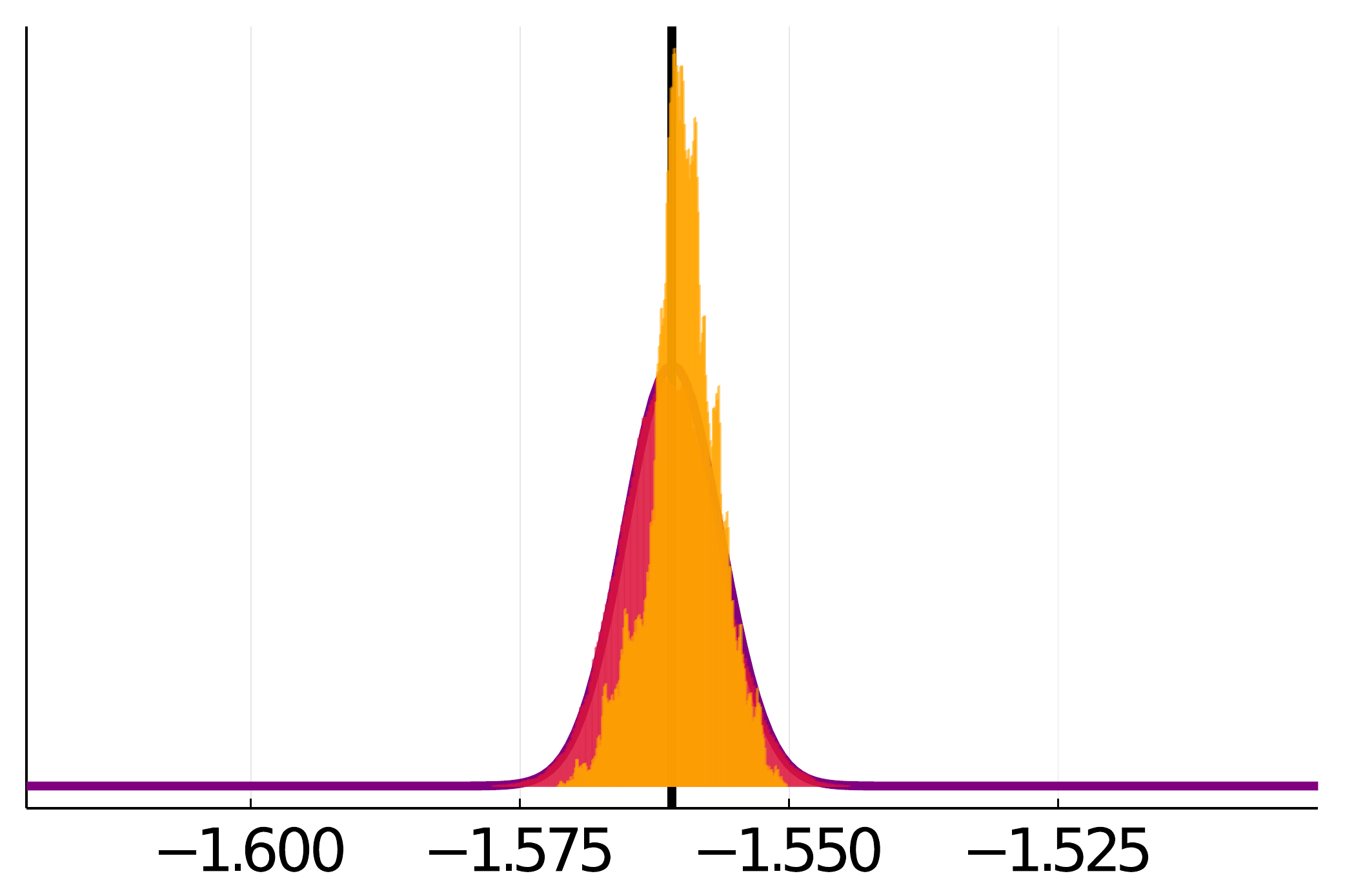}
    \caption{\ }\label{fig:exp2-results-a}
  \end{subfigure} & \hspace{-2em}
  \begin{subfigure}[b]{0.24\textwidth}
    \includegraphics[width=\textwidth]{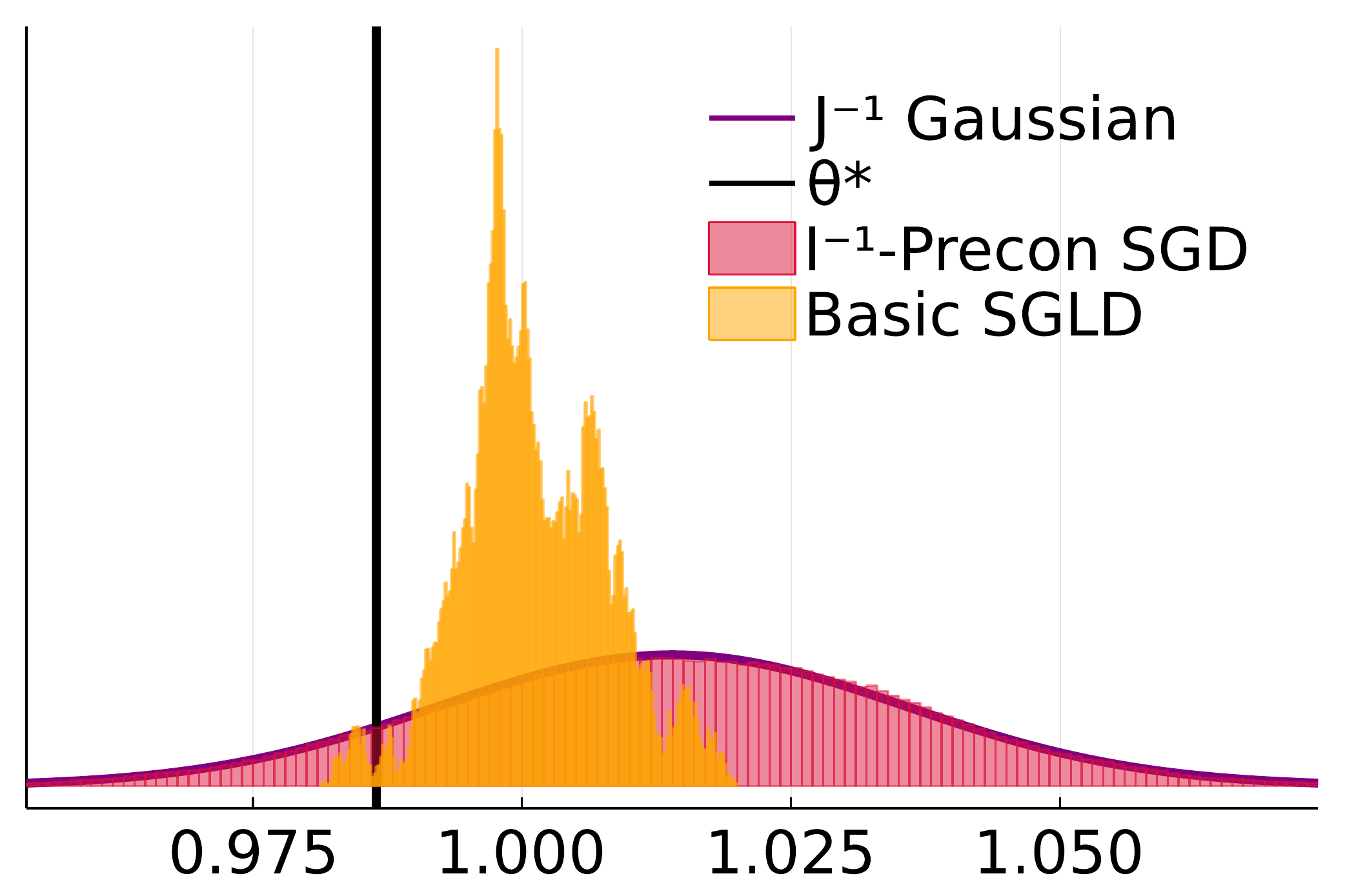}
    \caption{\ }\label{fig:exp2-results-b}
  \end{subfigure} 
  \begin{subfigure}[b]{0.24\textwidth}
    \includegraphics[width=\textwidth]{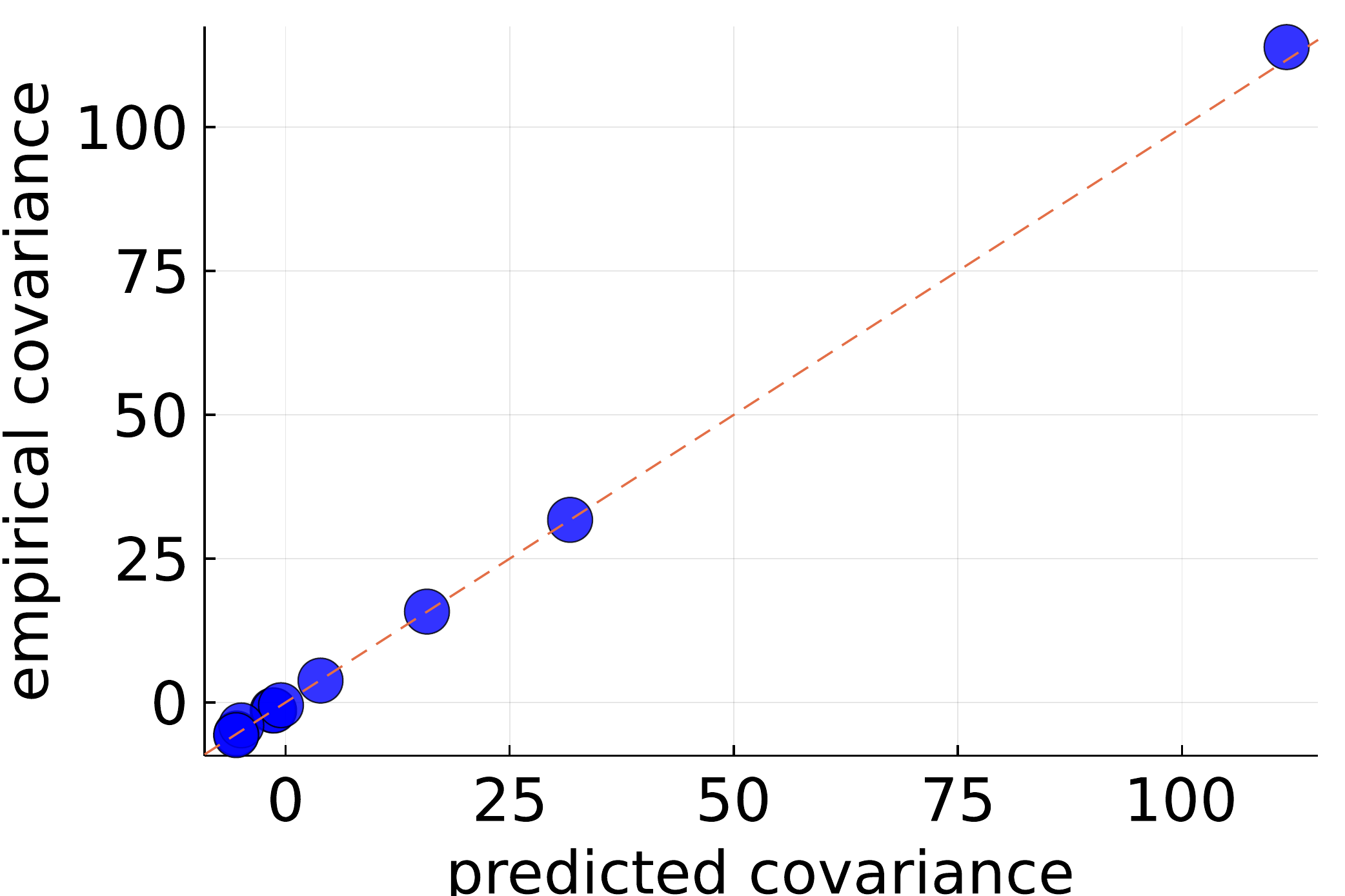}
    \caption{\ }\label{fig:exp2-results-c}
  \end{subfigure} & \hspace{-2em}
  \begin{subfigure}[b]{0.24\textwidth}
    \includegraphics[width=\textwidth]{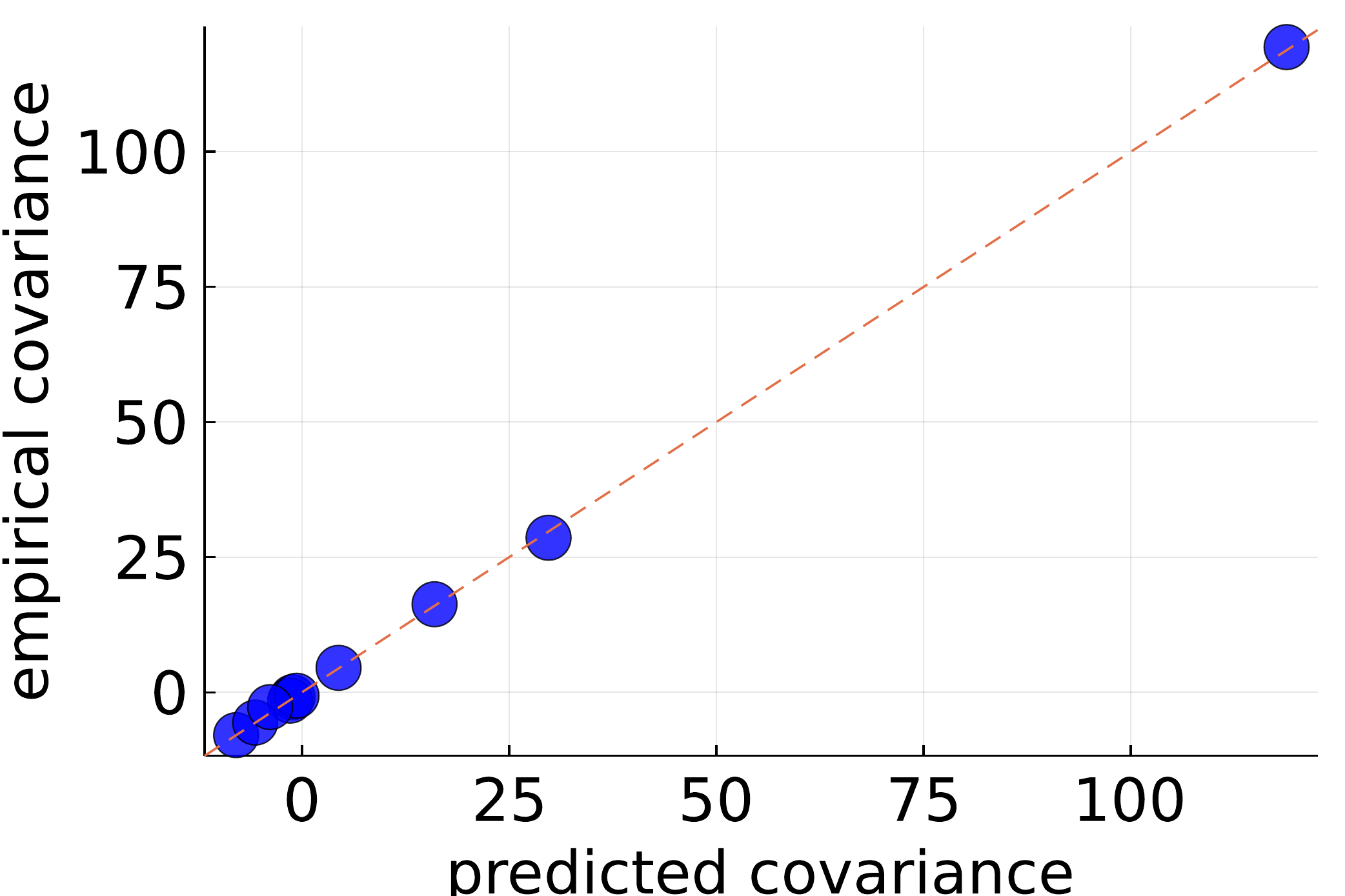}
    \caption{\ }\label{fig:exp2-results-d}
  \end{subfigure} 
  \vspace{-1em}
\end{tabular}
  \caption{ Results of experiment 2. \cref{fig:exp2-results-a,fig:exp2-results-b} show the univariate results for the marginal distributions of parameters 1 and 4 respective when $(\stepsizePow,\batchsizePow)=(1,0)$, and $\batchsizePow=0$. \cref{fig:exp2-results-c,fig:exp2-results-d} show the predicted and actual entries of the variance-covariance matrix for iterate averages when $(\stepsizePow,\batchsizePow)$ is $(1,0)$ and $(1/2,1/2)$ respectively. The predictions for the iterate average are based upon \cref{eq:asymptotic-iteravg-cov}.  }
  \label{fig:exp2-results}
\end{figure}

\subsection{Experiment 3: Large-scale inference for airline delay data---Poisson regression}

Finally, to validate the value of our tuning recommendations in a more complex, clearly misspecified model, we examine the \delete{final year from} the original airline dataset  \citep{DVN/HG7NV7_2008} that the experiments in \citet{pollock2020quasi} were based upon.
In this case the responses are non-negative integers and significantly zero-inflated (relative to a Poisson distribution), and we have opted not to model the zero-inflation to magnify the effect of misspecification.
The model has 25 parameters.
We use the full 2008 data ($\approx$ 1.5 million observations) to estimate the ``ground truth'' quantities $(\trueparm, \finfoOpt, \vinfoOpt)$, and we apply the stochastic gradient algorithms to a dataset consisting of a random subsample of size 150{,}000 from the full \add{2008} dataset. For this example, the matrices $\finfoOpt$ and $\vinfoOpt$ differ significantly in scale,
and hence both preconditioned methods we examine yield materially different uncertainty quantification for the parameter. The non-preconditioned methods are numerically unstable at the comparable step sizes and quickly diverge.
\Cref{fig:exp4-results,tab:exp4-mixing-times} show that both preconditioned methods behave exactly as predicted by the asymptotic theory in \cref{sec:sampling,sec:implications-mixingtime}\add{, and the results directly confirm the predictions made in, and support the recommendations implied by, \cref{tab:tuning-combos}.}

\begin{figure}[t]
  \centering
  \begin{tabular}{cccc}
  \begin{subfigure}[b]{0.24\textwidth}
    \includegraphics[width=\textwidth]{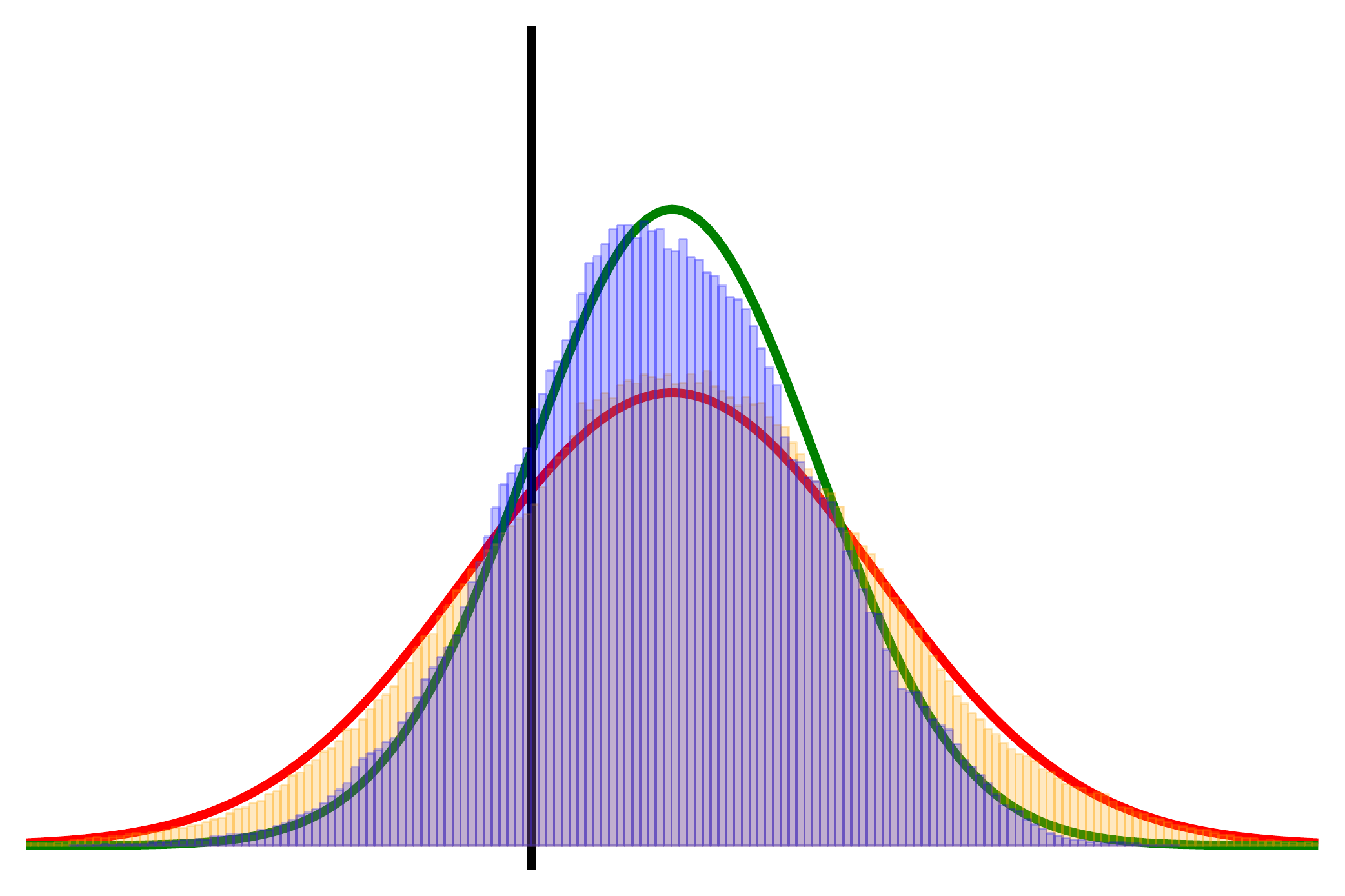}
    \caption{Parameter 2}
  \end{subfigure} & \hspace{-2em}
  \begin{subfigure}[b]{0.24\textwidth}
    \includegraphics[width=\textwidth]{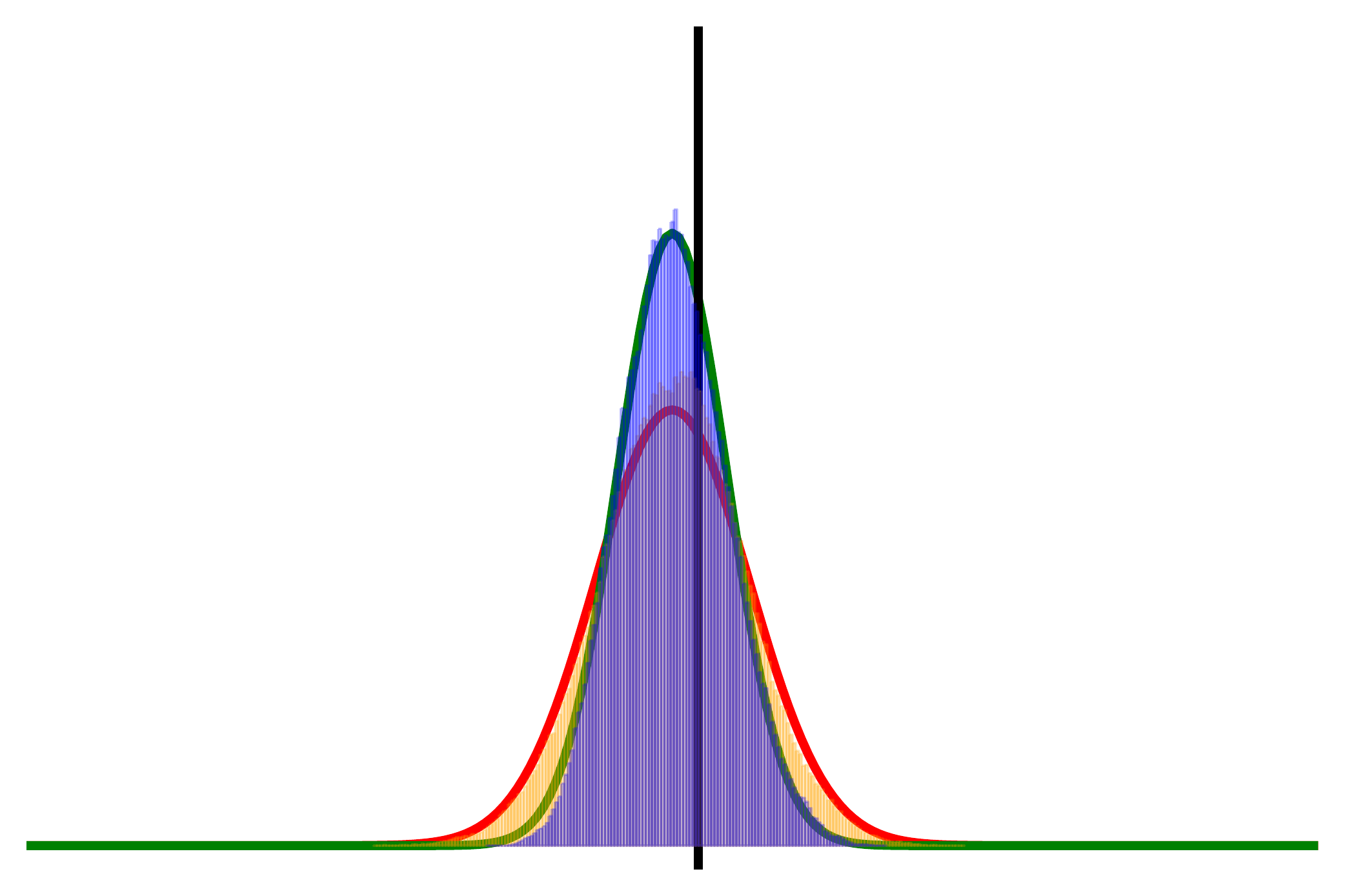}
    \caption{Parameter 3}
  \end{subfigure} & \hspace{-2em}
  \begin{subfigure}[b]{0.24\textwidth}
    \includegraphics[width=\textwidth]{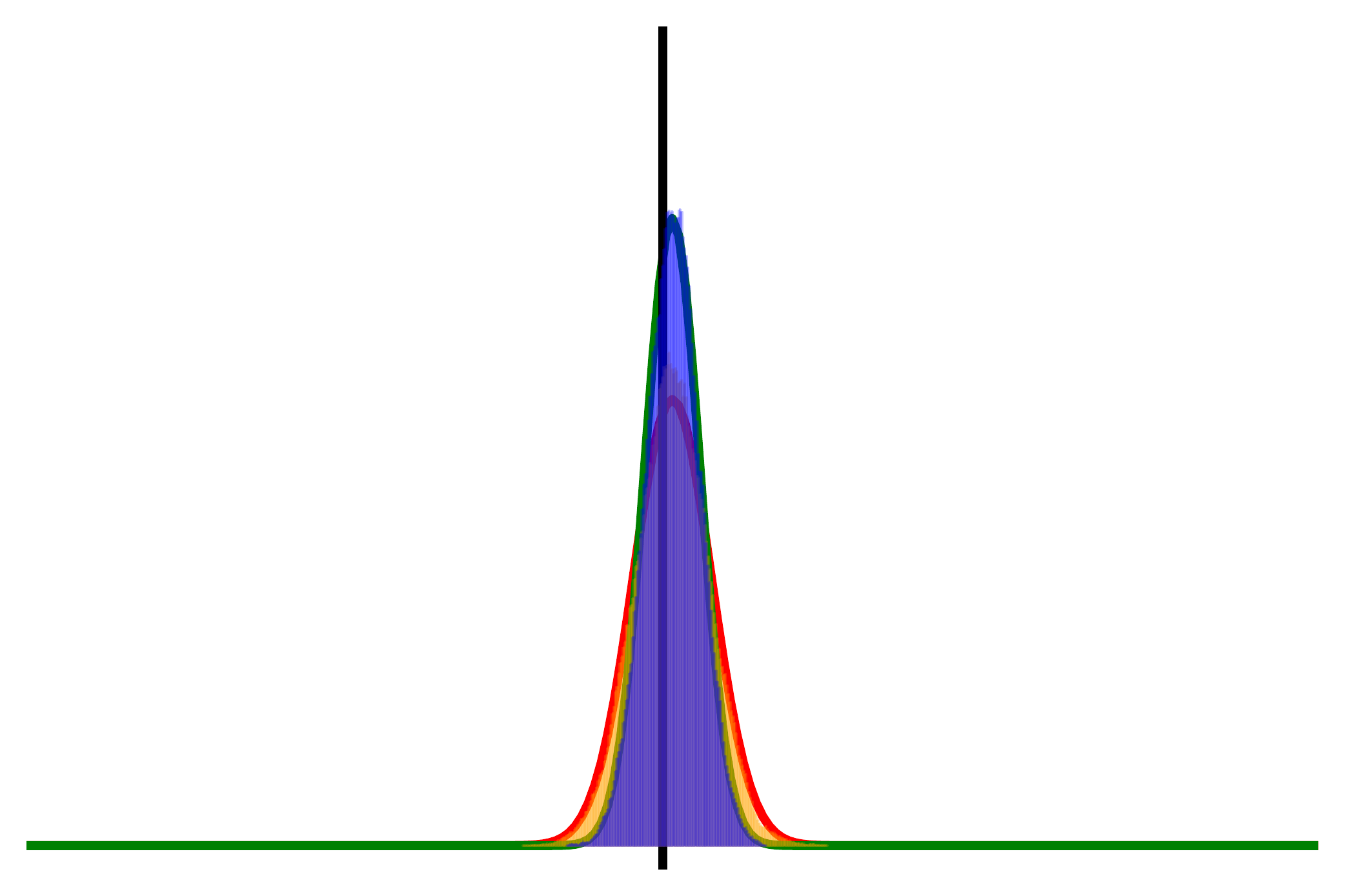}
    \caption{Parameter 24}
  \end{subfigure} & \hspace{-2em}
  \begin{subfigure}[b]{0.24\textwidth}
    \includegraphics[width=\textwidth]{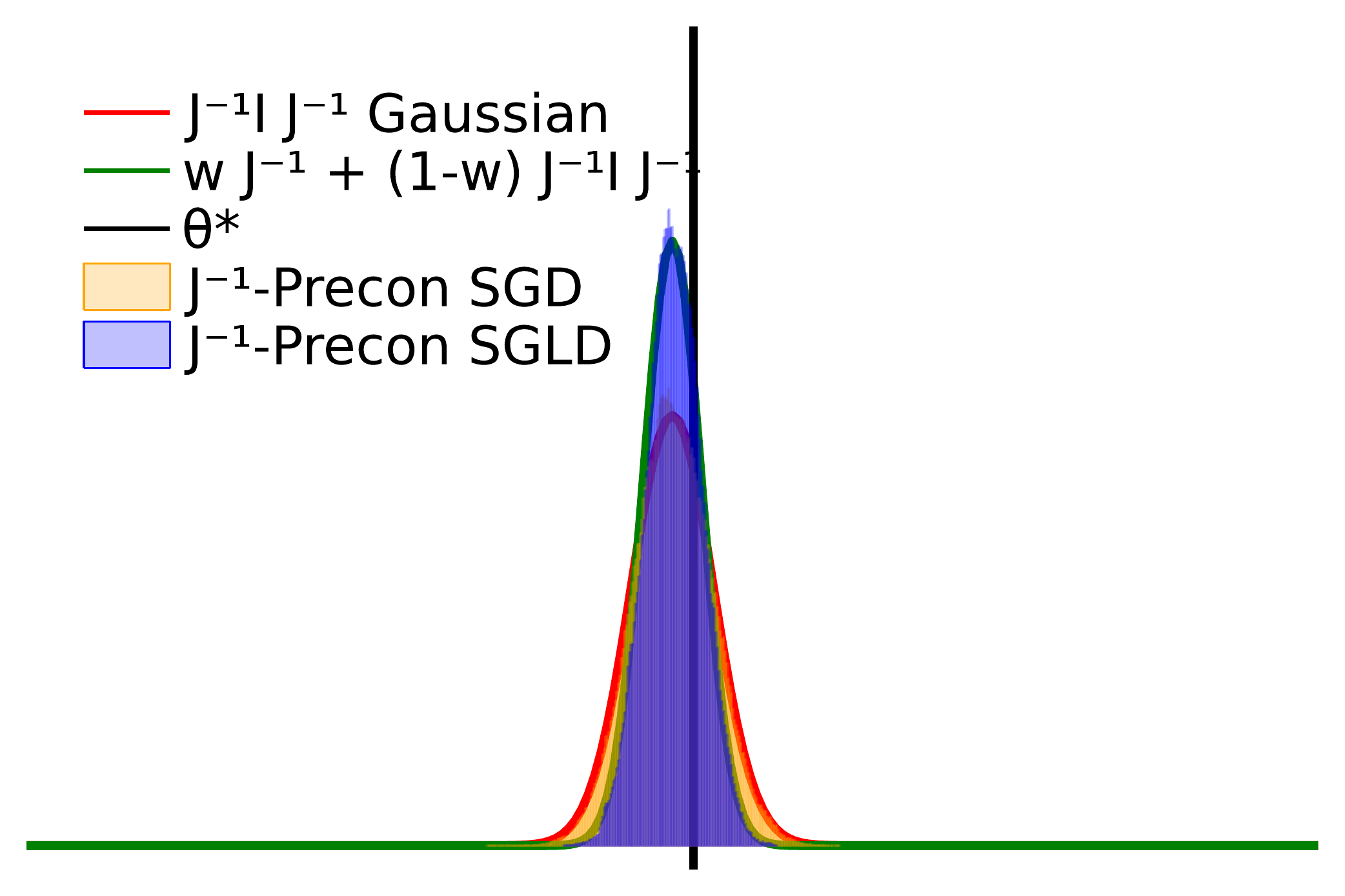}
    \caption{Parameter 25}
  \end{subfigure}
\end{tabular}
\vspace{-1em}
  \caption{Univariate results for experiment 3.}
  \label{fig:exp4-results}
\end{figure}

\section{Discussion}

Given their ubiquity, stochastic gradient methods for optimization and sampling have been analyzed from a range of mathematical perspectives. 
Our work represents a convergence between  non-statistical, continuous-time analyses \citep[e.g,][]{kushner2003stochastic} (and also an often heuristic machine learning literature), %
statistical, discrete-time characterizations \citep[e.g.,][]{toulis2017asymptotic}, and Markov chain analyses of constant--step-size algorithms \citep[e.g.,][]{dieuleveut2020bridging}.

By focusing on the practically relevant fixed--step-size\add{, large-sample} setting, we are able to characterize the stationary distributions of the limiting stochastic processes. 
In combination with our statistical perspective, we are able to derive Bernstein--von Mises theorems: Corollary \ref{cor:sgld-bvm} for the marginal iterates, which is relevant to sampling applications,
and Corollary \ref{cor:iteravg} for iterative averages, which is relevant to optimization.
The latter result complements analogous characterizations of iterate averages with decreasing step size schedules and fixed data \citep{polyak1992acceleration,kushner2003stochastic}.
Both results show that iterate averaging is robust to the choices of tuning parameters, including preconditioning, and can provide statistically optimal numerical estimates of the optimum. 
Our Bernstein--von Mises theorems offer insight into misspecified settings and clarify potential benefits of using stochastic gradients---something present in previous work from the statistical, discrete-time perspective, but limited to the marginal behaviour of individual iterates \citep{toulis2017asymptotic}.  %

Compared to previous heuristic arguments, our theory provides a more precise delineation of when continuous-time approximations are applicable \citep[c.f.][]{mandt2017stochastic,li2018statistical}.
\delete{Indeed, the intuition underlying \citet{mandt2017stochastic} was excellent and largely led to correct results. 
However, while that work lacked a rigorous proof for such scaling limits, we have established the mathematical justification for joint scaling of step size and sample size for stochastic gradient algorithms and have shown that these types of scaling limits can apply much more broadly than those authors asserted.}
For example we show that (a) there is no requirement for batch sizes to be large enough that single iteration increments are approximately Gaussian, and
(b) these scaling limits exist for much broader combinations of joint scaling of step size and sample size leading to different rates of contraction.
At the same time, the precise nature of our results allow us to more clearly understand the limitations of scaling limit analyses: heuristic calculations involving the behaviour across a large number of iterations can be replaced with corresponding approximations from the limiting process precisely when the time horizon involved is $O(1)$ on the limiting time scale. 
Our iterate averaging results provide a case in point: in order to have the same rate of contraction as the posterior distribution and/or the MLE, we must carefully choose the scaling of the step size and batch size together to ensure the spatial scaling is of order $1/\sqrt{\samplesize}$. 
Furthermore, Corollary \ref{cor:iteravg} shows that, in general, iterate averaging must be done over $m \gg 1$ epochs for the covariance of the iterate average to accurately approximate the (rescaled) covariance of the MLE (which is in contrast to the claims of \citet{mandt2017stochastic}).

\add{Our results are advances in two of the key stochastic gradient MCMC research areas identified by \citet{Nemeth:2021:SGLD}. Namely, by invoking a large-sample limit we are able to provide results that circumvent strong convexity assumptions, and we are able to provide comprehensive analyses of various tuning combinations and to make tuning recommendations that can be implemented by or for practitioners. Using this new-found understanding, we were able to explain the empirical results of some critiques of SGLD-like methods and to show that with adequate tuning (that we identify) they would have performed significantly better.}

Overall, our rigorous, continuous-time statistical approach to analyzing stochastic gradient algorithms complements existing work, yielding new insights into the practical effectiveness of stochastic gradient methods, and opens new avenues for future research. 
Because our results are expressed in terms of the joint scaling/choice of dataset size, step size, batch size, and other algorithm parameters, users can gain insight into a stochastic gradient algorithm's statistical behavior for specific choices of these values, which is not possible when taking the limit with the dataset size fixed. 

\add{
Besides the concrete guidance for tuning SGAs and the explanations of prior work's empirical observations, our methods also lay the foundation for similar analyses of other SGAs
and data generating models. Two such examples would be hierarchical models where the parameter dimension depends on the sample size and nonparametric models.  
Such analyses would allow for a systematic and fair comparison between inference methods and provide a better understanding of when stochastic gradient methods are effective. 
Another interesting new line of inquiry is to quantify in which finite-sample regimes our asymptotic results hold, which would enable more precise guidance for use in practice.
}
\delete{However, to fully realize the potential for new insights, it will be necessary to explore extensions to other algorithms and more complex scenarios (e.g., proximal updates and high-dimensional models like neural networks), with our analysis of SGLD with control variates (SGLD-FP) providing a proof-of-concept that such extensions are feasible.}

\bibliography{bib-files/scaling-limit-sgld}

\appendix

\section{Further applications and extensions}

We now discuss applications and extensions of our scaling limit to more complex, practically relevant stochastic gradient algorithms. The poor approximation accuracy of SGLD with uninformed tunings has led to the proposal of many alternatives \citep[e.g.,][]{pollock2020quasi,Nemeth:2021:SGLD,vollmer2016exploration}.
Of particular note are two approaches which are used to reduce the error of both stochastic optimization and sampling.
First, momentum-based methods such as
(stochastic) heavy ball \citep{gupal1972stochastic} and underdamped (stochastic gradient) Langevin dynamics \citep{an2018pid,lessard2016analysis,cyrus2018robust,ma2018quasi} aim to improve on SGLD by decreasing the mixing time of the stochastic process being discretized, typically by moving to a non-reversible process which can in general mix faster than a reversible one. %
Second, variance reduction methods aim to improve the accuracy of the approximate posterior obtained by improving the stochastic estimates of the gradients used in the update formula at each step. For example \citet{nagapetyan2017true} and \citet{Baker.2019} do this with a clever choice of control variates. Lastly, in practice, often parameter spaces are constrained, and we show that this does not affect the scaling limit.

\subsection{Applications to momentum-based algorithms}
\label{sec:example-underdamped}

Special cases of our results include momentum-based acceleration of SGD such as the quasi-hyperbolic momentum algorithm of \citet{ma2018quasi}, 
which includes many momentum-based algorithms as special cases (e.g., Nesterov's accelerated gradient, PID control algorithms \citep{an2018pid}, and more; 
see \citep[Table 1]{ma2018quasi}).
As an example, we show how we can express underdamped stochastic gradient Langevin dynamics in terms of our general stochastic gradient algorithm.
We lift the parameter space to a \emph{phase space} given by $\tilde \parmspace = \parmspace\times \Reals^\parmdim$,
for $\tilde\parm = (\parm, \psi) \in \tilde\parmspace$, we extend the log-likelihood to the phase space according to $\tilde\loglik(\tilde\parm; x) = \loglik(\parm; x) - \psi\transpose\massmatrix^{-1} \psi/2$,
and lift the prior to phase space using the (improper) prior $\tilde\priorposteriordist^{(0)}(\tilde\parm) = \prior(\parm)$.
For (stochastic) heavy ball and underdamped (stochastic gradient) Langevin dynamics  (cf., e.g., \citet[Eqs.~4~and~5]{duncan2017using}),
the lifted Hamiltonian preconditioner $\tilde\precon$ and the lifted diffusion matrix $\tilde\aniso$
are
$\tilde\precon = \sbra{\begin{smallmatrix} 0\ \ & -I \\ I & \precon \end{smallmatrix}}$
and
$\tilde\aniso = \sbra{\begin{smallmatrix} 0\ \ & 0 \\ 0\ \  & \precon\  \end{smallmatrix}}$. 
This yields a combined parameter update formula of
\[\label{eq:sgld-iter-update-momentum}
\parmiter{\samplesize}{\iternum+1}
	& = \parmiter{\samplesize}{\iternum}
		+ \frac{\stepsize\upper{\samplesize}}{2} M^{-1} {\tilde \parm_{\iternum}\upper{\samplesize}}, \qquad
		&
{\psi_{\iternum+1}\upper{\samplesize}}
	& = \rbra{I-\frac{\stepsize\upper{\samplesize}\precon}{2} M^{-1}}{\psi_{\iternum}\upper{\samplesize}}
				+ \frac{\stepsize\upper{\samplesize}}{2}\stochgrad{\samplesize}{\iternum}
				+ \sqrt{\frac{\stepsize\upper{\samplesize}}{\invtemp\upper{\samplesize}}\precon}\ \innoviter{\iternum}.
\]
The corresponding limiting process is
$
\dee\tilde \localparm_{t}
		= -\frac{1}{2} \tilde \driftMat \tilde \localparm_{t} \dee t + \vphantom{\frac{1}{2}}\sqrt{\tilde\diffusionMat} \dee\tilde\weiner_{t},
$
where $\tilde\weiner_{t}$ is a $2\parmdim$-dimensional standard Brownian motion, and the 
drift and diffusion matrices are, respectively,
$
	\tilde\driftMat
		= \const{\stepsize}\sbra{\begin{smallmatrix} 0 & -M^{-1} \\ \finfoOpt & \precon M^{-1} \end{smallmatrix}}
$
and
$\tilde\diffusionMat
		 = \ind{\batchsizePow + \stepsizePow \le \invtempPow}\frac{\const{\stepsize}^2\overline{\const{\batchsize}}}{4\const{\batchsize}}
		 \sbra{\begin{smallmatrix} 0 & 0\\ 0 & \vinfoOpt \end{smallmatrix}}
		  + \ind{\invtempPow \le \batchsizePow + \stepsizePow}\frac{\const{\stepsize}}{\const{\invtemp}} \sbra{\begin{smallmatrix} 0\ \ & 0 \\ 0\ \  & \precon\  \end{smallmatrix}}.
$

\subsection{Extension to control variates}
\label{sec:control-variates}

SGLD methods with control variates  aim improve the reliability of SGLD as an MCMC method to reduce the variance caused by mini-batching by introducing a ``zero variance control variate'' \citep{Baker.2019,nagapetyan2017true}.
Because this modification corresponds to a data-dependent change in the structure of the way stochastic gradients for the potential function are generated, this algorithm does not fit into the framework of \cref{sec:main-results}.
However, our analysis can be easily modified to apply to these control variate methods, as we show in \cref{apx:sketch-cv}.
We find that the scaling limit for SGLD with control variates is nearly the same as without control variates, except that the diffusion term corresponding to mini-batch noise is always $0$. 
This is because the average drift is (by design) not affected by the control variate, the additional Gaussian innovations have the same contribution as before, and the mini-batch noise is now always lower order.
Hence, the spatial scaling can always be chosen so that the noise from Gaussian innovations persists in the limit by taking $\localscalePow=\invtempPow/2$,
where the corresponding limiting process takes the form of \cref{eq:generic-OU} 
with $\driftMat = \const{\stepsize}\precon \finfoOpt$ %
and $\diffusionMat = \frac{\const{\stepsize}}{\const{\invtemp}} \aniso$.

\subsection{Extension to constrained parameter spaces}
\label{sec:constrained-parmspace}

If $\parmspace \subsetneq \Reals^\parmdim$, then the iterations given by \cref{eq:sgld-iter-update} may
exit $\parmspace$, resulting in undefined behaviour.
The typical way to handle this case is to impose \emph{boundary dynamics}.
The two most common examples of such boundary dynamics are \emph{reflecting} and \emph{projecting}.
Projecting maps iterates that would exit $\parmspace$ to the nearest point within $\parmspace$.
Reflecting, defined when the boundary is sufficiently smooth, treats the dynamics between two iterates as the motion of a particle in constant speed linear motion over a fixed time, and when the particle reaches the boundary it collides elastically and ``bounces'' off.
In either case the new iterate is a measurable function of the previous iterate and the vector between the previous iterate what the new iterate would have been without adjusting for the constraint.
Moreover, these conditions both satisfy that the distance between iterates is constrained by what the distance would have been without adjusting for the constraint.
In \cref{apx:sketch-bounded} we consider boundary dynamics satisfying a generalized version of this property.
    When $\parmspace \subsetneq \Reals^\parmdim$ and $\trueparm\in\interior(\parmspace)$ the proof is essentially the same %
 because, intuitively, the assumption that $\localparm\upper{\samplesize}(0)\distto \localparm(0)$ ensures that the processes we consider all start near $\trueparm$ and away from the boundary of $\parmspace$, and thus the spatial scaling drives the boundary of $\parmspace$ outside any bounded set.

\section{Additional Definitions and Technical Results}

Before presenting proofs of the various results of this work, we introduce some additional miscellaneous notations, definitions, and technical results that we will use. 

\subsection{Bernstein-von Mises under misspecification}
\begin{definition}
	The first and second order \emph{Fisher information matrices}, $\vinfo$ and $\finfo$ respectively, are defined for a log-likelihood function $\loglik$ and probability distribution $\datadist$ by
	\*[
	\vinfo(\parm)
			& = \EEE{\data\sim\datadist}\sbra{\grad_\parm\loglik(\parm; \data)\otimes \grad_\parm\loglik(\parm; \data) },
	    &
	    \andT&&
			\finfo(\parm)
	        & =  -\EEE{\data\sim\datadist} \hess_\parm\loglik(\parm; \data).
	\]
\end{definition}

Let $\dataspace$ be a Polish space with $\sigma$-field $\SA{\dataspace}$, $\probmeasures{\dataspace}$ denote the set of probability measures on $\dataspace$, and suppose that $\datadist\in\probmeasures{\dataspace}$.
Suppose that $\dataset{\Nats} \defas \rbra{\data_{\dataidx}}_{\dataidx\in\Nats} \distas \datadist^{\otimes\Nats}$.
Let $\samplesize\in\Nats$ denote a sample size, let $\range{\samplesize} \defas \set{1,\dots,\samplesize}$, and let $\dataset{n} \defas \rbra{\data_{\dataidx}}_{\dataidx\in\range{\samplesize}} \distas \datadist^{\otimes \samplesize}$ be an \iid{} sample of size $\samplesize$ from $\datadist$.

Let $\parmspace \subseteq \Reals^\parmdim$ be open and nonempty, let $Q$ be a regular conditional distribution from $\parmspace$ to $(\dataspace,\SA{\dataspace})$; i.e.:
\begin{itemize}
    \item[(i)] for all $\parm\in\parmspace$, $\modeldist{\parm}\in\probmeasures{\dataspace}$, and
    \item[(ii)] for all $A\in\SA{\dataspace}$, $\modeldist{\cdot}(A):\parm\mapsto \modeldist{\parm}(A)$ is measurable\footnote{$\parmspace$ is equipped with the Borel $\sigma$-field inherited from $\Reals^\parmdim$}.
\end{itemize}
Suppose there exists a $\sigma$-finite measure, $\basemeasure$, on $\dataspace$, such that for all $\parm\in\parmspace$, $\modeldist{\parm} \ll \basemeasure$.
Let $\modeldens{\parm}$ denote a version of $\rnderiv[inline]{\modeldist\parm}{\basemeasure}$ for each $\parm\in\parmspace$.
Let $\loglik(\parm; \dumdata) \defas \log\modeldens{\parm}(\dumdata)$ for all $\parm\in\parmspace$ and $\dumdata\in\dataspace$.
We consider $\model \defas \set{\modeldist{\parm} \ \vert \ \parm\in\parmspace}$ to be a \emph{model} for $\datadist$. The model is \emph{well-specified} when $\datadist \in \MM$, and is \emph{misspecified} otherwise.
The \emph{pseudo-true parameter} of the model is defined as $\trueparm \defas \argmax_{\parm\in\parmspace} \EEE{\data{}\distas\datadist} \loglik(\parm; \data{})$.
If $\basemeasure\ll \datadist$ then
\*[
\trueparm = \argmax_{\parm\in\parmspace} \EEE{\data{}\distas\datadist} \loglik(\parm; \data{}) = \argmin_{\parm\in\parmspace} \KL{\datadist}{\modeldist{\parm}}.
\]
Let $\priormeas\in \probmeasures{\parmspace}$ be any distribution on $\parmspace$.
Let $\jointPP \in \probmeasures{\parmspace\otimes \dataspace^{\Nats{}}}$, given by
\*[
    \jointPP(A\times B)
		& \defas \int \ind{\parm \in A} \sbra{\int \ind {\dumdata\upper{\Nats}\in B} \modeldist{\parm}^\Nats(\dee \dumdata\upper{\Nats})} \priormeas(\dee \parm)
\]
denote the joint distribution of the data and the parameter according to the model and the prior, where $\modeldist{\parm}^\Nats(\dee \dumdata\upper{\Nats})$ denotes the law of an \iid{} sequence from $\modeldist{\parm}$ (an infinite product measure on the cylinder $\sigma$-field). Let $\jointEE$ denote the expectation under $\jointPP$.
The posterior for $\parm$ under the model $\model$ given data $\dataset{\samplesize}$ is the random probability measure on $\parmspace$ given by
\*[
\posteriormeas{\dataset{\samplesize}}(A)
	& \defas \jointCEE{\dataset{\samplesize}}\sbra{\ind{\parm \in A}},
\]
where for a random variable or $\sigma$-field $G$, an expectation operator $\EE$ and a random variable $Y$, $\cEE{G}(Y)$ is the conditional expectation of $Y$ given $G$.
The posterior $\posteriormeas{\dataset{\samplesize}}$ can be viewed as a probability kernel from $\dataspace^\samplesize$ to $\parmspace$.

Let $\lebesgue$ denote the Lebesgue measure.
If $\priormeas\ll \lebesgue$ with $\rnderiv[inline]{\priormeas}{\lebesgue} \revdefas \prior$, then $\posteriormeas{\dataset{\samplesize}} \ll \lebesgue$ with $\rnderiv[inline]{\posteriormeas{\dataset{\samplesize}}}{\lebesgue} = \posterior{\dataset{\samplesize}}$ given by 
\[
    \posterior{\dataset{\samplesize}}(\parm)
        \propto {\prior(\parm) \prod_{\dataidx\in\range{\samplesize}} \modeldens{\parm}(\data_{\dataidx})} = \prior(\parm) \exp\rbra{\sum_{\dataidx\in\range{\samplesize}} \loglik(\parm; \data_\dataidx)}.
\]

Let $\parmMLE{{\samplesize}} \defas \argmax_{\parm\in\parmspace}\sum_{\dataidx\in\range{\samplesize}} \loglik(\parm; \data_\dataidx)$ denote the maximum likelihood estimator (MLE) of $\trueparm$ given the data $\dataset{\samplesize}$.
Posterior distributions have a general tendency to concentrate around the MLE as the sample size increases. Therefore, we will often reparameterize the model by considering a \emph{local parametrization}, where to each parameter $\parm\in\parmspace$ we associate a \emph{local parameter}, $\localparm\in\sqrt{\samplesize}\rbra{\parmspace - \parmMLE{{\samplesize}}}$ based on the identification
\*[
    \localparm = \sqrt{\samplesize}\rbra{\parm - \parmMLE{{\samplesize}}}
\]
and the \emph{local model} is given by
\*[
    \localmodel{\dataset{\samplesize}}
        & \defas \set{\modeldist{\parmMLE{{\samplesize}} + \frac{1}{\sqrt{\samplesize}}\localparm} \ \vert \ \localparm \in \sqrt{\samplesize}\rbra{\parmspace - \parmMLE{{\samplesize}}} }.
\]
The random localization map is given by
\*[
    \loc{\dataset{\samplesize}} : \parm \mapsto \sqrt{\samplesize}\rbra{\parm - \parmMLE{{\samplesize}}}
\]

For a measurable function $f:\Aa\to\Bb$ and a measure $\mu$ on $\Aa$, the \emph{pushforward} of $\mu$ through $f$ is the measure $\pushfwdmeas{f}{\mu}$ on $\Bb$ defined by $[\pushfwdmeas{f}{\mu}](B) = \mu(f^{-1}(B))$ for all measurable $B\subset \Bb$.

\begin{proposition}[BvM under model misspecification, \citet{kleijn2012bernstein}]
    \label{prop:misspec-bvm}
    Under regularity conditions, 
    \*[
        \normTV{\pushfwdmeas{\sbra{\loc{\dataset{\samplesize}}}}{\posteriormeas{\dataset{\samplesize}}} - \Phi} \stk\to{P} 0.
    \]
		with $\trueparm = \argmax_{\parm\in\parmspace} \EEE{\data{}\distas\datadist} \loglik(\parm; \data{})$,
    $\finfoOpt = -\EEE{\data\distas\datadist}\sbra{\hess\loglik(\trueparm;\ \data)}$,
    and $\Phi = \normaldist\rbra{0, \finfoOpt^{-1}}$.
\end{proposition}
 
 \subsection{Convergence modes of measures and operators}
Let $\Aa$ be a measurable space, and let $\bddmeasfuns(\Aa)$ denote the collection of bounded measurable functions on $\Aa$. 
For a function $f:\Aa\to \banachspace$ with $(\banachspace,\norm{\cdot})$ a normed linear space, define
\*[
	\normInf{f}
		& \defas \sup_{a\in \Aa} \norm{f(a)}.
\]
For a sequence of probability measures, $\set{\mu_n}_{n\in\Nats}$ and a probability measure $\mu$ on a measurable space $\Aa$, we have the following modes of convergence:
\begin{itemize}
	\item $\mu_n$ converges in \emph{total variation} to $\mu$, denoted by $\mu_n \stk\to{TV} \mu$, \emph{if and only if}
		\*[
			\sup_{f\in \bddmeasfuns(\Aa)} \frac{\abs{\mu_n f - \mu f}}{\normInf{f}}
				&\to 0.
		\]
	\item if $\Aa$ is also a topological space and the $\sigma$-field on $\Aa$ is the Borel $\sigma$-field, then $\mu_n$ converges \emph{in distribution} (also called \emph{weakly}) to $\mu$, denoted by $\mu_n \distto \mu$, \emph{if and only if} for all $f\in\contbddfuns(\Aa)$, $\abs{\mu_n f - \mu f} \to 0$.
\end{itemize}
Clearly
	\*[
		\rbra{\mu_n \stk\to{TV} \mu} \implies \rbra{\mu_n \stk\to{s} \mu} \implies \rbra{\mu_n \distto \mu}
	\]
while the converses do not hold in general. 

For a Banach Space $\banachspace$ with norm $\norm{\cdot}$ denote its dual space (the space of all bounded linear operators on $\banachspace$) by $\dualbanachspace$. 
$\dualbanachspace$ is a Banach space with norm $\norm{y}\defas \sup_{x\in \banachspace\setminus\set{0}} \abs{fx} / \norm{x}$ for all $f\in\dualbanachspace$. 
Denote the set of bounded linear operators from $\banachspace$ to itself by $\bddlinear(\banachspace)$. 
$\bddlinear(\banachspace)$ is also a Banach space with norm given by $\norm{T} = \sup_{x\in \banachspace\setminus\set{0}} \norm{Tx} / \norm{x}$.

For a sequence of bounded linear operators, $\set{T_n}_{n\in\Nats}$, and a bounded linear operator, $T$, all mapping a Banach Space $L$ to itself, we have the following modes of convergence:
\begin{itemize}
	\item $T_n$ converges \emph{in norm} to $T$ if and only if
		\[
			\norm{T_n-T} = \sup_{(x,y)\in L\times L'}\frac{\abs{\inner{y}{(T_n-T)x}}}{\norm{x}\norm{y}} \to 0
		\]
	\item $T_n$ converges \emph{strongly} to $T$, denoted $T_n \stk\to{s} T$ if and only if for all $x\in L$
		\[
			\sup_{y\in L'}\frac{\abs{\inner{y}{(T_n-T)x}}}{\norm{y}} \to 0
		\]
\end{itemize}
Clearly
\*[
	\rbra{\norm{T_n-T}\to 0} \implies \rbra{T_n \stk\to{s} T}
\]
while the converse does not hold in general.
 
\subsection{Operator Semigroups and Weak Convergence of Markov Processes} \label{sec:semigroups-etc}

For a Banach space, $(\banachspace,\norm{\cdot})$, let $\bddlinear(\banachspace)$ denote the collection of all bounded linear operators from $\banachspace$ to itself, and let $I$ denote the identity operator.
An \emph{operator semigroup} on $\banachspace$ is a function $T:\PosReals \to \bddlinear(\banachspace)$ such that
\begin{enumerate}
    \item[i)] $\semigroup(0) = I$,
    \item[ii)] $\semigroup(t+s) = \semigroup(t)\semigroup(s)$ for all $t,s\in\Reals$.
\end{enumerate}
An operator semigroup is \emph{strongly continuous} if
\begin{enumerate}
    \item[iii)] $\lim_{t\to 0^+} \norm{\semigroup_{t} \testfun - \testfun} = 0$ for all $\testfun \in \banachspace$.
\end{enumerate}
An operator semigroup is \emph{contractive} if
\begin{enumerate}
    \item[iv)] $\norm{T_{t}} \leq 1$ for all $t\in\PosReals$.
\end{enumerate}
The \emph{infinitesimal generator} (or just \emph{generator}, for brevity) of the semigroup $\semigroup$ is the (possibly unbounded) linear operator defined by
\*[
    \generator \testfun = \lim_{t\to 0^+} \frac{\semigroup_{t}\testfun - \testfun}{t}
\]
for $\testfun\in \domain(\generator) = \set{\testfun\in\banachspace \ \vert \ \lim_{t\to 0^+} \rbra{\semigroup_{t}\testfun - \testfun}/{t} \text{ exists}}$.
Let
\*[
    \contvanfuns(\Reals^\parmdim)
        & = \set{\testfun\in\contfuns(\Reals^\parmdim) \ \vert \ \forall\epsilon>0\  \exists K_{\testfun,\epsilon}\subset\Reals^\parmdim \text{ compact} \withT \sup_{\parm\not\in K_{\testfun,\epsilon}}\abs{\testfun(\parm)}\leq \epsilon}
\]
Then $\contvanfuns(\Reals^\parmdim)$ is a Banach space under the norm $\normInf{\testfun} = \sup_{\parm\in\Reals^\parmdim}\abs{\testfun(\parm)}$.
The dual space of $\contvanfuns(\Reals^\parmdim)$ is the space of bounded signed measures under the \emph{total variation norm}
\*[
    \normTV{\mu} = \sup_{\substack{\testfun \in \contvanfuns(\Reals^\parmdim) \\ \normInf{f}\leq 1}} \abs{\int f(\parm) \mu(d\parm) }.
\]
We will work with $\rbra{\banachspace,\norm{\cdot}} = \rbra[1]{\contvanfuns(\Reals^\parmdim), \normInf{\cdot}}$.
A semigroup on $\rbra[1]{\contvanfuns(\Reals^\parmdim), \normInf{\cdot}}$ is \emph{positive} if
\begin{enumerate}
    \item[v)] $\testfun\geq 0 \implies \semigroup\testfun \geq 0$.
\end{enumerate}
A semigroup on $\rbra[1]{\contvanfuns(\Reals^\parmdim), \normInf{\cdot}}$ is \emph{Feller} if it is strongly continuous, contractive, and positive.

Semigroups naturally model the \emph{forward operators} of Markov processes in continuous time. If $X_{t}$ is a Markov process with transition kernels $\kernel_t(\cdot,\cdot)$ then the forward operator corresponding to the Markov process (equivalently, corresponding to its transitio kernels) is defined by
\[
    \semigroup_{t} \testfun(x)
        & = \EE_x \testfun(X_{t}) = \int \testfun(y) \kernel_t(x,dy)
\]
where $\EE_x$ denotes expectation under the law of the Markov process given when $X(0) = x$ almost surely. The semigroup property is then equivalent to the Kolmogorov forward equation.

The generator, $\generator$, of a Feller semigroup $\semigroup$ has a dense domain; $\domain(\generator)$ is dense in $\contvanfuns(\Reals^\parmdim)$. 
A Markov process for which the corresponding forward operators form a a Feller semigroup is called a \emph{Feller process}. 
Feller processes have a richly developed theory; see, for example, \citet{ethier2009markov} or \citet{kallenberg2006foundations}.
The following facts will be useful to us. 
First, every Feller process on $\Reals^\parmdim$ has a version with \emph{\cadlag} (a.k.a \emph{right continuous with left limits}, or  \emph{rcll}) paths, that is for all $t>0$, $\lim_{s\to t^-} X(s)$ exists and $\lim_{s\to t^+} X_{t}$. 
Second for each $I \in \set{[0,T] \ \vert \ T>0} \union \set{\PosReals}$, the collection of all \cadlag functions from $I$ to $\Reals^\parmdim$ is a separable and complete metric space under the \emph{Skorohod metric} \citep[Theorem~A2.2]{kallenberg2006foundations}.
The formula for the Skorohod metric is not particularly illuminating, so is omitted here and may be found in the reference.
This space is denoted by $\skspace(I,\Reals^\parmdim)$.
The Borel $\sigma$-field generated by the Skorohod metric is equal to $\sigma(\set{\pi_t \ \vert \ t\in I'})$ where $\pi_t(X) = X_{t}$ are the projection maps, and $I'$ is any dense subset of $I$.

Let $\smoothcompactsuppfuns(\Reals^\parmdim)$ be the set of functions $\Reals^d\to\Reals$ with compact support and with continuous derivatives of all orders. $\smoothcompactsuppfuns(\Reals^\parmdim)$ is dense in $\contbddfuns(\Reals^\parmdim)$.

\begin{proposition}[Approximation of Markov Chains (compiled from  {\citet{ethier2009markov}}]
	\label{prop:semigroup-convergence-discrete-ek}
	Let $\generator: \smoothcompactsuppfuns(\Reals^\parmdim) \to \contbddfuns(\Reals^\parmdim)$ be linear and suppose that the closure of the graph of $\generator$ (with respect to the graph norm defined by $\norm{f}_\generator = \normInf{f}+\normInf{A f}$ for all $f\in\banachspace$) generates a Feller semigroup $\semigroup$ on $\Reals^\parmdim$.
	Let $\rbra{\limparm_{t}}_{t\in\PosReals}$ be a Markov process with forward operator semigroup $\semigroup$. Let $\rbra{(\parmiter{\samplesize}{\iternum})_{\iternum\in\NatsO}}_{\samplesize\in\Nats}$ be a sequence of (discrete-time) Markov chains on $\Reals^\parmdim$ with respective transition kernels $(\mkernel\upper{\samplesize})_{\samplesize\in\Nats}$. Suppose that $0<\timescale\upper{\samplesize} \to \infty$, and let
    \*[
        \generator\upper{\samplesize}
            & = \timescale\upper{\samplesize}\rbra{\mkernel\upper{\samplesize} - I}
        &&
        & \semigroup\upper{\samplesize}_{t}
            & = \rbra{\mkernel\upper{\samplesize}} ^{\floor{\timescale\upper{\samplesize}\ t }}
        &&
        & \limparm\upper{\samplesize}_{t} = \parmiter{\samplesize}{\floor{\timescale\upper{\samplesize}\ t}}.
    \]
		
	If $\normInf{\generator\upper{\samplesize} \testfun - \generator \testfun} \to 0$
    for all $\testfun \in \smoothcompactsuppfuns(\Reals^\parmdim)$,
    then
    \begin{itemize}
        \item[(a)] $\semigroup\upper{\samplesize}_{t} \stk\to{s} \semigroup_{t}$ for each $t>0$, and
        \item[(b)] If $\limparm\upper{\samplesize}(0) \distto \limparm(0)$ then $\limparm\upper{\samplesize}(\cdot) \distto \limparm(\cdot)$ in the Skorohod metric.
    \end{itemize}
\end{proposition}

\begin{proof}[Proof of \cref{prop:semigroup-convergence-discrete-ek}]
	(a) Follows from Chapter~1, Theorem~6.5 of \citet{ethier2009markov}.
	(b) Follows by combining Chapter~4, Theorem~8.2, Corollary~8.5, and Corollary~8.9 of \citet{ethier2009markov}.
\end{proof}

\subsection{Miscellaneous notation and definitions}

\begin{definition}[Convergence in Probability to a constant]
	\label{def:conv-prob}
	Let $(\Omega,\Ff,\Pr)$ be a probability space, let $(\Xx,\tau)$ be a topological space endowed with the $\sigma$-field $\Ff_\Xx = \sigma(\tau)$, let $\rbra{X_n}_{n\in\Nats}$ be a sequence of $\Xx$-valued random elements, and let $x\in\Xx$. 
	Then \emph{$X_n$ converges to $x$ in probability as $n\to\infty$}, denoted $X_n\stk\to{p}x$, 
	when for every neighbourhood $x\in U\in \tau$ we have
	\*[
		\lim_{n\to\infty}\Pr(X_n \in U^c) = 0.
	\]
\end{definition}	

\begin{lemma}
	\label{lem:conv-prob-subseq}
	Let $(\Omega,\Ff,\Pr)$ be a probability space, let $(\Xx,\tau)$ be a topological space endowed with the $\sigma$-field $\Ff_\Xx = \sigma(\tau)$, let $\rbra{X_n}_{n\in\Nats}$ be a sequence of $\Xx$-valued random elements, and let $x\in\Xx$. 
	
	If for every sub-sequence $n_m$ there is a sub-sub-sequence $n_{m_k}$ such that $X_{n_{m_k}}\to x$ almost surely as $k\to\infty$ then $X_n\stk\to{p}x$.\
	
	If $(\Xx,\tau)$ is first-countable then the converse also holds; if $X_n\stk\to{p}x$ then for every sub-sequence $n_m$ there is a sub-sub-sequence $n_{m_k}$ such that $X_{n_{m_k}}\to x$ almost surely as $k\to\infty$.
\end{lemma}
The proof of this result is the same as in \citet[Theorem~2.3.2]{durrett2019probability}, generalizing the metric space definition of convergence in probability and replacing a sequence of balls of vanishing radius with a countable neighbourhood basis. 
\section{Proof of {\cref{thm:sgld-scaling-limit-multiple-scalings-boundary-simple}}}
\label{sec:pf:sgld-scaling-limit-multiple-scalings}

In this section we prove \cref{thm:sgld-scaling-limit-multiple-scalings-boundary-simple}, as well as an additional result along with what was stated, since both follow from the same premises. The full statement of what we prove is given below. Item 2 below is used in the proof of \cref{cor:sgld-bvm}.

\begin{theorem}[Scaling Limits of SGD/SGLD/LD (Full)]
    \label{thm:sgld-scaling-limit-multiple-scalings-boundary}
    \label{THM:SGLD-SCALING-LIMIT-MULTIPLE-SCALINGS-BOUNDARY}

    Suppose that $(\parmiter{\samplesize}{\iternum})_{\iternum\in\Nats}$ evolves according to the gradient-based algorithm in \cref{eq:sgld-iter-update-boundary} with step size $\stepsize\upper{\samplesize} = \const{\stepsize}\samplesize^{-\stepsizePow}$, $\batchsize\upper{\samplesize} = \floor{\const{b}\samplesize^\batchsizePow}$,
    $\invtemp\upper{\samplesize} = \const{\invtemp}\samplesize^{\invtempPow}$,
    all other tuning parameters constant in $\samplesize$.
    Let $\trueparm \in \Reals^\parmdim$.
    Let $\dataset{\Nats} = \rbra{\data_{\dataidx}}_{\dataidx\in\Nats} \distas \datadist^{\otimes\Nats}$,
    and $\parmLMLE{{\samplesize}}$ be a critical point of the log-likelihood function $\sum_{\dataidx=1}^\samplesize \loglik(\cdot,\data_\dataidx)$ for each $\samplesize\in\Nats$; that is $\sum_{\dataidx=1}^\samplesize \grad\loglik(\parmLMLE{{\samplesize}},\data_\dataidx) = 0$ for all $\samplesize\in\Nats$.

    Let $\localparm\upper{\samplesize}_{t} = \localscale\upper{\samplesize}\rbra{\parmiter{\samplesize}{\floor{\timescale(\samplesize) t}} - \parmLMLE{{\samplesize}}}$, where $\localscale\upper{\samplesize} = \samplesize^{\localscalePow}$,  $\timescale\upper{\samplesize} = \samplesize^{\timescalePow}$, $\localscalePow\in\ointer{0,1}$,
    \*[
        \timescalePow = \min\cbra{\stepsizePow,\  \rbra{\invtempPow+\stepsizePow-2\localscalePow},\ \rbra{\batchsizePow + 2\stepsizePow - 2\localscalePow}}
    .
    \]

    If \cref{ass:basic-reg,ass:gll,ass:grad-norm-mom,ass:J,ass:V} all hold, $\timescalePow > 0$, and $\localparm\upper{\samplesize}(0) \distto \localparm(0)$ then
    \begin{enumerate}
        \item $
            (\localparm\upper{\samplesize}_{t})_{t\in\PosReals} \distto (\localparm_{t})_{t\in\PosReals}$
        in the Skorohod topology in probability,
        where
        $(\localparm_{t})_{t\in\Reals}$  follows the Ornstein--Uhlenbeck process:
        \*[
            d \localparm_{t} = - \frac{\const{\driftName}}{2}\precon \finfo(\trueparm) \localparm_{t} dt + \sqrt{\const{\langNoiseName} \aniso + \const{\mbNoiseName} \precon \vinfo(\trueparm) \precon'} \ d\weiner_{t},
        \]
        with
        \*[
            \const{\driftName} =
                & \begin{cases}
                    \const{\stepsize}
                        & \timescalePow = \stepsizePow \\
                    0   & \timescalePow < \stepsizePow
                \end{cases}
            &,\qquad &&
            \const{\langNoiseName} =
                & \begin{cases}
                    \frac{\const{h}}{\const{\invtemp}}
                        & \timescalePow = \stepsizePow+\invtempPow-2\localscalePow \\
                    0   & \timescalePow < \stepsizePow+\invtempPow-2\localscalePow \\
                \end{cases}
        \]
        and
        \*[
            \const{\mbNoiseName} =
                & \begin{cases}
                    \frac{\const{\stepsize}^2 (1-\const{\batchsize})}{4\const{\batchsize}}
                        & \timescalePow = 1 + 2\stepsizePow - 2\localscalePow \andT \batchsizePow = 1 \andT \text{no replacement}\\
                    \frac{\const{\stepsize}^2}{4\const{\batchsize}}
                        & \timescalePow = \batchsizePow + 2\stepsizePow - 2\localscalePow \andT ( \batchsizePow \neq 1 \orT \text{replacement} )\\
                    0   & \timescalePow < \batchsizePow + 2\stepsizePow - 2\localscalePow.
                \end{cases}
        \]
        \item If $\semigroup\upper{\samplesize}$ and $\semigroup$ are defined as in \cref{prop:semigroup-convergence-discrete-ek},
            then under the conditions above, every subsequence of $\rbra{\semigroup\upper{\samplesize}}_{\samplesize\in\Nats}$, then $\rbra{\semigroup\upper{\subseqsamplesize}}_{\subseqidx\in\Nats}$,
            has a further sub-subsequence, $\rbra{\semigroup\upper{\subsubseqsamplesize}}_{\subsubseqidx\in\Nats}$,
            such that with probability $1$,
     $\semigroup\upper{\subsubseqsamplesize}_{t} \stk\to{s} \semigroup_{t}$ for all $t>0$.
 \end{enumerate}
\end{theorem}

Before beginning the proof of this result, \cref{thm:sgld-scaling-limit-multiple-scalings-boundary}, we require the following lemma, which is used to turn the moment conditions in our assumptions into bounds on the magnitudes of certain random variables that hold all but finitely often with probability 1.

\begin{lemma}
    \label{lem:as-maximal-ineq}
    Let $\alpha:\PosReals \to \PosReals$ be non-decreasing, right continuous with left limits, with $\alpha(0)=0$, and $\lim_{t\to\infty}\alpha_{t} = \infty$.
    Let $Z_i \distas \mu$ for all $i\in\Nats$ (possibly not independent) with $Z_1\geq 0$ almost surely such that $\EE\sbra{\alpha(Z_1)}<\infty$. Let $\alpha^+: u\mapsto \inf\set{t\geq 0 \stT \alpha_{t}\geq u}$ be the generalized inverse of $\alpha$.
    Then
    \*[
        \Pr\rbra{\max_{i \in \range{n}} Z_i \geq \alpha^{+}(n) \quad \textup{i.o.}} = 0.
    \]
\end{lemma}
\begin{proof}[Proof of \cref{lem:as-maximal-ineq}]
    Let $S_{t} = \Pr(Z_1 > t)$ be the survival function of $\mu$, and let $W_n = \alpha(Z_n)$ for each $n\in\Nats$. Note that $\Pr(W_1 > t) = S(\alpha^+_{t})$.
    Then
    \*[
        \infty
            & > \EE\sbra{\rbra{\alpha(Z_1)}} = \int_0^\infty \Pr(W_1 > t) dt \geq \sum_{n=1}^\infty \Pr(W_1 > n)  = \sum_{n=1}^\infty \Pr(W_n > n) \\
    \]
    Therefore, from the Borel--Cantelli lemma $\Pr(W_n > n \quad \textup{i.o.}) = 0$, and equivalently $\Pr(W_n \leq n \quad \textup{a.b.f.o.}) = 1$.
    Now, whenever $W_n \leq n$ for all but finitely many $n$, then there exists $K \in \Nats$ and $I_1,\dots I_K \in \Nats$ with $W_n \leq n$ for all $n \in \Nats \setminus \set{I_j : j\in\range{K}}$. Therefore, for all $n\geq \max_{j\leq K} W_{I_j}$, $\max_{i\leq n} W_i \leq n$.
    Therefore $\Pr( \max_{i\leq n} W_i \leq n \quad \textup{a.b.f.o.}) = 1$, and equivalently $\Pr( \max_{i\leq n} W_i > n \quad \textup{i.o.}) = 0$.
    Finally, $W_i>n$ if and only if $Z_i>\alpha^+(n)$, hence
    \*[
        \Pr( \max_{i\leq n} Z_i > \alpha^+(n) \quad \textup{i.o.}) = 0.
    \]
\end{proof}

\subsection{Proof of \cref{thm:sgld-scaling-limit-multiple-scalings-boundary}}
	Let $\finfoOpt = \finfo(\trueparm)$ and $\vinfoOpt = \vinfo(\trueparm)$.

    The proof proceeds in the following stages.
		In \cref{sec:mainproof-reduce-to-as}, we will reduce the problem of weak convergence in the Skorohod topology in probability to one of weak convergence in the Skorohod topology almost-surly along subsequences and construct appropriate such subsequences.
		In \cref{sec:mainproof-notation} we introduce notation that will be useful in the remainder of the proof.
		In \cref{sec:mainproof-apply-ek} we discuss what is needed to apply \cref{prop:semigroup-convergence-discrete-ek} to establish the processes converge weakly in the Skorohod topology almost-surely.
		This amounts to showing that the difference between the approximate generator and limiting generator evaluated a smooth test function with compact support vanishes uniformly.
		We will examine this difference in two regimes.
		First, in \cref{sec:mainproof-outside}, we will consider arguments sufficiently far from the support of the test function.
		Then, in \cref{sec:mainproof-inside}, we will consider arguments in or close to the support of the test function, and use a Taylor series expansion of the approximate generator to divide this into three types of non-zero terms.
		The first type is non-remainder terms that vanish and have no corresponding term in the limiting generator; these are handled in \cref{sec:mainproof-inside-vanish}.
		The second type is terms that do not vanish and do have corresponding terms in the limiting generator; these are handled in \cref{sec:mainproof-inside-drift,sec:mainproof-inside-innov,sec:mainproof-inside-ll}.
		The third type of term is the remainder term, which is handled in \cref{sec:mainproof-inside-remainder}.
		Putting all of this together allows us to apply \cref{prop:semigroup-convergence-discrete-ek} along our subsequences, establishing the main result.

    \subsubsection{Reduction to almost-sure convergence on subsequences}\label{sec:mainproof-reduce-to-as}
		Let
    \*[
        \Upsilon\upper{\samplesize}
            & = \max\rbra{ \Upsilon\upper{\samplesize}_{1},\Upsilon\upper{\samplesize}_{2},\Upsilon\upper{\samplesize}_{3}}, \\
        \Upsilon\upper{\samplesize}_{1}
            & = \samplesize^{\asspower{3}} \norm{\parmLMLE{{\samplesize}} - \trueparm},\\
        \Upsilon\upper{\samplesize}_{2}
            & = \textstyle \sup_{\parm\in\ball\rbra{\parmMLE{\samplesize},r_{\finfo,\samplesize}/\samplesize^\localscalePow}} \norm{\finfoEmp{\samplesize}(\parm) - \finfo(\trueparm)}, \\
        \Upsilon\upper{\samplesize}_{3}
            & = \textstyle\sup_{\parm\in\ball\rbra{\parmMLE{\samplesize},r_{\vinfo,\samplesize}/\samplesize^\localscalePow}} \norm{\vinfoEmp{\samplesize}(\parm) - \vinfo(\trueparm)}.
    \]
    Each of the $\Upsilon$ terms corresponds to the important quantity that vanishes in probability for one of the assumptions.
    For example, $\Upsilon\upper{\samplesize}_{1}$ controls how quickly the local MLE converges under \cref{ass:grad-norm-mom} which lets us use a weaker moment assumption for the sup-norm of the Hessian of the log-likelihood.

    By assumption,
		$\Upsilon\upper{\samplesize}\stk\to{p}0$.
    Then, by \cref{lem:conv-prob-subseq}, for every subsequence $(\subseqsamplesize)_{\subseqidx\in\Nats}$ there is a further sub-sub\-sequence $(\subsubseqsamplesize)_{\subsubseqidx\in\Nats}$ so that this convergence is almost sure.
    Along an arbitrary such sub-subsequence, we will verify that $(\localparm\upper{\subsubseqsamplesize})_{t\in\PosReals} \distto (\localparm_{t})_{t\in\PosReals}$ in the Skorohod topology almost surely.
    Since weak convergence is metrizable (e.g., by the Levi--Prokhorov metric, and hence corresponds to a topology on probability distributions), and since for any subsequence $(\subseqsamplesize)_{\subseqidx\in\Nats}$ we will have shown a further subsequence $(\subsubseqsamplesize)_{\subsubseqidx\in\Nats}$ such that $(\localparm\upper{\subsubseqsamplesize}_{t})_{t\in\PosReals} \distto (\localparm_{t})_{t\in\PosReals}$ a.s., by \cref{lem:conv-prob-subseq} it must hold that
    $(\localparm\upper{\samplesize}_{t})_{t\in\PosReals} \distto (\localparm_{t})_{t\in\PosReals}$ in probability.

    Now, let $(\subseqsamplesize)_{\subseqidx\in\Nats}$ be an arbitrary subsequence\footnote{Since every sub-subsequence is itself a subsequence, we can simplify our notation from here onward.} of $\Nats$ such that
    $\Upsilon\upper{\subseqsamplesize} \stk\to{a.s.} 0$.
    Let $\samplespace$ denote the underlying probability space.
    Let
    \*[
        \samplespace\upper{0}
            & = \Intersect_{i=1}^3 \samplespace\upper{i},\\
        \samplespace\upper{1}
            & = \set{\Upsilon\upper{\subseqsamplesize} \to 0},\\
        \samplespace\upper{2}
            & = \set{\max_{\dataidx\in\range{\samplesize}} \norm{\grad\loglik(\trueparm; \data_\dataidx)} \leq \samplesize^{1/\assmoment{2}}  \quad \textup{a.b.f.o} }\ , \\
				\samplespace\upper{3}
						& =             \set{\max_{\dataidx\in\range{\samplesize}} \normInf{{\hess\loglik(\cdot; \data_\dataidx)}} \leq \samplesize^{1/\assmoment{3}}  \quad \textup{a.b.f.o.} }\ .
    \]
    By assumption, and by applying \cref{lem:as-maximal-ineq} to power functions of the form $\alpha:t\mapsto t^p$ and random variables $\norm{\grad\loglik(\trueparm; \data_\dataidx)}$ and $\normInf{{\hess\loglik(\cdot; \data_\dataidx)}}$, $\samplespace\upper{0}$ is a sure set.

		\subsubsection{Additional notation used in the proof}\label{sec:mainproof-notation}
    We notate the increments of the localized iterative algorithms (given that $\scaleparmiter{\samplesize}{0} = \limparm$) due to the Gaussian innovation $(\innov{})$, the gradient step contribution of the prior $(\prior{})$, the mini-batch gradient step based on the log-likelihood $(\loglik{})$, and the total increment, respectively, as
    \*[
        \incrInnov{\samplesize}
            & \defas \localscale\upper{\samplesize}\sqrt{\stepsize\invtemp^{-1}\aniso}\ \innoviter{1}, \\
        \incrPrior{\samplesize}(\limparm)
            & \defas \frac{\stepsize\localscale\upper{\samplesize} \precon}{2 \samplesize}
                \grad \log\prior\rbra{\parmLMLE{{\samplesize}}
                    + (\localscale\upper{\samplesize})^{-1} \limparm}, \\
        \incrLoglik{\samplesize}(\limparm)
            & \defas \frac{\stepsize\localscale\upper{\samplesize} \precon}{2\batchsize\upper{\samplesize}} \ \sum_{\mbidx\in \range{\batchsize\upper{\samplesize}}} \grad\loglik\rbra{\parmLMLE{{\samplesize}}
                + (\localscale\upper{\samplesize})^{-1} \limparm;\ \data_{\mbdataidx{\samplesize}{1}{\mbidx}}} , \andT \\
        \incrTotal{\samplesize}(\limparm)
            & \defas \incrInnov{\samplesize} + \incrPrior{\samplesize}(\limparm) + \incrLoglik{\samplesize}(\limparm).
    \]

    We define the sequence of operators $\generator\upper{\samplesize}$ by
    \[\label{eq:approx-gen-defn}
        [\generator\upper{\samplesize} \testfun](\limparm)
            & = \timescale\upper{\samplesize}
                \rbra{ \cEE{\dataset{\Nats}}
                    \sbra{
                        \testfun(\limparm + \incrTotal{\samplesize}(\limparm))} - \testfun(\limparm) }.
    \]
    for all $\samplesize\in\Nats$, and all $\testfun \in \smoothcompactsuppfuns(\Reals^\parmdim)$,  where $\timescale\upper{\samplesize} = \samplesize$.
    The generator of the (presumed, at this point) limiting OU process is given by
    \[\label{eq:lim-gen-defn}
        [\generator\testfun](\limparm)
            & = - \inner{\frac{\const{\driftName}}{2}\precon \finfoOpt \limparm}{\grad \testfun(\limparm)} + \frac{1}{2} \rbra{\const{\langNoiseName} \aniso + \const{\mbNoiseName} \precon \vinfoOpt \precon'} \frobprod \hess\testfun(\limparm)
    \]

		\subsubsection{How \cref{prop:semigroup-convergence-discrete-ek} is applied}\label{sec:mainproof-apply-ek}
		Consider a single realization of $\dataset{\Nats} \in \samplespace\upper{0}$.
		Our goal, now, is to apply \cref{prop:semigroup-convergence-discrete-ek}, treating $\dataset{\Nats}$ as fixed.
    To do so, it suffices to show that for each $\testfun \in \smoothcompactsuppfuns(\Reals^\parmdim)$ we have
    \*[
        \lim_{\subseqidx\to\infty} \sup_{\limparm\in\Reals^\parmdim} \abs{[\generator\upper{\subseqsamplesize} \testfun](\limparm) - [\generator \testfun](\limparm)} = 0.
    \]
    For an arbitrary test function, $\testfun\in\smoothcompactsuppfuns(\Reals^\parmdim)$, with compact support $\compactset_0$, we will show this in two parts.
    First we will identify a compact extension, $\compactset_1\supset \compactset_0$ to the compact support of $f$ such that
    \*[
        \lim_{\subseqidx\to\infty} \sup_{\limparm\in \compactset_1^c} \abs{[\generator\upper{\subseqsamplesize} \testfun](\limparm) - [\generator \testfun](\limparm)} = 0.
    \]
    Then we will separately show that
    \*[
        \lim_{\subseqidx\to\infty} \sup_{\limparm\in \compactset_1} \abs{[\generator\upper{\subseqsamplesize} \testfun](\limparm) - [\generator \testfun](\limparm)} = 0.
    \]

		\subsubsection{Convergence away from the test function support}\label{sec:mainproof-outside}
    For all $\limparm \in \compactset_0^c$, $f(\limparm) = 0$, $\grad f (\limparm) = 0$, and $\hess f (\limparm) = 0$.
    Therefore, for any $\compactset_1 \supset \compactset_0$,
    \[ \label{eq:compact-complement-bound-br-prob}
        &\hspace{-1em}
        \sup_{\limparm\in \compactset_1^c} \abs{[\generator\upper{\subseqsamplesize} \testfun](\limparm) - [\generator \testfun](\limparm)}\\
            & \leq \timescale\upper{\subseqsamplesize} \normInf{\testfun} \sup_{\limparm\in \compactset_1^c}\cPr{\dataset{\Nats}}\sbra{
                    \limparm + \incrTotal{\subseqsamplesize}(\limparm)
                        \in \compactset_0 }.
    \]
    Let $\radius_0 = \sup_{\limparm\in\compactset_0} \norm{\limparm}$.
    Let $\compactset_1 = \set{\limparm\in\Reals^\parmdim \stT \norm{\limparm}\leq 2\radius_0 + 2\const{0}}$, where
    \*[
        \const{0} = \frac{\const{\stepsize} \norm{\precon}}{2} \rbra{3+\norm{\grad \log\prior\rbra{\trueparm}}} + \sqrt{\const{\stepsize}/\const{\invtemp} \norm{\aniso}}.
    \]

    Then, using \cref{eq:compact-complement-bound-br-prob} and $\incrTotal{\subseqsamplesize}(\limparm) = \incrInnov{\subseqsamplesize}(\limparm)+\incrPrior{\subseqsamplesize}(\limparm)+\incrLoglik{\subseqsamplesize}(\limparm)$,
    \[ \label{eq:compact-complement-bound-br-prob-simp}
        &\hspace{-1em}\sup_{\limparm\in \compactset_1^c} \abs{[\generator\upper{\subseqsamplesize} \testfun](\limparm) - [\generator \testfun](\limparm)}\\
            & \leq \timescale\upper{\subseqsamplesize} \normInf{\testfun} \sup_{\norm{\limparm} > 2 \radius_0 + 2 \const{0}} \cPr{\dataset{\Nats}}\sbra{\norm{\limparm + \incrTotal{\subseqsamplesize}(\limparm)} \leq \radius_0} \\
            & \leq \timescale\upper{\subseqsamplesize} \normInf{\testfun} \sup_{\norm{\limparm} > 2 \radius_0 + 2 \const{0}} \cPr{\dataset{\Nats}}\sbra{\norm{\incrInnov{\subseqsamplesize}}
                \geq \norm{\limparm} - \norm{\incrPrior{\subseqsamplesize}(\limparm)} - \norm{\incrLoglik{\subseqsamplesize}(\limparm)} - \radius_0 }.
    \]

    For $\limparm \in \compactset_1^c$, using the assumption that $\grad\log\prior$ is $L_0$-Lipschitz and $\stepsize\upper{\samplesize} = \const{\stepsize}\samplesize^{\stepsizePow}$ and $\localscale\upper{\samplesize} = \samplesize^{\localscalePow}$,
    \*[
    \norm{\incrPrior{\subseqsamplesize}(\limparm)}
        & \leq \frac{\stepsize\upper{\subseqsamplesize}\localscale\upper{\subseqsamplesize} \norm{\precon}}{2\subseqsamplesize}
            \norm{\grad \log\prior\rbra{\parmLMLE{{\subseqsamplesize}}
                + (\localscale\upper{\subseqsamplesize})^{-1} \limparm}} \\
        & \leq \frac{\stepsize\upper{\subseqsamplesize}\localscale\upper{\subseqsamplesize} \norm{\precon}}{2\subseqsamplesize}
            \rbra{\norm{\grad \log\prior\rbra{\trueparm}} + \lipschitzconst_0 \norm{\parmLMLE{{\subseqsamplesize}} - \trueparm} + \frac{\lipschitzconst_0 \norm{\limparm}}{\localscale\upper{\subseqsamplesize}}}\\
        & \leq \frac{\const{\stepsize}\subseqsamplesize^{\localscalePow-\stepsizePow-1} \norm{\precon}}{2}
            \rbra{\norm{\grad \log\prior\rbra{\trueparm}} + \lipschitzconst_0 \norm{\parmLMLE{{\subseqsamplesize}} - \trueparm} + \frac{\lipschitzconst_0 \norm{\limparm}}{\subseqsamplesize^{\localscalePow}}},
    \]
    and similarly
    \*[
    &\hspace{-1em}\norm{\incrLoglik{\subseqsamplesize}(\limparm)}\\
        & \leq \frac{\stepsize\upper{\subseqsamplesize}\localscale\upper{\subseqsamplesize} \norm{\precon}}{2\batchsize\upper{\subseqsamplesize}} \norm{ \sum_{\mbidx\in \range{\batchsize\upper{\subseqsamplesize}}} \grad\loglik\rbra{ \parmLMLE{{\subseqsamplesize}}
            + (\localscale\upper{\subseqsamplesize})^{-1} \limparm;\ \data_{\mbdataidx{\subseqsamplesize}{1}{\mbidx}}}} \\
        & \leq \frac{\const{\stepsize}\subseqsamplesize^{\localscalePow-\stepsizePow} \norm{\precon}}{2\batchsize\upper{\subseqsamplesize}} \sum_{\mbidx\in \range{\batchsize\upper{\subseqsamplesize}}} \rbra{\norm{\grad\loglik\rbra{ \trueparm; \ \data_{\mbdataidx{\subseqsamplesize}{1}{\mbidx}}}}
            + \lipschitzconst(\data_{\mbdataidx{\subseqsamplesize}{1}{\mbidx}}) \norm{\parmLMLE{{\subseqsamplesize}} - \trueparm} + \frac{\lipschitzconst(\data_{\mbdataidx{\subseqsamplesize}{1}{\mbidx}}) \norm{\limparm}}{{\subseqsamplesize}^{\localscalePow}}} \\
        & \leq \frac{\const{\stepsize}\subseqsamplesize^{\localscalePow-\stepsizePow} \norm{\precon}}{2} \rbra{\lipschitzconst_\star(\dataset{\subseqsamplesize})
            + \lipschitzconst(\dataset{\subseqsamplesize}) \norm{\parmLMLE{{\subseqsamplesize}} - \trueparm} + \lipschitzconst(\dataset{\subseqsamplesize})\frac{ \norm{\limparm}}{{\subseqsamplesize}^{\localscalePow}}}
    \]
    where we define the (random) Lipschitz constants $\lipschitzconst(\data_\dataidx)$, $\lipschitzconst_\star(\dataset{\subseqsamplesize})$, and $\lipschitzconst(\dataset{\subseqsamplesize})$ by:
    \*[
        \lipschitzconst(\data_\dataidx)
            & \defas \normInf{{\hess\loglik(\cdot; \data_\dataidx)}}, \\
        \lipschitzconst_\star(\dataset{\subseqsamplesize})
            & \defas \max_{\dataidx\leq\subseqsamplesize}\norm{\grad\loglik\rbra{ \trueparm; \ \data_{\dataidx}}}, \andT \\
        \lipschitzconst(\dataset{\subseqsamplesize})
            & \defas \max_{\dataidx\leq\subseqsamplesize}\lipschitzconst(\data_\dataidx).
    \]
    Using that $\dataset{\Nats}\in\samplespace\upper{0}$, so that $\Upsilon\upper{\subseqsamplesize}\to 0$ etc., if $\subseqidx$ is large enough that all of the following hold:
    \*[
        \sup_{\subseqidx'\geq \subseqidx}\Upsilon\upper{\subseqsamplesize}
            & \leq \min(1, \lipschitzconst_0^{-1}), \\
        1
            &\geq \sup_{\subseqidx'\geq \subseqidx} \frac{\lipschitzconst_\star(\dataset{\samplesize_{\subseqidx'}})}{\samplesize_{\subseqidx'}^{1/\assmoment{2}}},
            \\
        \subseqsamplesize
            & \geq
                \max(\rbra{2\const{\stepsize} \norm{\precon}}^{1/(1/\assmoment{3}-\stepsizePow)},\rbra{2\const{\stepsize}\lipschitzconst_0 \norm{\precon}}^{\frac{1}{\stepsizePow+1-\timescalePow-\localscalePow}}), \qquad  \andT \\
        1
            &\geq
                \sup_{\subseqidx'\geq \subseqidx} \frac{\lipschitzconst(\dataset{\samplesize_{\subseqidx'}})}{\samplesize_{\subseqidx'}^{1/\assmoment{3}}}
             ;
    \]
    then, using that $0<\localscalePow<1$,
    \*[
    \norm{\incrPrior{\subseqsamplesize}(\limparm)}
        & \leq \frac{\const{\stepsize}\norm{\precon}}{2}
            \rbra{\norm{\grad \log\prior\rbra{\trueparm}} + 1} + \frac{1}{4}\norm{\limparm},
    \]
    and
    \*[
    \norm{\incrLoglik{\subseqsamplesize}(\limparm)}
        & \leq \frac{\const{\stepsize}\subseqsamplesize^{-\stepsizePow+\localscalePow} \norm{\precon}}{2} \rbra{
            \subseqsamplesize^{1/\assmoment{2}}
            + \subseqsamplesize^{1/\assmoment{3}}\Upsilon\upper{\subseqsamplesize} + \subseqsamplesize^{1/\assmoment{3}-\localscalePow} \norm{\limparm}} \\
        & \leq \frac{\const{\stepsize} \norm{\precon}}{2} \rbra{
            \subseqsamplesize^{1/\assmoment{2}-\stepsizePow+\localscalePow}
            + \subseqsamplesize^{1/\assmoment{3}-\stepsizePow+\localscalePow}\Upsilon\upper{\subseqsamplesize} + \subseqsamplesize^{1/\assmoment{3}-\stepsizePow} \norm{\limparm}},\\
        & \leq \const{\stepsize} \norm{\precon} + \frac{1}{4} \norm{\limparm}.
    \]

    Therefore, for $\limparm\in \compactset_1^c$ (and hence $\norm{\limparm} > 2\radius_0+2\const{0}$),
    \*[
    &\hspace{-1em}
    \norm{\limparm} - \norm{\incrPrior{\subseqsamplesize}(\limparm)} - \norm{\incrLoglik{\subseqsamplesize}(\limparm)} - \radius_0\\
        & \geq \frac{1}{2} \norm{\limparm} - \frac{\const{\stepsize} \norm{\precon}}{2} \rbra{3+\norm{\grad \log\prior\rbra{\trueparm}}} - \radius_0 \\
        & \geq \sqrt{\const{\stepsize}/\const{\invtemp} \norm{\aniso}} .
    \]
    Therefore, combining this with \cref{eq:compact-complement-bound-br-prob-simp} and the definition of $\incrInnov{\subseqsamplesize}(\limparm)$,
    \*[
        \lim_{\subseqidx\to\infty}\sup_{\limparm\in \compactset_1^c} \abs{[\generator\upper{\subseqsamplesize} \testfun](\limparm) - [\generator \testfun](\limparm)}
            & \leq \lim_{\subseqidx\to\infty} \timescale\upper{\subseqsamplesize} \normInf{\testfun} \cPr{\dataset{\Nats}}\rbra{\norm{\innoviter{1}} \geq \subseqsamplesize^{\stepsizePow/2+\invtempPow/2-\localscalePow}} \\
            & \leq \lim_{\subseqidx\to\infty} \timescale\upper{\subseqsamplesize} \normInf{\testfun} \parmdim\  \cPr{\dataset{\Nats}}\rbra{\abs{\innoviter{1,1}} \geq \frac{1}{\sqrt{\parmdim}} \subseqsamplesize^{\stepsizePow/2+\invtempPow/2-\localscalePow}}\\
            & \leq \lim_{\subseqidx\to\infty} 2 {\subseqsamplesize^\timescalePow} \normInf{\testfun} \parmdim \exp(-\subseqsamplesize^{{\stepsizePow+\invtempPow-2\localscalePow}} / 2 \parmdim) \\
            & = 0 .
    \]
    since $\stepsizePow+\invtempPow-2\localscalePow\geq \timescalePow>0$.

  	\subsubsection{Taylor expansion near the test function support}\label{sec:mainproof-inside}
    Recalling the definition of $\generator\upper{\subseqsamplesize}$ in \cref{eq:approx-gen-defn},
    using the definition of the time-scaling factor $\timescale\upper{\samplesize} = \samplesize^\timescalePow$,
    taking a second-order Taylor expansion of the test function $\testfun\in\smoothcompactsuppfuns$,
    and applying the decomposition $\incrTotal{\subseqsamplesize}(\limparm) = \incrInnov{\subseqsamplesize}(\limparm)+\incrPrior{\subseqsamplesize}(\limparm)+\incrLoglik{\subseqsamplesize}(\limparm)$,
    \*[
    &\hspace{-1em}
    [\generator\upper{\subseqsamplesize} \testfun](\limparm) \\
        & = \timescale\upper{\subseqsamplesize}\rbra{\cEE{\dataset{\Nats}}
                \sbra{\testfun\rbra{\limparm + \incrTotal{\subseqsamplesize}(\limparm) }} - \testfun(\limparm)} \\
        & = \underbrace{\subseqsamplesize^{\timescalePow}\cEE{\dataset{\Nats}}
            \inner{\grad \testfun (\limparm)}{\incrInnov{\subseqsamplesize}}}_{\term[$\sbra{1.\innov}\upper{\subseqsamplesize}(\limparm)$]{term2:1-innov} = 0}
            + \underbrace{\subseqsamplesize^{\timescalePow}\cEE{\dataset{\Nats}}
                \inner{\grad \testfun (\limparm)}{\incrPrior{\subseqsamplesize}(\limparm)}}_{\term[$\sbra{1.\prior}\upper{\subseqsamplesize}(\limparm)$]{term2:1-prior}}
            + \underbrace{\subseqsamplesize^{\timescalePow}\cEE{\dataset{\Nats}}
                \inner{\grad \testfun (\limparm)}{\incrLoglik{\subseqsamplesize}(\limparm)}}_{\term[$\sbra{1.\loglik}\upper{\subseqsamplesize}(\limparm)$]{term2:1-loglik}}
            \\
        &\qquad
            + \underbrace{\subseqsamplesize^{\timescalePow} \cEE{\dataset{\Nats}}
                \inner{\frac{1}{2}\hess\testfun(\limparm) \incrInnov{\subseqsamplesize} }{\incrInnov{\subseqsamplesize}}}_{\term[$\sbra{2.\innov\innov}\upper{\subseqsamplesize}(\limparm)$]{term2:2-innov-innov}}
            + \underbrace{\subseqsamplesize^{\timescalePow} \cEE{\dataset{\Nats}}
                \inner{\hess\testfun(\limparm) \incrPrior{\subseqsamplesize}(\limparm) }{\incrInnov{\subseqsamplesize}}}_{\term[$\sbra{2.\prior\innov}\upper{\subseqsamplesize}(\limparm)$]{term2:2-prior-innov}=0} \\
        &\qquad
            + \underbrace{\subseqsamplesize^{\timescalePow} \cEE{\dataset{\Nats}}
                \inner{\hess\testfun(\limparm) \incrLoglik{\subseqsamplesize}(\limparm) }{\incrInnov{\subseqsamplesize}}}_{\term[$\sbra{2.\loglik\innov}\upper{\subseqsamplesize}(\limparm)$]{term2:2-loglik-innov}=0}
            + \underbrace{\subseqsamplesize^{\timescalePow}\cEE{\dataset{\Nats}}
                \inner{\frac{1}{2}\hess\testfun(\limparm) \incrPrior{\subseqsamplesize}(\limparm) }{\incrPrior{\subseqsamplesize}(\limparm)}}_{\term[$\sbra{2.\prior\prior}\upper{\subseqsamplesize}(\limparm)$]{term2:2-prior-prior}} \\
        &\qquad
            + \underbrace{\subseqsamplesize^{\timescalePow} \cEE{\dataset{\Nats}}
                \inner{\hess\testfun(\limparm) \incrLoglik{\subseqsamplesize}(\limparm) }{\incrPrior{\subseqsamplesize}}}_{\term[$\sbra{2.\loglik\prior}\upper{\subseqsamplesize}(\limparm)$]{term2:2-loglik-prior}}
            + \underbrace{\subseqsamplesize^{\timescalePow}\cEE{\dataset{\Nats}}
                \inner{\frac{1}{2}\hess\testfun(\limparm) \incrLoglik{\subseqsamplesize}(\limparm) }{\incrLoglik{\subseqsamplesize}(\limparm)}}_{\term[$\sbra{2.\loglik\loglik}\upper{\subseqsamplesize}(\limparm)$]{term2:2-loglik-loglik}} \\
        & \qquad +  \underbrace{\subseqsamplesize^{\timescalePow}\cEE{\dataset{\Nats}}\sbra{\frac{1}{6}\sbra{\grad^{\otimes 3} \testfun(\limparm+S \incrTotal{\subseqsamplesize}(\limparm))}\rbra{\incrTotal{\subseqsamplesize}(\limparm),\incrTotal{\subseqsamplesize}(\limparm),\incrTotal{\subseqsamplesize}(\limparm)}} }_{\term[$\sbra{3.R}\upper{\subseqsamplesize}(\limparm)$]{term2:3-remainder}} \\
    \]
for some $S\in[0,1]$ depending on $\testfun,\limparm, \incrTotal{\subseqsamplesize}(\limparm)$, where $\grad^{\otimes 3} \testfun(\limparm)$ is the trilinear from of third order partials of $\testfun$ at $\limparm$ (and hence is linear in each of its three arguments). Terms that are linear in $\incrInnov{\subseqsamplesize}$ have mean $0$ and can be eliminated outright, as indicated in their corresponding underbraces. Terms are labelled by the order of the term, followed by the increments that appear in the term; for example \ref*{term2:2-loglik-innov} is the second order term involving a likelihood increment and a Gaussian noise (innovation) increment.
The $R$ in \ref*{term2:3-remainder} denotes that it is the \emph{remainder}.

Recall that
\*[
    [\generator\testfun](\limparm)
        & = \underbrace{- \inner{\frac{\const{\driftName}}{2}\precon \finfoOpt \limparm}{\grad \testfun(\limparm)}}_{\term[$\sbra{\textrm{I}.\precon\finfoOpt}(\limparm)$]{term2:I}}
        + \underbrace{\frac{\const{\langNoiseName}}{2}  \aniso \frobprod \hess\testfun(\limparm)}_{\term[$\sbra{\textrm{II}.\aniso}(\limparm)$]{term2:II-aniso}}
        + \underbrace{\frac{\const{\mbNoiseName}}{2} \precon \vinfoOpt \precon' \frobprod \hess\testfun(\limparm)}_{\term[$\sbra{\textrm{II}.\precon\vinfoOpt\precon'}(\limparm)$]{term2:II-vinfoOpt}}.
\]
We have similarly labelled these terms, with the roman numeral denoting the order and the subsequent symbol denoting the coefficient matrix (up to scaling factors).
Thus, after eliminating terms which are linear in $\incrInnov{\subseqsamplesize}$, and thus have mean $0$, the difference of approximate and limiting generator applied to the test function can be expressed as
\*[
    &\hspace{-1em}
    \abs{[\generator\upper{\subseqsamplesize} \testfun](\limparm) - [\generator \testfun](\limparm)}\\
        & \leq
            \abs{\text{\ref*{term2:1-prior}}}
						+ \abs{\text{\ref*{term2:2-prior-prior}}}
		        + \abs{\text{\ref*{term2:2-loglik-prior}}}\\
        & \qquad + \abs{\text{\ref*{term2:1-loglik}} - \text{\ref*{term2:I}}} \\
        & \qquad + \abs{\text{\ref*{term2:2-innov-innov}} - \text{\ref*{term2:II-aniso}}} \\
				&\qquad
            + \abs{\text{\ref*{term2:2-loglik-loglik}} - \text{\ref*{term2:II-vinfoOpt}}} \\
        & \qquad + \abs{\text{\ref*{term2:3-remainder}}}.
\]
We will show that each of these seven terms vanish uniformly on $\compactset_1$. The first three terms listed above, those non-remainder terms with no corresponding term in the limiting generator, will be handled first. Then we will handle each of the terms which corresponds to part of the limiting generator, and lastly we will handle the remainder term.

\subsubsection{Terms that do not contribute to the limit}\label{sec:mainproof-inside-vanish}
\*[
    \abs{\text{\ref*{term2:1-prior}}}
        & = \subseqsamplesize^{\timescalePow}\abs{\cEE{\dataset{\Nats}}
            \inner{\grad \testfun (\limparm)}{\incrPrior{\subseqsamplesize}(\limparm)}} \\
        & \leq \frac{\const{\stepsize}\subseqsamplesize^{\timescalePow-\stepsizePow+\localscalePow-1} \norm{\precon}}{2} \abs{\cEE{\dataset{\Nats}}
            \inner{\grad \testfun (\limparm)}{
                \grad \log\prior\rbra{\parmLMLE{{\subseqsamplesize}}
                    + \subseqsamplesize^{-\localscalePow} \limparm}}} \\
        & \leq \frac{\const{\stepsize}\subseqsamplesize^{\timescalePow-\stepsizePow+\localscalePow-1}  \norm{\precon}}{2} \normInf{{\grad\testfun}} \rbra{\norm{\grad \log\prior\rbra{\trueparm}} +\lipschitzconst_0\rbra{\Upsilon\upper{\subseqsamplesize} + \frac{2 \radius_0 +2\const{0}}{\subseqsamplesize^{\localscalePow}}}},
\]
which vanishes uniformly on $\compactset_1$, since ${\timescalePow+\localscalePow-\stepsizePow-1}\leq \localscalePow -1 < 0$.

\*[
& \hspace{-1em}
\abs{\text{\ref*{term2:2-prior-prior}}} \\
    & = \abs{\subseqsamplesize^{\timescalePow} \cEE{\dataset{\Nats}}
        \inner{\frac{1}{2}\hess\testfun(\limparm) \incrPrior{\subseqsamplesize}(\limparm) }{\incrPrior{\subseqsamplesize}(\limparm)}} \\
    & \leq \subseqsamplesize^{\timescalePow} \normInf{{\hess\testfun}}
        \rbra{\frac{\const{\stepsize}\subseqsamplesize^{\localscalePow-\stepsizePow-1}\norm{\precon}}{2}}^2\\
    &\qquad \times\rbra{\norm{\grad\log\prior(\trueparm)} + \lipschitzconst_0 \norm{\parmLMLE{{\subseqsamplesize}} - \trueparm} + \lipschitzconst_0 \frac{2\radius_0+2\const{0}}{\subseqsamplesize^\localscalePow}}^2
\]
which vanishes uniformly since $\timescalePow +2\localscalePow - 2\stepsizePow -2 \leq (2\localscalePow -2) -h <0$ (which follows from $\stepsizePow\geq \timescalePow$ and $\localscalePow <1$).

\*[
& \hspace{-1em}
\abs{\text{\ref*{term2:2-loglik-prior}}} \\
    & = \abs{\subseqsamplesize^{\timescalePow} 2\cEE{\dataset{\Nats}}
        \inner{\frac{1}{2}\hess\testfun(\limparm) \incrLoglik{\subseqsamplesize}(\limparm) }{\incrPrior{\subseqsamplesize}(\limparm)}} \\
    & \leq 2 \subseqsamplesize^{\timescalePow} \normInf{{\hess\testfun}}
        \rbra{\frac{\const{\stepsize}\subseqsamplesize^{\localscalePow-\stepsizePow-1}\norm{\precon}}{2}}
        \rbra{\frac{\const{\stepsize}\subseqsamplesize^{\localscalePow-\stepsizePow}\norm{\precon}}{2}}\\
    &\qquad \times\rbra{\norm{\grad\log\prior(\trueparm)} + \lipschitzconst_0 \Upsilon\upper{\subseqsamplesize} + \lipschitzconst_0 \frac{2\radius_0+2\const{0}}{\subseqsamplesize^{\localscalePow}}}\\
    &\qquad \times\rbra{
        \subseqsamplesize^{1/\assmoment{2}}
        + \subseqsamplesize^{1/\assmoment{3}}\Upsilon\upper{\subseqsamplesize} + \subseqsamplesize^{1/\assmoment{3} - \localscalePow}   }\\
\]
which vanishes uniformly due to the assumptions of the relationship between $\stepsizePow, \timescalePow, \localscalePow, \assmoment{3}, \assmoment{2}$ under each assumption.

\subsubsection{Convergence of the drift term}\label{sec:mainproof-inside-drift}
Third, using that $\sum_{\dataidx\in \range{\subseqsamplesize}} \grad\loglik\rbra{ \parmLMLE{{\subseqsamplesize}};\ \data_{\dataidx}} = 0$,
\[ \label{eq:approx-gen-loglik-term}
&\hspace{-1em}
\text{\ref*{term2:1-loglik}}\\
    & = \subseqsamplesize^{\timescalePow}\cEE{\dataset{\Nats}}
        \inner{\grad \testfun (\limparm)}{\incrLoglik{\subseqsamplesize}(\limparm)} \\
    & = \cEE{\dataset{\Nats}}
        \inner{\grad \testfun (\limparm)}{\frac{\const{\stepsize}\subseqsamplesize^{\timescalePow+\localscalePow-\stepsizePow} \precon}{2\batchsize\upper{\subseqsamplesize}}\sum_{\mbidx\in \range{\batchsize\upper{\subseqsamplesize}}} \grad\loglik\rbra{ \parmLMLE{{\subseqsamplesize}}
            + \frac{1}{\subseqsamplesize^{\localscalePow}} \limparm;\ \data_{\mbdataidx{\subseqsamplesize}{1}{\mbidx}}}} \\
    & =
        \inner{\frac{\const{\stepsize} \precon\adj}{2} \grad \testfun (\limparm)}{\subseqsamplesize^{\timescalePow+\localscalePow-\stepsizePow-1}\sum_{\dataidx\in \range{\subseqsamplesize}} \grad\loglik\rbra{ \parmLMLE{{\subseqsamplesize}}
            + \frac{1}{\subseqsamplesize^{\localscalePow}} \limparm;\ \data_{\dataidx}}} \\
    & =
        \inner{\frac{\const{\stepsize} \precon\adj}{2}  \grad \testfun (\limparm)}{
            \rbra{\int_0^1 \subseqsamplesize^{\timescalePow-\stepsizePow-1}\sum_{\dataidx\in \range{\subseqsamplesize}}\hess\loglik\rbra{ \parmLMLE{{\subseqsamplesize}}+ \frac{s}{\subseqsamplesize^{\localscalePow}} \limparm;\ \data_{\dataidx}}ds}} \\
\]
Now, for all $\subseqsamplesize$ large enough that $r_{\finfo,\subseqsamplesize}\geq \radius_0 + \const{0}$
\*[
&\hspace{-1em}
 \abs{
        \inner{\frac{\const{\stepsize} \precon\adj}{2}  \grad \testfun (\limparm)}{
            \rbra{\int_0^1 \subseqsamplesize^{-1}\sum_{\dataidx\in \range{\subseqsamplesize}}\hess\loglik\rbra{ \parmLMLE{{\subseqsamplesize}}+ \frac{s}{\subseqsamplesize^{\localscalePow}} \limparm;\ \data_{\dataidx}}ds\
            +\finfoOpt} \limparm}} \\
    & \leq \const{\stepsize}\norm{\precon} \normInf{{\grad\testfun}} \rbra{\radius_0 + \const{0}}\norm{
            \int_0^1 \sbra{\subseqsamplesize^{-1}\sum_{\dataidx\in \range{\subseqsamplesize}}\hess\loglik\rbra{ \parmLMLE{{\subseqsamplesize}}+ \frac{s}{\subseqsamplesize^{\localscalePow}} \limparm;\ \data_{\dataidx}} +\finfoOpt} ds\
            }\\
		& \leq \const{\stepsize}\norm{\precon} \normInf{{\grad\testfun}} \rbra{\radius_0 + \const{0}}\cdot \Upsilon\upper{\subseqsamplesize},
\]
and thus vanishes uniformly on $\compactset_1$.

When $\timescalePow > \stepsizePow$, so $\const{\driftName}=0$ and hence $\text{\ref*{term2:I}} = 0$ (where $\text{\ref*{term2:I}}$ is the drift term appearing in the definition of the limiting generator $\generator$ in \cref{eq:lim-gen-defn}), then the drift term will be inactive in the limit.
We show this by using the fact that $\text{\ref*{term2:1-loglik}}$ is a vanishing distance from a sequence that vanishes:
\*[
    &\hspace{-1em}
    \abs{\text{\ref*{term2:1-loglik}} - \text{\ref*{term2:I}}}\\
        & \leq \subseqsamplesize^{\stepsizePow-\timescalePow}  \abs{
                \inner{\frac{\const{\stepsize} \precon\adj}{2}  \grad \testfun (\limparm)}{
                    \rbra{\int_0^1 \subseqsamplesize^{-1}\sum_{\dataidx\in \range{\subseqsamplesize}}\hess\loglik\rbra{ \parmLMLE{{\subseqsamplesize}}+ \frac{s}{\subseqsamplesize^{\localscalePow}} \limparm;\ \data_{\dataidx}}ds\
                    +\finfoOpt} \limparm}} \\
        & \qquad + \subseqsamplesize^{\stepsizePow-\timescalePow}  \abs{
                \inner{\frac{\const{\stepsize} \precon\adj}{2}  \grad \testfun (\limparm)}{
                    \finfoOpt \limparm}};
\]
and hence vanishes uniformly on $\compactset_1$.

When $\timescalePow = \stepsizePow$, then the drift term is active in the limit, and we show that $\text{\ref*{term2:1-loglik}}$ converges to the drift term from the limiting process $\text{\ref*{term2:I}}$:
\*[
&\hspace{-1em}
\abs{\text{\ref*{term2:1-loglik}} - \text{\ref*{term2:I}}}\\
    &= \subseqsamplesize^{\stepsizePow-\timescalePow}  \abs{
                \inner{\frac{\const{\stepsize} \precon\adj}{2}  \grad \testfun (\limparm)}{
                    \rbra{\int_0^1 \subseqsamplesize^{-1}\sum_{\dataidx\in \range{\subseqsamplesize}}\hess\loglik\rbra{ \parmLMLE{{\subseqsamplesize}}+ \frac{s}{\subseqsamplesize^{\localscalePow}} \limparm;\ \data_{\dataidx}}ds\
                    +\finfoOpt} \limparm}}
\]
 vanishes uniformly on $\compactset_1$.

\subsubsection{Convergence of the diffusion term corresponding to Gaussian noise}\label{sec:mainproof-inside-innov}
\*[
    &\hspace{-1em}
    \abs{\text{\ref*{term2:2-innov-innov}} - \text{\ref*{term2:II-aniso}}} \\
        & = \abs{\subseqsamplesize^{\timescalePow} \cEE{\dataset{\Nats}}
            \inner{\frac{1}{2}\hess\testfun(\limparm) \incrInnov{\subseqsamplesize} }{\incrInnov{\subseqsamplesize}} - \frac{\const{\stepsize}}{2\const{\invtemp}} \aniso \frobprod \hess\testfun(\limparm)} \\
\]
If $\timescalePow + 2\localscalePow-\stepsizePow-\invtempPow = 0$ then, the corresponding diffusion term is active in the limit.
Using the definition of $\incrInnov{\subseqsamplesize}$ and that $\invtemp\upper{\samplesize} = \const{\invtemp} \samplesize^{\invtempPow}$, $\invtemp_{\stepsize} = \const{\stepsize} \samplesize^{\stepsizePow}$, and $\invtemp_{\localscale} = \samplesize^{\localscalePow}$
\*[
    &\hspace{-1em}
    \abs{\text{\ref*{term2:2-innov-innov}} - \text{\ref*{term2:II-aniso}}} \\
        & \leq  \frac{\const{\stepsize}}{2\const{\invtemp}}\abs{\subseqsamplesize^{\timescalePow + 2\localscalePow-\stepsizePow-\invtempPow}\cEE{\dataset{\Nats}}
            \inner{\hess\testfun(\limparm)  \sqrt{\aniso}\innoviter{1} }{\sqrt{\aniso}\innoviter{1}} -  \aniso \frobprod \hess\testfun(\limparm)} \\
        & = 0
\]
If $\timescalePow + 2\localscalePow-\stepsizePow-\invtempPow<0$ then the corresponding diffusion term is inactive in the limit, and so $\const{\langNoiseName}=0$ and so $\text{\ref*{term2:II-aniso}}=0$. In that case we show that $\text{\ref*{term2:2-innov-innov}}$ vanishes uniformly.
\*[
    &\hspace{-1em}
    \abs{\text{\ref*{term2:2-innov-innov}} - \text{\ref*{term2:II-aniso}}} \\
        & \leq  \frac{\const{\stepsize}}{2\const{\invtemp}} \subseqsamplesize^{\timescalePow + 2\localscalePow-\stepsizePow-\invtempPow}\abs{\cEE{\dataset{\Nats}}
            \inner{\hess\testfun(\limparm)  \sqrt{\aniso}\innoviter{1} }{\sqrt{\aniso}\innoviter{1}}} \\
        & = \frac{\const{\stepsize}}{2\const{\invtemp}} \subseqsamplesize^{\timescalePow + 2\localscalePow-\stepsizePow-\invtempPow}\norm{\aniso}_F \normInf{\norm{\hess \testfun}_F},
\]
which vanishes uniformly.

\subsubsection{Convergence of the diffusion term corresponding to minibatch noise}\label{sec:mainproof-inside-ll}
\*[
    & \hspace{-1em}
    \abs{\text{\ref*{term2:2-loglik-loglik} - \ref*{term2:II-vinfoOpt}}} \\
        & = \bigg\vert \subseqsamplesize^{\timescalePow} \cEE{\dataset{\Nats}}
            \inner{\frac{1}{2}\hess\testfun(\limparm) \incrLoglik{\subseqsamplesize}(\limparm) }{\incrLoglik{\subseqsamplesize}(\limparm)}  - \frac{\const{\mbNoiseName}}{2}  \precon \vinfoOpt \precon' \frobprod \hess\testfun(\limparm) \bigg\vert \\
        & = \frac{1}{2}\abs{\sbra{\subseqsamplesize^{\timescalePow} \cEE{\dataset{\Nats}}\sbra{
            \rbra{\incrLoglik{\subseqsamplesize}(\limparm)}^{\otimes 2}} \frobprod \hess \testfun (\limparm) - \frac{\const{\mbNoiseName}}{2}   \precon \vinfoOpt \precon' \frobprod \hess\testfun(\limparm)} } \\
        & \leq \frac{\normInf{{\hess\testfun}_{F}}}{2}
            \norm{\sbra{\subseqsamplesize^{\timescalePow} \cEE{\dataset{\Nats}}\sbra{
            \rbra{\incrLoglik{\subseqsamplesize}(\limparm)}^{\otimes 2}} - \frac{\const{\mbNoiseName}}{2}   \precon \vinfoOpt \precon'} }_F \\
        & \leq \sqrt{\parmdim}\frac{\normInf{{\hess\testfun}_{F}}}{2}
            \norm{\sbra{\subseqsamplesize^{\timescalePow} \cEE{\dataset{\Nats}}\sbra{
            \rbra{\incrLoglik{\subseqsamplesize}(\limparm)}^{\otimes 2}} - \frac{\const{\mbNoiseName}}{2}   \precon \vinfoOpt \precon'} } \\
\]
Now,
\small
\*[
    &\hspace{-1em}
    \cEE{\dataset{\Nats}}\subseqsamplesize^{\timescalePow}\rbra{\incrLoglik{\subseqsamplesize}(\limparm)}^{\otimes 2} \\
        & = \frac{\const{\stepsize}^2 \subseqsamplesize^{\timescalePow + 2\localscalePow - 2\stepsizePow} }{4 (\batchsize\upper{\subseqsamplesize})^2} \precon \rbra{\cEE{\dataset{\Nats}}\sum_{\mbidx\in \range{\batchsize\upper{\subseqsamplesize}}} \grad\loglik\rbra{ \parmLMLE{{\subseqsamplesize}}
                + \frac{1}{\subseqsamplesize^{\localscalePow}} \limparm;\ \data_{\mbdataidx{\subseqsamplesize}{1}{\mbidx}}}^{\otimes 2}}\precon' \\
            &\quad + \frac{\const{\stepsize}^2\subseqsamplesize^{\timescalePow + 2\localscalePow - 2\stepsizePow}}{4 (\batchsize\upper{\subseqsamplesize})^2} \precon \lcrx[3]({\cEE{\dataset{\Nats}}\sum_{\mbidx\in \range{\batchsize\upper{\subseqsamplesize}}}\sum_{\mbidx'\in \range{\batchsize\upper{\subseqsamplesize}}\setminus\set{\mbidx}} \grad\loglik\rbra{ \parmLMLE{{\subseqsamplesize}}
                    + \frac{\limparm}{\subseqsamplesize^{\localscalePow}} ;\ \data_{\mbdataidx{\subseqsamplesize}{1}{\mbidx}}} }.\\
                    &\hspace{20em}\lcrx[3].{\otimes \grad\loglik\rbra{ \parmLMLE{{\subseqsamplesize}}
                        + \frac{\limparm}{\subseqsamplesize^{\localscalePow}} ;\ \data_{\mbdataidx{\subseqsamplesize}{1}{\mbidx'}}}}) \precon' \\
        & = \frac{\const{\stepsize}^2\subseqsamplesize^{\timescalePow + 2\localscalePow - 2\stepsizePow}}{4 \batchsize\upper{\subseqsamplesize}} \precon \rbra{\frac{1}{\subseqsamplesize}\sum_{\dataidx\in \range{\subseqsamplesize}} \grad\loglik\rbra{ \parmLMLE{{\subseqsamplesize}}
                    + \frac{1}{\subseqsamplesize^{\localscalePow}} \limparm;\ \data_{\dataidx}}^{\otimes 2}}\precon' \\
            &\quad + \frac{\const{\stepsize}^2\subseqsamplesize^{\timescalePow + 2\localscalePow - 2\stepsizePow}}{4 (\batchsize\upper{\subseqsamplesize})^2} \precon \lcrx[3]({\cEE{\dataset{\Nats}}\sum_{\mbidx\in \range{\batchsize\upper{\subseqsamplesize}}}\sum_{\mbidx'\in \range{\batchsize\upper{\subseqsamplesize}}\setminus\set{\mbidx}} \grad\loglik\rbra{ \parmLMLE{{\subseqsamplesize}}
                    + \frac{\limparm}{\subseqsamplesize^{\localscalePow}} ;\ \data_{\mbdataidx{\subseqsamplesize}{1}{\mbidx}}}}. \\
                &\hspace{20em}\lcrx.{\otimes \grad\loglik\rbra{ \parmLMLE{{\subseqsamplesize}}
                        + \frac{\limparm}{\subseqsamplesize^{\localscalePow}} ;\ \data_{\mbdataidx{\subseqsamplesize}{1}{\mbidx'}}}}) \precon'
\]
\normalsize
If the mini-batches are drawn with replacement, then
\*[
    &\hspace{-1em} \cEE{\dataset{\Nats}}\sum_{\mbidx\in \range{\batchsize\upper{\subseqsamplesize}}}\sum_{\mbidx'\in \range{\batchsize\upper{\subseqsamplesize}}\setminus\set{\mbidx}} \grad\loglik\rbra{ \parmLMLE{{\subseqsamplesize}}
        + \frac{1}{\subseqsamplesize^{\localscalePow}} \limparm;\ \data_{\mbdataidx{\subseqsamplesize}{1}{\mbidx}}} \otimes \grad\loglik\rbra{ \parmLMLE{{\subseqsamplesize}}
            + \frac{1}{\subseqsamplesize^{\localscalePow}} \limparm;\ \data_{\mbdataidx{\subseqsamplesize}{1}{\mbidx'}}} \\
        & = \frac{\batchsize\upper{\subseqsamplesize}(\batchsize\upper{\subseqsamplesize}-1)}{\subseqsamplesize^2}\sum_{\dataidx\in \range{\subseqsamplesize}}\sum_{\dataidx'\in \range{\subseqsamplesize}} \grad\loglik\rbra{ \parmLMLE{{\subseqsamplesize}}
            + \frac{1}{\subseqsamplesize^{\localscalePow}} \limparm;\ \data_{\dataidx}} \otimes \grad\loglik\rbra{ \parmLMLE{{\subseqsamplesize}}
                + \frac{1}{\subseqsamplesize^{\localscalePow}} \limparm;\ \data_{\dataidx'}} \\
        & = \batchsize\upper{\subseqsamplesize}(\batchsize\upper{\subseqsamplesize}-1) \rbra{\frac{1}{\subseqsamplesize}\sum_{\dataidx\in \range{\subseqsamplesize}}\grad\loglik\rbra{ \parmLMLE{{\subseqsamplesize}}
            + \frac{1}{\subseqsamplesize^{\localscalePow}} \limparm;\ \data_{\dataidx}}}^{\otimes 2}  \\
        & = \batchsize\upper{\subseqsamplesize}(\batchsize\upper{\subseqsamplesize}-1) \rbra{\frac{1}{\subseqsamplesize}\sum_{\dataidx\in \range{\subseqsamplesize}} \int_0^1 \hess\loglik\rbra{ \parmLMLE{{\subseqsamplesize}}
            + \frac{s}{\subseqsamplesize^{\localscalePow}} \limparm;\ \data_{\dataidx}} ds \frac{1}{\subseqsamplesize^{\localscalePow}} \limparm}^{\otimes 2}  \\
\]

Thus, if $\timescalePow + 2\localscalePow - 2\stepsizePow - \batchsizePow= 0$, so that $\const{\mbNoiseName}\neq 0$ and the corresponding term is active in the limit, and the minibtaches are drawn with replacement, then combining the past several equations gives:
\*[
    &\hspace{-1em}
    \abs{\text{\ref*{term2:2-loglik-loglik} - \ref*{term2:II-vinfoOpt}}} \\
        & \leq \frac{\sqrt{\parmdim}\norm{\precon}^2\normInf{{\hess\testfun}_{F}}}{2}
            \norm{\frac{\const{\stepsize}^2}{4 \const{\batchsize}} \frac{\const{\batchsize}\subseqsamplesize^\batchsizePow}{\floor{\const{\batchsize}\subseqsamplesize^\batchsizePow}}  \rbra{\frac{1}{\subseqsamplesize}\sum_{\dataidx\in \range{\subseqsamplesize}} \grad\loglik\rbra{ \parmLMLE{{\subseqsamplesize}}
                        + \frac{1}{\subseqsamplesize^{\localscalePow}} \limparm;\ \data_{\dataidx}}^{\otimes 2}} - \const{\mbNoiseName}  \vinfoOpt  } \\
        & \qquad + \frac{\sqrt{\parmdim} \const{\stepsize}^2\norm{\precon}^2\normInf{{\hess\testfun}_{F}}\subseqsamplesize^{-2\localscalePow}}{8} \norm{\rbra{\frac{1}{\subseqsamplesize}\sum_{\dataidx\in \range{\subseqsamplesize}} \int_0^1 \hess\loglik\rbra{ \parmLMLE{{\subseqsamplesize}}
            + \frac{s}{\subseqsamplesize^{\localscalePow}} \limparm;\ \data_{\dataidx}} ds  \ \limparm}^{\otimes 2}}
\]
For $\limparm \in \compactset_1$, and for all $\subseqsamplesize$ large enough that $r_{\finfo,\subseqsamplesize}\geq \radius_0 + \const{0}$
\*[
    &\hspace{-1em}
    \subseqsamplesize^{-2\localscalePow}\norm{\rbra{\frac{1}{\subseqsamplesize}\sum_{\dataidx\in \range{\subseqsamplesize}} \int_0^1 \hess\loglik\rbra{ \parmLMLE{{\subseqsamplesize}}
        + \frac{s}{\subseqsamplesize^{\localscalePow}} \limparm;\ \data_{\dataidx}} ds  \ \limparm}^{\otimes 2}}\\
        & = \subseqsamplesize^{-2\localscalePow}\norm{\frac{1}{\subseqsamplesize}\sum_{\dataidx\in \range{\subseqsamplesize}} \int_0^1 \hess\loglik\rbra{ \parmLMLE{{\subseqsamplesize}}
            + \frac{s}{\subseqsamplesize^{\localscalePow}} \limparm;\ \data_{\dataidx}} ds  \ \limparm}^2 \\
        & \leq \frac{\rbra{2\radius_0+2\const{0}}^2}{\subseqsamplesize^{2\localscalePow}} \rbra{\norm{\finfoOpt}+\Upsilon\upper{\subseqsamplesize}}^2,
\]
which vanishes uniformly.

Since the mini-batches are drawn with replacement, using the definition of $\const{\mbNoiseName}$, for all $\subseqsamplesize$ large enough that $r_{\vinfo,\subseqsamplesize}\geq \radius_0 + \const{0}$
\*[
    & \hspace{-1em}
    \norm{\frac{\const{\stepsize}^2}{4 \const{\batchsize}} \frac{\const{\batchsize}\subseqsamplesize^\batchsizePow}{\floor{\const{\batchsize}\subseqsamplesize^\batchsizePow}} \rbra{\frac{1}{\subseqsamplesize}\sum_{\dataidx\in \range{\subseqsamplesize}} \grad\loglik\rbra{ \parmLMLE{{\subseqsamplesize}}
                + \frac{1}{\subseqsamplesize^{\localscalePow}} \limparm;\ \data_{\dataidx}}^{\otimes 2}} - \const{\mbNoiseName} \vinfoOpt }  \\
        & \leq \frac{\const{\stepsize}^2}{4 \const{\batchsize}} \frac{\const{\batchsize}\subseqsamplesize^\batchsizePow}{\floor{\const{\batchsize}\subseqsamplesize^\batchsizePow}}\norm{ \rbra{\frac{1}{\subseqsamplesize}\sum_{\dataidx\in \range{\subseqsamplesize}} \grad\loglik\rbra{ \parmLMLE{{\subseqsamplesize}}
                        + \frac{1}{\subseqsamplesize^{\localscalePow}} \limparm;\ \data_{\dataidx}}^{\otimes 2}} - \vinfoOpt  } \\
        &\qquad + \abs{\frac{\const{\stepsize}^2}{4 \const{\batchsize}} \frac{\const{\batchsize}\subseqsamplesize^\batchsizePow}{\floor{\const{\batchsize}\subseqsamplesize^\batchsizePow}} -\frac{\const{\stepsize}^2}{4 \const{\batchsize}}}\norm{\vinfoOpt}\\
				& \leq \frac{\const{\stepsize}^2}{4 \const{\batchsize}} \frac{\const{\batchsize}\subseqsamplesize^\batchsizePow}{\floor{\const{\batchsize}\subseqsamplesize^\batchsizePow}}\Upsilon\upper{\subseqsamplesize}
				+ \abs{\frac{\const{\stepsize}^2}{4 \const{\batchsize}} \frac{\const{\batchsize}\subseqsamplesize^\batchsizePow}{\floor{\const{\batchsize}\subseqsamplesize^\batchsizePow}} -\frac{\const{\stepsize}^2}{4 \const{\batchsize}}}\norm{\vinfoOpt}.
\]

And, if $\timescalePow + 2\localscalePow - 2\stepsizePow - \batchsizePow< 0$ and the mini-batches are drawn with replacement, so that $\const{\mbNoiseName}=0$, and the corresponding diffusion term is inactive in the limit and $\text{\ref*{term2:II-vinfoOpt}} = 0$, then
\*[
		&\hspace{-1em}
		\abs{\text{\ref*{term2:2-loglik-loglik}} - \text{\ref*{term2:II-vinfoOpt}}} \\
				& \leq \abs{\cEE{\dataset{\Nats}}\subseqsamplesize^{\timescalePow}\rbra{\incrLoglik{\subseqsamplesize}(\limparm)}^{\otimes 2} - \subseqsamplesize^{\timescalePow + 2\localscalePow - 2\stepsizePow - \batchsizePow}\frac{\const{\stepsize}^2}{4\const{\batchsize}} \vinfoOpt}
				+ \subseqsamplesize^{\timescalePow + 2\localscalePow - 2\stepsizePow - \batchsizePow}\abs{\frac{\const{\stepsize}^2}{4\const{\batchsize}} \vinfoOpt}
\]
which vanishes uniformly by the previous arguments.

Therefore,
when the mini-batches are drawn with replacement, we find that \*[\abs{\text{\ref*{term2:2-loglik-loglik} - \ref*{term2:II-vinfoOpt}}}\] vanishes uniformly on $\compactset_1$.

If the mini-batches are drawn without replacement.
\*[
    &\hspace{-1em} \cEE{\dataset{\Nats}}\sum_{\mbidx\in \range{\batchsize\upper{\subseqsamplesize}}}\sum_{\mbidx'\in \range{\batchsize\upper{\subseqsamplesize}}\setminus\set{\mbidx}} \grad\loglik\rbra{ \parmLMLE{{\subseqsamplesize}}
        + \frac{1}{\subseqsamplesize^{\localscalePow}} \limparm;\ \data_{\mbdataidx{\subseqsamplesize}{1}{\mbidx}}} \otimes \grad\loglik\rbra{ \parmLMLE{{\subseqsamplesize}}
            + \frac{1}{\subseqsamplesize^{\localscalePow}} \limparm;\ \data_{\mbdataidx{\subseqsamplesize}{1}{\mbidx'}}} \\
        & = \frac{\batchsize\upper{\subseqsamplesize}(\batchsize\upper{\subseqsamplesize}-1)}{\subseqsamplesize(\subseqsamplesize-1)}\sum_{\dataidx\in \range{\subseqsamplesize}}\sum_{\dataidx'\in \range{\subseqsamplesize}\setminus\set{\dataidx}} \grad\loglik\rbra{ \parmLMLE{{\subseqsamplesize}}
            + \frac{1}{\subseqsamplesize^{\localscalePow}} \limparm;\ \data_{\dataidx}} \otimes \grad\loglik\rbra{ \parmLMLE{{\subseqsamplesize}}
                + \frac{1}{\subseqsamplesize^{\localscalePow}} \limparm;\ \data_{\dataidx'}} \\
        & = \frac{\batchsize\upper{\subseqsamplesize}(\batchsize\upper{\subseqsamplesize}-1)}{\subseqsamplesize(\subseqsamplesize-1)}\sum_{\dataidx\in \range{\subseqsamplesize}}\sum_{\dataidx'\in \range{\subseqsamplesize}} \grad\loglik\rbra{ \parmLMLE{{\subseqsamplesize}}
            + \frac{1}{\subseqsamplesize^{\localscalePow}} \limparm;\ \data_{\dataidx}} \otimes \grad\loglik\rbra{ \parmLMLE{{\subseqsamplesize}}
                + \frac{1}{\subseqsamplesize^{\localscalePow}} \limparm;\ \data_{\dataidx'}} \\
            &\qquad - \frac{\batchsize\upper{\subseqsamplesize}(\batchsize\upper{\subseqsamplesize}-1)}{\subseqsamplesize(\subseqsamplesize-1)}\sum_{\dataidx\in \range{\subseqsamplesize}} \grad\loglik\rbra{ \parmLMLE{{\subseqsamplesize}}
                + \frac{1}{\subseqsamplesize^{\localscalePow}} \limparm;\ \data_{\dataidx}}^{\otimes 2} \\
        & = \batchsize\upper{\subseqsamplesize}(\batchsize\upper{\subseqsamplesize}-1)\frac{\subseqsamplesize}{\subseqsamplesize-1} \rbra{\frac{1}{\subseqsamplesize}\sum_{\dataidx\in \range{\subseqsamplesize}}\grad\loglik\rbra{ \parmLMLE{{\subseqsamplesize}}
            + \frac{1}{\subseqsamplesize^{\localscalePow}} \limparm;\ \data_{\dataidx}}}^{\otimes 2} \\
            &\qquad - \frac{\batchsize\upper{\subseqsamplesize}(\batchsize\upper{\subseqsamplesize}-1)}{\subseqsamplesize(\subseqsamplesize-1)}\sum_{\dataidx\in \range{\subseqsamplesize}} \grad\loglik\rbra{ \parmLMLE{{\subseqsamplesize}}
                + \frac{1}{\subseqsamplesize^{\localscalePow}} \limparm;\ \data_{\dataidx}}^{\otimes 2},
\]
and so,
\*[
    &\hspace{-1em}
    \cEE{\dataset{\Nats}}\subseqsamplesize\rbra{\incrLoglik{\subseqsamplesize}(\limparm)}^{\otimes 2} \\
        & = \frac{\const{\stepsize}^2}{4 \batchsize\upper{\subseqsamplesize}} \frac{\subseqsamplesize-\batchsize\upper{\subseqsamplesize}}{\subseqsamplesize-1} \precon \rbra{\frac{1}{\subseqsamplesize}\sum_{\dataidx\in \range{\subseqsamplesize}} \grad\loglik\rbra{ \parmLMLE{{\subseqsamplesize}}
                    + \frac{1}{\subseqsamplesize^{\localscalePow}} \limparm;\ \data_{\dataidx}}^{\otimes 2}}\precon' \\
            &\qquad + \frac{\const{\stepsize}^2}{4 (\batchsize\upper{\subseqsamplesize})^2} \precon \rbra{\batchsize\upper{\subseqsamplesize}(\batchsize\upper{\subseqsamplesize}-1)\frac{\subseqsamplesize}{\subseqsamplesize-1} \rbra{\frac{1}{\subseqsamplesize}\sum_{\dataidx\in \range{\subseqsamplesize}}\grad\loglik\rbra{ \parmLMLE{{\subseqsamplesize}}
                + \frac{1}{\subseqsamplesize^{\localscalePow}} \limparm;\ \data_{\dataidx}}}^{\otimes 2}}\precon'\\
\]
In this case, for all $\subseqsamplesize$ large enough that $r_{\vinfo,\subseqsamplesize}\geq \radius_0 + \const{0}$
\*[
    & \hspace{-1em}
    \norm{\frac{\const{\stepsize}^2}{4 \batchsize\upper{\subseqsamplesize}} \frac{\subseqsamplesize-\batchsize\upper{\subseqsamplesize}}{\subseqsamplesize-1} \rbra{\frac{1}{\subseqsamplesize}\sum_{\dataidx\in \range{\subseqsamplesize}} \grad\loglik\rbra{ \parmLMLE{{\subseqsamplesize}}
                + \frac{1}{\subseqsamplesize^{\localscalePow}} \limparm;\ \data_{\dataidx}}^{\otimes 2}} - \const{\mbNoiseName} \vinfoOpt }  \\
        & \leq \frac{\const{\stepsize}^2}{4 \batchsize\upper{\subseqsamplesize}} \frac{\subseqsamplesize-\batchsize\upper{\subseqsamplesize}}{\subseqsamplesize-1} \norm{ \rbra{\frac{1}{\subseqsamplesize}\sum_{\dataidx\in \range{\subseqsamplesize}} \grad\loglik\rbra{ \parmLMLE{{\subseqsamplesize}}
                        + \frac{1}{\subseqsamplesize^{\localscalePow}} \limparm;\ \data_{\dataidx}}^{\otimes 2}} - \vinfoOpt  } \\
        &\qquad + \abs{\frac{\const{\stepsize}^2}{4 \batchsize\upper{\subseqsamplesize}} \frac{\subseqsamplesize-\batchsize\upper{\subseqsamplesize}}{\subseqsamplesize-1} -\const{\mbNoiseName}}\norm{\vinfoOpt} \\
				& \leq \frac{\const{\stepsize}^2}{4 \batchsize\upper{\subseqsamplesize}} \frac{\subseqsamplesize-\batchsize\upper{\subseqsamplesize}}{\subseqsamplesize-1} \Upsilon\upper{\subseqsamplesize}
				+ \abs{\frac{\const{\stepsize}^2}{4 \batchsize\upper{\subseqsamplesize}} \frac{\subseqsamplesize-\batchsize\upper{\subseqsamplesize}}{\subseqsamplesize-1} -\const{\mbNoiseName}}\norm{\vinfoOpt},
\]
Thus,
when the mini-batches are drawn without replacement, we find that \*[\abs{\text{\ref*{term2:2-loglik-loglik} - \ref*{term2:II-vinfoOpt}}}\] vanishes uniformly on $\compactset_1$.

\subsubsection{Convergence of the Remainder Term}\label{sec:mainproof-inside-remainder}
\*[
&\hspace{-1em}
\abs{\text{\ref*{term2:3-remainder}}}\\
    & = \subseqsamplesize^{\timescalePow}\cEE{\dataset{\Nats}}\sbra{\frac{1}{6}\sbra{\grad^{\otimes 3} \testfun(\limparm+S \incrTotal{\subseqsamplesize}(\limparm))}\rbra{\incrTotal{\subseqsamplesize}(\limparm),\incrTotal{\subseqsamplesize}(\limparm),\incrTotal{\subseqsamplesize}(\limparm)}} \\
    &\leq \frac{\subseqsamplesize^{\timescalePow}}{6}\normInf{{\grad^{\otimes 3} \testfun}}  \cEE{\dataset{\Nats}}\norm{\incrTotal{\subseqsamplesize}(\limparm)}^3 \\
    &\leq \frac{27 \subseqsamplesize^{\timescalePow}}{6}\normInf{{\grad^{\otimes 3} \testfun}}  \rbra{\cEE{\dataset{\Nats}}\norm{\incrInnov{\subseqsamplesize}}^3 + \cEE{\dataset{\Nats}}\norm{\incrPrior{\subseqsamplesize}(\limparm)}^3 + \cEE{\dataset{\Nats}}\norm{\incrLoglik{\subseqsamplesize}(\limparm)}^3},
\]
Now
\*[
    \cEE{\dataset{\Nats}}\norm{\incrInnov{\subseqsamplesize}}^3
        & \leq \rbra{\frac{\const{\stepsize}}{2\const{\invtemp}}\subseqsamplesize^{-\stepsizePow-\invtempPow+2\localscalePow}\norm{\aniso}}^{3/2}\cEE{\dataset{\Nats}}\norm{\innoviter{1}}^3\\
        &= \subseqsamplesize^{-3/2\, (\stepsizePow+\invtempPow-2\localscalePow)} \rbra{\frac{\const{\stepsize}}{2\const{\invtemp}}\norm{\aniso}}^{3/2} 2^{3/2} \frac{\varGamma\rbra{\frac{\parmdim+3}{2}}} {\varGamma\rbra{\frac{\parmdim}{2}}},
\]
where $\varGamma$ is the gamma function. Note that $\timescale -3/2\, (\stepsizePow+\invtempPow-2\localscalePow)\leq -1/2 \, (\stepsizePow+\invtempPow-2\localscalePow)\leq-\timescalePow/2 < 0 $

Second,
\*[
    \norm{\incrPrior{\subseqsamplesize}(\limparm)}^3
        &\leq \rbra{\frac{\const{\stepsize}\subseqsamplesize^{-\stepsizePow+\localscalePow-1}\norm{\precon}}{2}}^3\rbra{\norm{\grad\log\prior(\trueparm)} + \lipschitzconst_0 \norm{\parmLMLE{{\subseqsamplesize}} - \trueparm} + \lipschitzconst_0 \frac{2\radius_0+2\const{0}}{\subseqsamplesize^{\localscalePow}}}^3.
\]
Note that $\timescalePow -3\stepsizePow+3\localscalePow-3\leq -2\stepsizePow-3(1-\localscalePow) < 0 $.

Third,
\*[
&\hspace{-1em}
\cEE{\dataset{\Nats}}\norm{\incrLoglik{\subseqsamplesize}(\limparm)}^3\\
    & \leq
    \rbra{\frac{\const{\stepsize}\subseqsamplesize^{-\stepsizePow+\localscalePow}\norm{\precon}}{2}}^3\rbra{
            \subseqsamplesize^{1/\assmoment{2}}
            + \subseqsamplesize^{1/\assmoment{3}}\Upsilon\upper{\subseqsamplesize} + \subseqsamplesize^{1/\assmoment{3} - \localscalePow}}^3\\
    & \leq
    \rbra{\frac{\const{\stepsize}\norm{\precon}}{2}}^3\rbra{
          \subseqsamplesize^{1/\assmoment{2} -\stepsizePow+\localscalePow}
            + \subseqsamplesize^{1/\assmoment{3} -\stepsizePow+\localscalePow}\Upsilon\upper{\subseqsamplesize} + \subseqsamplesize^{1/\assmoment{3} -\stepsizePow}}^3\\
\]
Therefore, $\abs{\text{\ref*{term2:3-remainder}}}$ vanishes uniformly.
\qedsymbol

\section{Proof of {\cref{cor:sgld-bvm}}}

\begin{proof}[Proof of \cref{cor:sgld-bvm}]
    To verify that that the stationary measures, $\statmeas\upper{\subseqsamplesize}$ of $\semigroup\upper{\subseqsamplesize}$ converge weakly in probability to $\statmeas$,
    we need to verify that every sub-subsequence $\statmeas\upper{\subsubseqsamplesize}$ has a sub-sub-subsequence $\statmeas\upper{\subsubsubseqsamplesize}$ converging weakly to $\statmeas$ almost surely. Since weak convergence of probability measures is metrizable, then applying \cref{lem:conv-prob-subseq} yields the desired result.

    By the second part of \cref{thm:sgld-scaling-limit-multiple-scalings-boundary}, every sub-subsequence of $\rbra[1]{\semigroup\upper{\subseqsamplesize}}_{\subseqidx\in\Nats}$, $\rbra[1]{\semigroup\upper{\subsubseqsamplesize}}_{\subsubseqidx\in\Nats}$,
    has a further sub-sub-subsequence, $\rbra[1]{\semigroup\upper{\subsubsubseqsamplesize}}_{\subsubsubseqidx\in\Nats}$, such that
    with probability $1$,
    $\semigroup\upper{\subsubsubseqsamplesize}_{t} \stk\to{s} \semigroup_{t}$ on $\contbddfuns(\Reals^\parmdim)$ for all $t>0$. 
     
    Applying \citet[Part 4, Theorem~9.10]{ethier2009markov}, we have that every weak limit of $\smash{\set[1]{\statmeas\upper{\subsubsubseqsamplesize}}_{\subsubsubseqidx\in\Nats}}$ is stationary for $\semigroup$.
    As a consequence of the assumption that the spectrum of $\precon\finfo(\trueparm)$ is a subset of $\set{x\in\Complex \stT \Re(x)>0}$, $\semigroup$ has a unique stationary distribution (see, for example, \citet[]{karatzas2014brownian}), $\statmeas = \normaldist(0, \statcov)$. Thus every weak limit of $\smash{\set[1]{\statmeas\upper{\subsubsubseqsamplesize}}_{\subsubsubseqidx\in\Nats}}$ must be $\statmeas$.
    
    Since $\set[1]{\statmeas\upper{\subseqsamplesize}}_{\subseqidx\in\Nats}$ is assumed to be tight, then all of its sub-subsequences have a weakly converging sub-sub-subsequence, concluding the proof.
\end{proof}

\section{Sufficient conditions for {\cref{ass:J}} and {\cref{ass:V}}}
\label{sec:assVJ-suff}\label{SEC:ASSVJ-SUFF}

In this section we provide some sufficient conditions that ensure \cref{ass:J,ass:V}. For each of the two assumptions, we one sufficient condition based on convergence of the corresponding information matrix empirical process, one sufficient condition based on equicontinuity of the derivatives of the likelihood function, and one sufficient condition based expected Lipschitz or local Lipschitz constants for the derivatives of the likelihood.

\begin{proposition}[Sufficient conditions for {\cref{ass:J}}]
	\label{prop:assJ-suff}
	Each of the following imply \cref{ass:J}.
	\begin{enumerate}
			\item[a)]
			there exists a $\delta_1>0$ with $\sup_{\parmdum\in \ball_{\delta_1}(\trueparm)} \norm{\frac{1}{\samplesize} \sum_{\dataidx\in\range{\samplesize}} \hess\loglik(\parmdum; \data_\dataidx) + \finfo(\parmdum)}\stk\to{p} 0$ and $\finfo$ is continuous at $\trueparm$,

			\item[b)]
			$\set{\hess\loglik(\cdot;\dumdata)\ \vert \ \dumdata\in\dataspace}$ is equicontinuous at $\trueparm$,

			\item[c)]
			there exists a $\delta_1>0$ with
			\*[
					\EE \sbra{ \sup_{\parm\in \ball_{\delta_1}(\trueparm)}{\frac{\norm{\grad^{\otimes 2}\loglik(\parm;\data_1) -\grad^{\otimes 2}\loglik(\trueparm;\data_1)}}{\norm{\parm - \trueparm}}}}
					& <\infty,
			\]
	\end{enumerate}
\end{proposition}

\begin{proof}[Proof of \cref{prop:assJ-suff}] \
	
	\begin{enumerate}
		\item[a)] Let $r_{\finfo,\samplesize} = \delta_1\samplesize^{\localscalePow/2} /2$. 
		Then $\ball\rbra{\parmMLE{\samplesize},r_{\finfo,\samplesize}/\samplesize^\localscalePow} \subseteq \ball\rbra{\parmMLE{\samplesize},\delta_1/2}$. 
		
		Given that $\parmMLE{\samplesize}\stk\to{p}\trueparm$, any subsequence of indices $\subseqsamplesize$ has a further sub-subsequence of indices $\subsubseqsamplesize$ where both $\parmMLE{\subsubseqsamplesize}\to\trueparm$ 
		and 
		\*[
			\sup_{\parmdum\in \ball_{\delta_1}(\trueparm)} \norm[2]{\frac{1}{\subsubseqsamplesize} \sum_{\dataidx\in\range{\subsubseqsamplesize}} \hess\loglik(\parmdum; \data_\dataidx) + \finfo(\parmdum)}\to 0 \text{ a.s.}
		\]
		Then there is a $\subsubseqidx_0$ such that if $\subsubseqidx\geq\subsubseqidx_0$ then $\norm{\parmMLE{\subsubseqsamplesize}-\trueparm}\leq \delta_1/2$. Therefore if $\subsubseqidx\geq\subsubseqidx_0$ then $\ball\rbra{\parmMLE{\subsubseqsamplesize},r_{\finfo,\samplesize}/\subsubseqsamplesize^\localscalePow} \subseteq \ball\rbra{\trueparm,\delta_1}$.
		
		Thus, for $\subsubseqidx\geq\subsubseqidx_0$,
		\*[
		&\hspace{-1em}\sup_{\parm\in\ball\rbra{\parmMLE{\subsubseqsamplesize},r_{\finfo,\samplesize}/\subsubseqsamplesize^\localscalePow}}\norm{\finfoEmp{\subsubseqsamplesize}(\parm) - \finfo(\trueparm)} \\
			& \leq \sup_{\parm\in\ball\rbra{\parmMLE{\subsubseqsamplesize},r_{\finfo,\samplesize}/\subsubseqsamplesize^\localscalePow}}\norm{\finfoEmp{\subsubseqsamplesize}(\parm) - \finfo(\parm)} 
			+ \sup_{\parm\in\ball\rbra{\parmMLE{\subsubseqsamplesize},r/\subsubseqsamplesize^\localscalePow}}\norm{\finfo(\parm) - \finfo(\trueparm)} \\ 
			& \leq \sup_{\parm\in\ball\rbra{\trueparm,\delta_1}}\norm{\finfoEmp{\subsubseqsamplesize}(\parm) - \finfo(\parm)} + \sup_{\parm\in\ball\rbra{\parmMLE{\subsubseqsamplesize},\delta_1/\subsubseqsamplesize^{\localscalePow/2}}}\norm{\finfo(\parm) - \finfo(\trueparm)}\\
			& \leq \sup_{\parm\in\ball\rbra{\trueparm,\delta_1}}\norm{\finfoEmp{\subsubseqsamplesize}(\parm) - \finfo(\parm)} + \sup_{\parm\in\ball\rbra{\trueparm,\norm{\parmMLE{\subsubseqsamplesize}-\trueparm}+\delta_1/\subsubseqsamplesize^{\localscalePow/2}}}\norm{\finfo(\parm) - \finfo(\trueparm)}\\
			&\stk\to{a.s.} 0.
		\]
		Therefore, every subsequence of $S_\samplesize = \sup_{\parm\in\ball\rbra{\parmMLE{\samplesize},r_{\finfo,\samplesize}/\samplesize^\localscalePow}}\norm{\finfoEmp{\samplesize}(\parm) - \finfo(\trueparm)}$ has a further sub-subsequence converging almost surely to $0$, and hence $S_\samplesize$ converges in probability to $0$.
		\item[b)] Equicontinuity implies there is a function $\equicontinuityfun_{\finfoOpt}: \PosReals \to \PosReals$ with $\lim_{t\to0}\equicontinuityfun_{\finfoOpt}(t) = 0$, and
		\*[
		    \sup_{\dumdata\in\dataspace}\sup_{\limparm\in B_\delta(\trueparm)} \norm{ \hess\loglik(\limparm;\dumdata) - \hess\loglik(\trueparm;\dumdata)} \leq \equicontinuityfun_{\finfoOpt}(\delta).
		\]
		Let $r_{\finfo,\samplesize} = \samplesize^{\localscalePow/2}$. 
		Then
		\*[
		&\hspace{-1em}\sup_{\parm\in\ball\rbra{\parmMLE{\samplesize},r_{\finfo,\samplesize}/\samplesize^\localscalePow}}\norm{\finfoEmp{\samplesize}(\parm) - \finfo(\trueparm)} \\
			& \leq \sup_{\parm\in\ball\rbra{\parmMLE{\samplesize},\samplesize^{-\localscalePow/2}}}\norm{\finfoEmp{\samplesize}(\parm) - \finfoEmp{\samplesize}(\trueparm)} 
			+ \norm{\finfoEmp{\samplesize}(\trueparm) - \finfo(\trueparm)
			} \\ 
			& \leq \sup_{\parm\in\ball\rbra{\trueparm,\norm{\parmMLE{\samplesize}-\trueparm}+\samplesize^{-\localscalePow/2}}}\norm{\finfoEmp{\samplesize}(\parm) - \finfoEmp{\samplesize}(\trueparm)} 
			+ \norm{\finfoEmp{\samplesize}(\trueparm) - \finfo(\trueparm)}\\
			& \leq \equicontinuityfun_{\finfoOpt}\rbra{\norm{\parmMLE{\samplesize}-\trueparm}+\samplesize^{-\localscalePow/2}}
			+ \norm{\finfoEmp{\samplesize}(\trueparm) - \finfo(\trueparm)}\\
			&\stk\to{p} 0.
		\]
		In the last step we used that the first term vanishes in probability because $\parmMLE{\samplesize}\stk\to{p}\trueparm$, and the second term vanishes in probability by the weak law of large numbers.
		\item[c)] Let 
			\*[
				Q_\samplesize 
					 & = \frac{1}{\samplesize} \sum_{\dataidx\in\range{\samplesize}}\sbra{ \sup_{\parm\in \ball_{\delta_1}(\trueparm)}{\frac{\norm{\grad^{\otimes 2}\loglik(\parm;\data_\dataidx) -\grad^{\otimes 2}\loglik(\trueparm;\data_\dataidx)}}{\norm{\parm - \trueparm}}}}, \andT\\
				q 
					& = \EE \sbra{ \sup_{\parm\in \ball_{\delta_1}(\trueparm)}{\frac{\norm{\grad^{\otimes 2}\loglik(\parm;\data_1) -\grad^{\otimes 2}\loglik(\trueparm;\data_1)}}{\norm{\parm - \trueparm}}}}.
			\]
		By the weak law of large numbers, $Q_\samplesize \stk\to{p} q$ and $\finfoEmp{\subsubseqsamplesize}(\trueparm) \stk\to{p} \finfo(\trueparm)$.
		Let $r_{\finfo,\samplesize} = \delta_1\samplesize^{\localscalePow/2}/2$.
		As in part a), given that $\parmMLE{\samplesize}\stk\to{p}\trueparm$, any subsequence of indices $\subseqsamplesize$ has a further sub-subsequence of indices $\subsubseqsamplesize$ where both $\parmMLE{\subsubseqsamplesize}\to\trueparm$, 
		$Q_{\subsubseqsamplesize}\to q $, and $\finfoEmp{\subsubseqsamplesize}(\trueparm) \to \finfo(\trueparm)$ almost surely. 
		Then there is a $\subsubseqidx_0$ such that if $\subsubseqidx\geq\subsubseqidx_0$ then $\norm{\parmMLE{\subsubseqsamplesize}-\trueparm}\leq \delta_1/2$. Therefore if $\subsubseqidx\geq\subsubseqidx_0$ then $\ball\rbra{\parmMLE{\subsubseqsamplesize},r_{\finfo,\samplesize}/\subsubseqsamplesize^\localscalePow} \subseteq \ball\rbra{\trueparm,\delta_1}$.
		
		Thus, for $\subsubseqidx\geq\subsubseqidx_0$,
		\*[
    &\hspace{-1em}
    \sup_{\parm\in\ball\rbra{\parmMLE{\subsubseqsamplesize},r_{\finfo,\samplesize}/\subsubseqsamplesize^\localscalePow}}\norm{\finfoEmp{\subsubseqsamplesize}(\parm) -\finfoOpt} \\
			& \leq \norm{\finfoEmp{\subsubseqsamplesize}(\trueparm) - \finfo(\trueparm)} 
			+ \sup_{\parm\in\ball\rbra{\parmMLE{\subsubseqsamplesize},\delta_1\subsubseqsamplesize^{-\localscalePow/2}/2}}\norm{\finfoEmp{\subsubseqsamplesize}(\parm) - \finfoEmp{\subsubseqsamplesize}(\trueparm)} \\ 
		  &\leq \norm{\finfoEmp{\subsubseqsamplesize}(\trueparm) - \finfo(\trueparm)} \\
				&\qquad + \rbra{\norm{\parmMLE{\subsubseqsamplesize}-\trueparm} + \delta_1\subsubseqsamplesize^{-\localscalePow/2}/2} \sup_{\parm\in\ball\rbra{\parmMLE{\subsubseqsamplesize},\delta_1\subsubseqsamplesize^{-\localscalePow/2}/2}}\frac{\norm{\finfoEmp{\subsubseqsamplesize}(\parm) - \finfoEmp{\subsubseqsamplesize}(\trueparm)}}{\norm{\parm - \trueparm}} \\   
			&\leq \norm{\finfoEmp{\subsubseqsamplesize}(\trueparm) - \finfo(\trueparm)} \\
				&\qquad + \rbra{\norm{\parmMLE{\subsubseqsamplesize}-\trueparm} + \delta_1\subsubseqsamplesize^{-\localscalePow/2}/2} \sup_{\parm\in\ball\rbra{\trueparm,\delta_1}}\frac{\norm{\finfoEmp{\subsubseqsamplesize}(\parm) - \finfoEmp{\subsubseqsamplesize}(\trueparm)}}{\norm{\parm - \trueparm}} \\   	
			&\leq \norm{\finfoEmp{\subsubseqsamplesize}(\trueparm) - \finfo(\trueparm)} \\
				&\qquad + \rbra{\norm{\parmMLE{\subsubseqsamplesize}-\trueparm} + \delta_1\subsubseqsamplesize^{-\localscalePow/2}/2} \sup_{\parm\in\ball\rbra{\trueparm,\delta_1}} \frac{1}{\subsubseqsamplesize}\sum_{\dataidx\in\range{\subsubseqsamplesize}}\sbra{ {\frac{\norm{\grad^{\otimes 2}\loglik(\parm;\data_\dataidx) -\grad^{\otimes 2}\loglik(\trueparm;\data_\dataidx)}}{\norm{\parm - \trueparm}}}}\\
			&\leq \norm{\finfoEmp{\subsubseqsamplesize}(\trueparm) - \finfo(\trueparm)} + \rbra{\norm{\parmMLE{\subsubseqsamplesize}-\trueparm} + \delta_1\subsubseqsamplesize^{-\localscalePow/2}/2} Q_{\subsubseqsamplesize}
			\stk\to{a.s.} 0 					   
		\]		 
			Therefore, every subsequence of $S_\samplesize = \sup_{\parm\in\ball\rbra{\parmMLE{\samplesize},r_{\finfo,\samplesize}/\samplesize^\localscalePow}}\norm{\finfoEmp{\samplesize}(\parm) - \finfo(\trueparm)}$ has a further sub-subsequence converging almost surely to $0$, and hence $S_\samplesize$ converges in probability to $0$.
	\end{enumerate}
	\end{proof}

\begin{proposition}[Sufficient conditions for {\cref{ass:V}}]
	\label{prop:assV-suff}
	Each of the following imply \cref{ass:V}.
	\begin{enumerate}    
	\item[a)]
	there exists a $\delta_2>0$ with $\sup_{\parmdum\in \ball_{\delta_2}(\trueparm)} \norm{\frac{1}{\samplesize} \sum_{\dataidx\in\range{\samplesize}} \grad\loglik(\parmdum; \data_\dataidx)^{\otimes 2} - \vinfo(\parmdum)}\stk\to{p} 0$
	and $\vinfo$ is continuous at $\trueparm$,

	\item[b)]
	$\set{\grad\loglik(\cdot;\dumdata)^{\otimes2} \ \vert \ \dumdata\in\dataspace}$ is equicontinuous at $\trueparm$,

	\item[c)] $\EE \sbra{ \normInf{{\hess\loglik(\cdot;\data_1)}}^{2}} <\infty$,
	\end{enumerate}
\end{proposition}

\begin{proof}[Proof of \cref{prop:assV-suff}] \
	
	\begin{enumerate}    
	\item[a), b)] The proofs are the same as for \cref{prop:assJ-suff} a), b).
	\item[c)] 	
	Let $Q_\samplesize = \frac{1}{\samplesize}\sum_{\dataidx\in\range{\samplesize}}\normInf{\hess\loglik(\cdot; \data_\dataidx)}^2$, $q = \EE \normInf{\hess\loglik(\cdot; \data_1)}^2$, and let $r_{\vinfo,\samplesize} = \samplesize^{\localscalePow/2}$. By the weak law of large numbers, $Q_\samplesize\stk\to{p} q$, and
	$\vinfoEmp{\samplesize}(\trueparm) \stk\to{p} \vinfo(\trueparm)$.
	Starting with
	\*[
	\sup_{\parm\in\ball\rbra{\parmMLE{\samplesize},r_{\vinfo,\samplesize}/\samplesize^\localscalePow}}\norm{\vinfoEmp{\samplesize}(\parm) -\vinfoOpt}
		& \leq \norm{\vinfoEmp{\samplesize}(\trueparm) - \vinfo(\trueparm)} 
		+ \sup_{\parm\in\ball\rbra{\parmMLE{\samplesize},\samplesize^{-\localscalePow/2}}}\norm{\vinfoEmp{\samplesize}(\parm) - \vinfoEmp{\samplesize}(\trueparm)},
	\]
	we can bound the second term with a Taylor series and Cauchy-Schwartz as
	\*[
	&\hspace{-1em}\norm{\vinfoEmp{\samplesize}(\parm) - \vinfoEmp{\samplesize}(\trueparm)}\\
		&\leq \frac{1}{\samplesize} \sum_{\dataidx\in\range{\samplesize}}	\norm[2]{\rbra[2]{\grad\loglik\rbra{\trueparm;\ \data_{\dataidx}} + \int_0^1 \hess\loglik\rbra{\trueparm + s(\parm-\trueparm) ;\ \data_{\dataidx}} ds\ (\parm-\trueparm) }^{\otimes 2} - \grad\loglik\rbra{\trueparm;\ \data_{\dataidx}}^{\otimes 2}} \\
		&\leq \frac{2}{\samplesize} \sum_{\dataidx\in\range{\samplesize}}	\norm{\grad\loglik\rbra{\trueparm;\ \data_{\dataidx}}}\norm[2]{\int_0^1 \hess\loglik\rbra{\trueparm + s(\parm-\trueparm) ;\ \data_{\dataidx}} ds\ (\parm-\trueparm) } \\
		&\qquad + \frac{1}{\samplesize} \sum_{\dataidx\in\range{\samplesize}}	\norm[2]{\rbra{\int_0^1 \hess\loglik\rbra{\trueparm + s(\parm-\trueparm) ;\ \data_{\dataidx}} ds\ (\parm-\trueparm) }^{\otimes 2}} \\
		&\leq \frac{2}{\samplesize} \sum_{\dataidx\in\range{\samplesize}}	\norm{\grad\loglik\rbra{\trueparm;\ \data_{\dataidx}}}\normInf{\hess\loglik(\cdot; \data_\dataidx)}\norm{\parm-\trueparm } 
		+ \frac{1}{\samplesize} \sum_{\dataidx\in\range{\samplesize}}	\normInf{\hess\loglik(\cdot; \data_\dataidx)}^2 \norm{\parm-\trueparm }^2 \\
		&\leq 2\norm{\parm-\trueparm }  \sqrt{\frac{1}{\samplesize} \sum_{\dataidx\in\range{\samplesize}}	\norm{\grad\loglik\rbra{\trueparm;\ \data_{\dataidx}}}^2} \sqrt{\frac{1}{\samplesize} \sum_{\dataidx\in\range{\samplesize}} L(\data_\dataidx)^2}
		+ \norm{\parm-\trueparm }^2Q_\samplesize \\
		&\leq 2\norm{\parm-\trueparm }  \sqrt{\tr(\vinfoEmp{\samplesize}(\trueparm)) } \sqrt{Q_\samplesize}
		+ \norm{\parm-\trueparm }^2Q_\samplesize,
	\]
		Plugging this back in,				   
		\*[
		&\hspace{-1em}
		\sup_{\parm\in\ball\rbra{\parmMLE{\samplesize},r_{\vinfo,\samplesize}/\samplesize^\localscalePow}}\norm{\vinfoEmp{\samplesize}(\parm) -\vinfoOpt} \\
			& \leq \norm{\vinfoEmp{\samplesize}(\trueparm) - \vinfo(\trueparm)} 
			+ \sup_{\parm\in\ball\rbra[2]{\parmMLE{\samplesize},\samplesize^{-\localscalePow/2}}}\rbra{2\norm{\parm-\trueparm }  \sqrt{\tr(\vinfoEmp{\samplesize}(\trueparm)) } \sqrt{Q_\samplesize}
			+ \norm{\parm-\trueparm }^2Q_\samplesize}\\ 
			& \leq \norm{\vinfoEmp{\samplesize}(\trueparm) - \vinfo(\trueparm)} 
			+ 2\rbra[2]{\norm{\parmMLE{\samplesize}-\trueparm }+ \samplesize^{-\localscalePow/2}}  \sqrt{\tr(\vinfoEmp{\samplesize}(\trueparm)) } \sqrt{Q_\samplesize}
			+ \rbra[2]{\norm{\parmMLE{\samplesize}-\trueparm }+ \samplesize^{-\localscalePow/2}}^2Q_\samplesize\\
			&\stk\to{p}0.
		\] 
\end{enumerate}
\end{proof}

\section{Proof of Proposition {\ref{prop:averaged-covariance}}} \label{proof:prop:averaged-covariance}

Recall that
\[
d \localparm_{t}
&= -\frac{1}{2} \driftMat \localparm_{t} \,dt + \vphantom{\frac{1}{2}}\sqrt{\diffusionMat} \, d\weiner_{t},
\]
which implies
\[
\localparm_{t} =\exp(-B/2)\localparm_{0}+\int_0^t\exp(-B(t-s)/2)A^{1/2}d\weiner_{s}.
\]
Assuming stationarity, $\localparm_{t}\sim \mathcal{N}(0,Q_{\infty})$ where $Q_{\infty}=\int_0^{\infty}\exp(-Bs/2)A\exp(-Bs/2)ds$, we have
\[
\Cov \rbra{\int_{0}^{t} \limparm_{s} \, ds}
&=\EE\rbra{\int_0^t\int_0^t \localparm_{s}\localparm_{r}^T ds dr}
=\int_0^t\int_0^s \EE(\localparm_{s}\localparm_{r}^T) dr ds +\int_0^t\int_0^r \EE(\localparm_{s}\localparm_{r}^T) ds dr.
\]
We focus on the first term since the second term can be written similarly:
\[
\int_0^t\int_0^s \EE(\localparm_{s}\localparm_{r}^T) dr ds
&=\int_0^t\int_0^s \EE\left[\left(\exp(-B(s-r)/2)\localparm_{r}+\int_r^s\exp(-B(s-u)/2)A^{1/2} d\weiner_u\right)\localparm_{r}^T\right] dr ds\\
&=\int_0^t\int_0^s \exp(-B(s-r)/2)\EE(\localparm_{r}\localparm_{r}^T)dr ds\\
&=\int_0^t\int_0^s \exp(-B(s-r)/2)Q_{\infty}dr ds\\
&=\int_0^t -2B^{-1} (\exp(-Bs/2)-1)Q_{\infty} ds\\
&=\left[4B^{-2}(\exp(-Bt/2)-1)+2tB^{-1}\right]Q_{\infty}.
\]
We can write $\int_0^t\int_0^r \EE(\localparm_{s}\localparm_{r}^T) ds dr$ similarly and combine the two results
\[
\Cov\rbra{\limparmavg_{t}}=\frac{1}{t^2}\Cov (\int_{0}^{t} \limparm_{s} \, ds)
&=\frac{1}{t^2}\left[\int_0^t\int_0^s \EE(\localparm_{s}\localparm_{r}^T) dr ds +\int_0^t\int_0^r \EE(\localparm_{s}\localparm_{r}^T) ds dr\right]\\
&=  \frac{4}{t}\Sym\rbra{\driftMat^{-1}\statcov} - \frac{8}{t^{2}}\Sym\rbra{\driftMat^{-2}\cbra{I - e^{-t \driftMat/2}}\statcov},
\]
which verifies \cref{eq:continuous-time-IA}.

Using Taylor's theorem and the assumption that $-\driftMat$ is Hurwitz, we obtain
\[
e^{-t \driftMat/2} - I
&= \sum_{k=1}^{\ell} \frac{1}{k!}\rbra{\frac{-t \driftMat}{2}}^{k} + R_{\ell}(t),
\]
where $\norm{R_{\ell}(t)} \le \frac{\norm{t \driftMat/2}^{\ell+1}}{(\ell+1)!}$.
Taking $\ell = 3$ yields
\[
\lefteqn{\frac{4}{t}\driftMat^{-1}\diffusionMat\driftMat^{-\top} - \frac{8}{t^{2}}\Sym\rbra{\driftMat^{-2}\cbra{I - e^{-t \driftMat/2}}\statcov}} \\
&= \frac{4}{t}\driftMat^{-1}\diffusionMat\driftMat^{-\top} + \frac{8}{t^{2}}\cbra{-\rbra{\frac{t}{2}} \Sym\rbra{\driftMat^{-1}\statcov} + \frac{1}{2}\rbra{\frac{t}{2}}^{2}\statcov - \frac{1}{6} \rbra{\frac{t}{2}}^{3} \Sym\rbra{\driftMat\statcov} + \Sym\rbra{\driftMat^{-2} {R}_{3}(t) \statcov}} \\
&= \statcov - \frac{t}{6} \diffusionMat + \tilde{R}_{3}(t),
\]
where $\norm{\tilde{R}_{3}(t)} \le \frac{t^{2}}{48}\norm{\driftMat}^{4}\norm{\driftMat^{-2}\statcov}$,
and we have used that $\Sym\rbra[0]{\driftMat^{-1}\statcov} = \driftMat^{-1}\diffusionMat\driftMat^{-\top}$ and $\Sym\rbra{\driftMat \statcov} = \diffusionMat$, and that $B^{-1}$ and $R_{\ell}(t)$ commute.

For any $t > 0$, we have
\[
\norm{\frac{8}{t^{2}}\Sym\rbra{\driftMat^{-2}\cbra{I - e^{-t \driftMat/2}}\statcov}}
& = \norm{\frac{8}{t^{2}}\Sym\rbra{\cbra{I - e^{-t \driftMat/2}}\driftMat^{-2}\statcov}} \\
&\le \frac{8}{t^{2}} \norm{\Sym\rbra{\driftMat^{-2}\statcov}}
\le \frac{8}{t^{2}} \norm{\driftMat^{-2}\statcov},
\]
which is small when $t \gg 3\norm{\driftMat^{-2}\statcov}^{1/2}$.

\section{Proof of {\cref{cor:iteravg}}}
\begin{proof}
For \cref{eq:asymptotic-iteravg-cov}, we have
\*[
\avgcov{\iternum}\upper{\samplesize}
	&= \Cov\rbra{\iteravg{\samplesize}{\floor{\numpasses\samplesize/\batchsize\upper{\samplesize}}}}
	\approx \frac{1}{(\localscale\upper{\samplesize})^{2}}\Cov\rbra{\limparmavg_{\numpasses\samplesize/(\batchsize\upper{\samplesize}\timescale\upper{\samplesize})}} \\
	&=  \frac{4}{\numpasses}\frac{\timescale\upper{\samplesize}\batchsize\upper{\samplesize}}{\samplesize(\localscale\upper{\samplesize})^{2}}\Sym\rbra{\cbra{\const{\stepsize}\precon\finfoOpt}^{-1}\statcov} \\
	&\phantom{=~} -  \frac{8}{\numpasses^{2}}\frac{(\timescale\upper{\samplesize}\batchsize\upper{\samplesize})^{2}}{(\samplesize\localscale\upper{\samplesize})^{2}}\Sym\rbra{\cbra{\const{\stepsize}\precon\finfoOpt}^{-2}\cbra{I - \exp\sbra{-\frac{\const{\stepsize}\numpasses\samplesize}{2\batchsize\upper{\samplesize}\timescale\upper{\samplesize}}\precon\finfoOpt}\statcov}}.
\]
Now, given $\batchsizePow + \stepsizePow \le \invtempPow$,
\*[
	\lim_{\samplesize\to\infty}\samplesize\avgcov{\iternum}\upper{\samplesize}
	&= \frac{4 \const{\batchsize}}{\numpasses} \Sym\rbra{\cbra{\const{\stepsize}\precon\finfoOpt}^{-1}\statcov} \\
	 &\qquad - \ind{\batchsizePow + \stepsizePow = 1} \frac{8\const{\batchsize}^2}{\numpasses^{2}} \Sym\rbra{\sbra{\const{\stepsize}\precon\finfoOpt}^{-2}\sbra{I - e^{-\frac{\const{\stepsize}\numpasses}{2\const{\batchsize}}\precon\finfoOpt}}\statcov}\bigg\}
\]
The rest follows by combining this with \cref{prop:averaged-covariance} and the simplifications following it, and by noting that since $\stepsizePow+\batchsizePow\leq 1$ and $\stepsizePow>0$ we must have $\batchsizePow<1$, and hence $\overline{\const{\batchsize}}=1$.
\end{proof}

\section{Sketch Proof of Scaling Limit for SGLD with Control Variates}
\label{apx:sketch-cv}
We argue that the mini-batch noise is always lower order for SGLD with control variates.
In SGLD-FP, the stochastic gradient $\grad\loglik\rbra{\theta;\ \data_I}$ is replaced by $\grad\loglik\rbra{\theta;\ \data_I} - \grad\loglik\rbra{\trueparm;\ \data_I}$. By construction this stochastic gradient is still unbiased, but its significantly lower variance leads to materially different behaviour in the asymptotic analysis. Specifically, the corresponding \ref{term2:2-loglik-loglik} from the proof of \cref{thm:sgld-scaling-limit-multiple-scalings-boundary-simple} in \cref{sec:pf:sgld-scaling-limit-multiple-scalings} is vanishing under any scaling limit where the drift term $\term[$\sbra{1.\loglik}$]{term1:1-loglik}$ does not vanish.

\*[
& \hspace{-1em}
\underbrace{\subseqsamplesize^{\timescalePow}\cEE{\dataset{\Nats}}
		\inner{\frac{1}{2}\hess\testfun(\limparm) \incrLoglik{\subseqsamplesize}(\limparm) }{\incrLoglik{\subseqsamplesize}(\limparm)}}_{\term[$\sbra{2.\loglik\loglik}\upper{\subseqsamplesize}(\limparm)$]{term2:2-loglik-loglik-cv}} \\
		& = \subseqsamplesize^{\timescalePow}\cEE{\dataset{\Nats}}
		\frac{1}{2}\hess\testfun(\limparm) \frobprod
				\rbra{ \incrLoglik{\subseqsamplesize}(\limparm) }^{\otimes2} \\
		& = \subseqsamplesize^{\timescalePow}\cEE{\dataset{\Nats}}
		\frac{1}{2}\precon\hess\testfun(\limparm) \frobprod
				\rbra{ \frac{\stepsize\localscale\upper{\samplesize} \precon}{2\batchsize\upper{\samplesize}} \ \sum_{\mbidx\in \range{\batchsize\upper{\samplesize}}} \rbra{\grad\loglik\rbra{\parmLMLE{{\samplesize}}
						+ (\localscale\upper{\samplesize})^{-1} \limparm;\ \data_{\mbdataidx{\samplesize}{1}{\mbidx}}}
						-\grad\loglik\rbra{\parmLMLE{{\samplesize}}
								;\ \data_{\mbdataidx{\samplesize}{1}{\mbidx}} }} }^{\otimes2} \\
		& = \frac{\const{\stepsize}^2}{\const{\batchsize}^2}\subseqsamplesize^{\timescalePow-2\stepsizePow+2\localscalePow-2\batchsizePow} \frac{1}{2}\precon\hess\testfun(\limparm)\precon\transpose \\
		&\qquad \frobprod \cEE{\dataset{\Nats}}
				\rbra{\sum_{\mbidx\in \range{\batchsize\upper{\samplesize}}} \rbra{\grad\loglik\rbra{\parmLMLE{{\samplesize}}
						+ (\localscale\upper{\samplesize})^{-1} \limparm;\ \data_{\mbdataidx{\samplesize}{1}{\mbidx}}}
						-\grad\loglik\rbra{\parmLMLE{{\samplesize}}
								;\ \data_{\mbdataidx{\samplesize}{1}{\mbidx}} }} }^{\otimes2} \\
		& \approx \frac{\const{\stepsize}^2}{\const{\batchsize}^2}\subseqsamplesize^{\timescalePow-2\stepsizePow+2\localscalePow-2\batchsizePow} \frac{1}{2}\precon\hess\testfun(\limparm)\precon\transpose \frobprod \cEE{\dataset{\Nats}}
				\rbra{\sum_{\mbidx\in \range{\batchsize\upper{\samplesize}}} \hess\loglik\rbra{\parmLMLE{{\samplesize}};\ \data_{\mbdataidx{\samplesize}{1}{\mbidx}}}(\localscale\upper{\samplesize})^{-1} \limparm }^{\otimes2}			\\
		& = \frac{\const{\stepsize}^2}{\const{\batchsize}^2}\subseqsamplesize^{\timescalePow-2\stepsizePow-2\batchsizePow} \frac{1}{2}\precon\hess\testfun(\limparm)\precon\transpose \frobprod \cEE{\dataset{\Nats}}
				\rbra{\sum_{\mbidx\in \range{\batchsize\upper{\samplesize}}} \hess\loglik\rbra{\parmLMLE{{\samplesize}};\ \data_{\mbdataidx{\samplesize}{1}{\mbidx}}}\limparm }^{\otimes2}	\\
		& \approx \subseqsamplesize^{\timescalePow-2\stepsizePow-2\batchsizePow} \frac{1}{2}\precon\hess\testfun(\limparm)\precon\transpose \frobprod \sbra{
		\batchsize\upper{\samplesize}(\batchsize\upper{\samplesize} -1) \finfoOpt\limparm\limparm\transpose\finfoOpt + \batchsize\upper{\samplesize} K(\trueparm; \limparm)
		}	\\
\]
where $K(\trueparm; \limparm) = \int \hess\loglik(\trueparm; x)\ \limparm^{\otimes2}\ \hess\loglik(\trueparm; x) \datadist(dx)$.

Now, we recall that for the drift term to be non-zero in the limit, we need $\timescalePow=\stepsizePow$. However, at any such scaling the \ref{term2:2-loglik-loglik-cv} term is $\Oo(\samplesize^{-\stepsizePow-2\batchsizePow})$, and so is always $0$ in the limit.
\section{Sketch Proof for constrained parameter spaces}
\label{apx:sketch-bounded}

Let $\boundarymap: \parmspace \times (\Reals^d)^3 \to \parmspace$ be a measurable function such that:
\begin{enumerate}
  \item[(i)] $\boundarymap$ is \emph{faithful} to $\parmspace$, meaning that
if $\convhull(\parm, \parm + \incrPriorNup + \incrLoglikNup + \incrInnovNup) \subset\parmspace$ then
    \[ \label{eq:boundaryfun-faithful}
      \boundarymap(\parm, \incrPriorNup, \incrLoglikNup, \incrInnovNup) = \parm + \incrPriorNup + \incrLoglikNup + \incrInnovNup,
    \]
	where $\convhull(\parm_1,\parm_2)$ is the line segment from $\parm_1$ to $\parm_2$.
  \item[(ii)] $\boundarymap$ is \emph{local}, meaning that there exists $\const{\boundarymap}>0$ such that for all $(\parm, \incrPriorNup, \incrLoglikNup, \incrInnovNup) \in \parmspace \times (\Reals^d)^3$
    \[ \label{eq:boundaryfun-local}
      \norm{\boundarymap(\parm, \incrPriorNup, \incrLoglikNup, \incrInnovNup) - \parm}
        & \leq \const{\boundarymap} \rbra{\norm{\incrPriorNup}+\norm{\incrLoglikNup}+\norm{\incrInnovNup}}.
    \]
\end{enumerate}

We will consider the iterative algorithm on $\parmspace$ given by
\[\label{eq:sgld-iter-update-boundary}
	\parmiter{\samplesize}{\iternum+1}
		& =
    \boundarymap\rbra{
      \parmiter{\samplesize}{\iternum},\
      \frac{\stepsize \precon}{2\samplesize} \grad \log\prior\rbra{\parmiter{\samplesize}{\iternum}},\
      \frac{\stepsize \precon}{2} \frac{1}{\batchsize}\sum_{\mbidx\in \range{\batchsize}} \grad\loglik\rbra{ \parmiter{\samplesize}{\iternum};\ \data_{\mbdataidx{\samplesize}{\iternum}{\mbidx}}},\
      \sqrt{\stepsize\invtemp^{-1}\aniso}\ \innoviter{\iternum}
    }.
\]

The key idea is that, if $\trueparm \in \interior(\parmspace)$, there is a $r>0$ with $\trueparm \in \ball(\trueparm,r)\subset \interior(\parmspace)$, and for any compactly supported test function $\testfun$ and compact extension of its support, $\compactset_1$, for sufficiently large sample sizes $\samplesize$, $\compactset_1 \subseteq  \ball(0,\localscale\upper{\samplesize} r)$.
In the proof of the $\parmspace = \Reals^\parmdim$ case we found that, along sub-sequences ($\subsubseqsamplesize$), the increments from the log-likelihood and from the prior vanish uniformly within a sufficiently large extension of the support of $\testfun$. 
Combining this with faithfulness of $\boundarymap$ (defined in \cref{sec:constrained-parmspace}) and an application of the Lebesgue dominated convergence theorem to handle truncation of the Gaussian increments shows that the $\generator_{\subseqsamplesize}\testfun \to \generator \testfun$ uniformly within the extension of the support of $\testfun$ when $\parmspace \neq \Reals^\parmdim$.
Moreover, the local property of the boundary condition (defined in \cref{sec:constrained-parmspace}) ensures that for sufficiently large sample sizes, if the process were far enough outside of the support of $\testfun$ then it cannot re-enter the support via an arbitrarily large jump caused by the boundary condition. Thus, outside of the extension of the support of $\testfun$, the deviation of $\generator_{\subseqsamplesize}\testfun$ from $0$ is essentially indistinguishable from the unconstrained case. Using those two facts we can rely on the faithfulness of the boundary dynamics to ensure that the process converges weakly to the same Ornstein-Uhlenbeck limit as in the unconstrained case.
\section{Further discussion of asymptotics of mixing times}
\label{apx:mixing-asympt}
The discussion of the implications on the mixing time from \cref{sec:implications-mixingtime} is only a heuristic because, even if the process converge weakly and the stationary distributions converge weakly, it is insufficient to conclude that the mixing times converge. 
Instead the mixing time of limiting process corresponds to fixing a duration of scaled time for which to run the process, say $T$, then computing the limit of the covariance of an estimator based on the run up to time $T$, then letting $T$ tend to infinity. The mixing time of the limit is of more practical relevance for our understanding of the local process since it accurately reflects the time needed for the limiting stationary distribution to provide a good approximation to a sample from the local process. On the other hand the limit of mixing times determines how long it would take to visit other modes if they exist, and would often tend to $\infty$ with sample size. This can be seen by considering a simple non-identifiable model, for example Gaussian location clustering, for which there would be two identical optimal solutions which differ only by permutations of the clusters. The limit of mixing times corresponds to the time it takes to explore both modes, while the mixing time of the limit corresponds to the time needed to explore the model closer to which the process is started. Even if there was not a second equally good mode, a second suboptimal mode that persists (though shrinking) at all sample sizes, and is moving farther away as the process is re-scaled, could lead to mixing times that do not converge.

In future work, we plan to introduce a more rigorous characterization of the correspondence between limit of mixing times and the mixing time of the limiting process. In particular, \citet{atchade2021approximate} introduces the $\zeta$-spectral gap, defined as
\[
	\textrm{SpecGap}_\zeta
		& \defas \inf \bigg\{\frac{\pi[f^2] - \inner{f}{Pf}_{L^2(\pi)}}{\pi[f^2] - \zeta/2} \ \big\vert \ \\
		&\qquad\qquad f\in L^2(\pi),\ \pi f=0,\ \pi[f^2] > \zeta,\ \norm{f}_{L^2(\pi)}<\infty \bigg\}.
\]
We conjecture that for any $\zeta>0$, under appropriate scaling (corresponding to the time rescaling factor $\timescale\upper{\samplesize}$), if the sequence of posterior distributions is tight, then the $\zeta$-spectral gap will converge to that of the OU-process for all $\zeta>0$. This is supported by the intuitive interpretation of the $\zeta$-spectral gap; that it corresponds to the mixing time of the process within a local region containing most of the probability mass of the stationary distribution. Under the tightness assumption we expect that this is sufficient to rule out the types of pathological behaviour described in the previous paragraph.

\section{Additional Details for Experimental Results}
\label{apx:experiment-supp}

\begin{table}[ht]
  \centering
  \begin{tabular}{l|l|l|l}
      & Experiment 1 & Experiment 2 & Experiment 3 \\ \hline
    true distribution
      & $\normaldist_{10}\rbra{0, \frac{1}{2} I + \frac{1}{2} \one\one'}$
      & unknown
      & unknown
      \\
    log-likelihood $\loglik(\cdot;\parm)$
      & $\sum_{i=1}^{10} \frac{(\dumdata_i- \parm_i)^2}{\sqrt{i}}$
      & $yx\transpose\theta - \log(1+e^{x\transpose\theta})$
      & $y x\transpose\theta - \exp(x\transpose\theta)$
      \\
    log-prior $\log\prior(\parm)$ 
      & $0$
      & $0$
      & $0$
      \\
    sample size $\samplesize$
      & $1000$
      & $1000000$
      & $150000$
      \\
    batch size $\batchsize$
      & $1$
      & $1000$
      & $250$
      \\
    number of steps $\iternum$
      & $10000 \samplesize / \batchsize$
      & $1000 \samplesize / \batchsize$
      & $1000 \samplesize / \batchsize$
      \\
    step size (SGD) $\stepsize$
      & $4\batchsize/\samplesize$
      & $4\batchsize/\samplesize$
      & $4\batchsize/\samplesize$
      \\
    step size (SGLD) $\stepsize$
      & $2\batchsize/\samplesize$
      & $\batchsize/\samplesize$
      & $2\batchsize/\samplesize$
      \\
    inv. temp. (SGLD) $\invtemp$
      & $2$
      & $1$
      & $2$
      \\
  \end{tabular}
  \caption{Settings for experiments 1, 2, \& 3. When the true distribution is unknown it is approximated by the empirical distribution on a larger version of the dataset for these experiments.}
  \label{tab:exp1-2-3-settings}
\end{table}

\begin{figure}[ht]
  \centering
  \begin{subfigure}[b]{0.45\textwidth}
    \includegraphics[width=\textwidth]{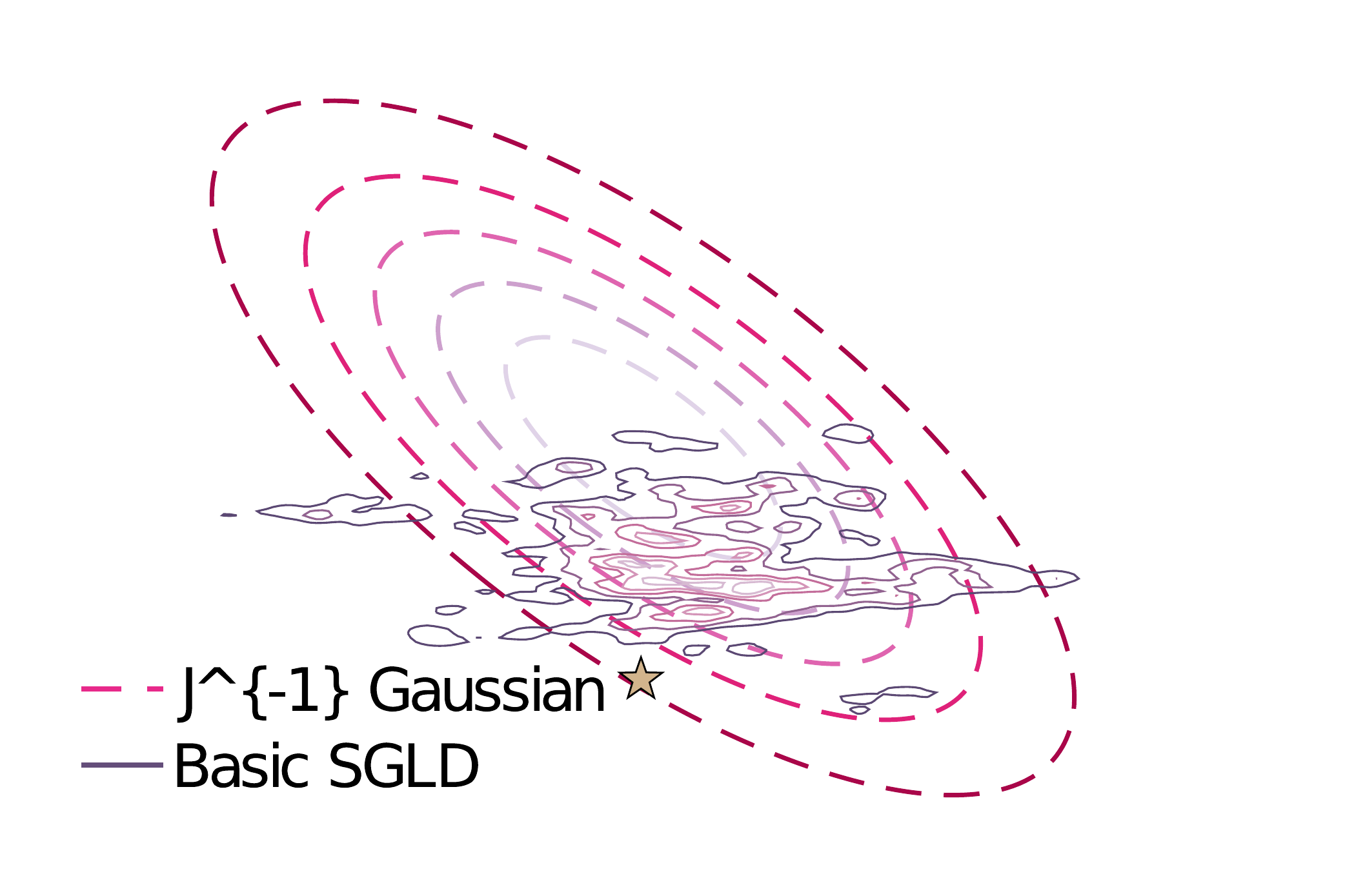}
    \caption{SGLD without Preconditioning}
  \end{subfigure}
  \begin{subfigure}[b]{0.45\textwidth}
    \includegraphics[width=\textwidth]{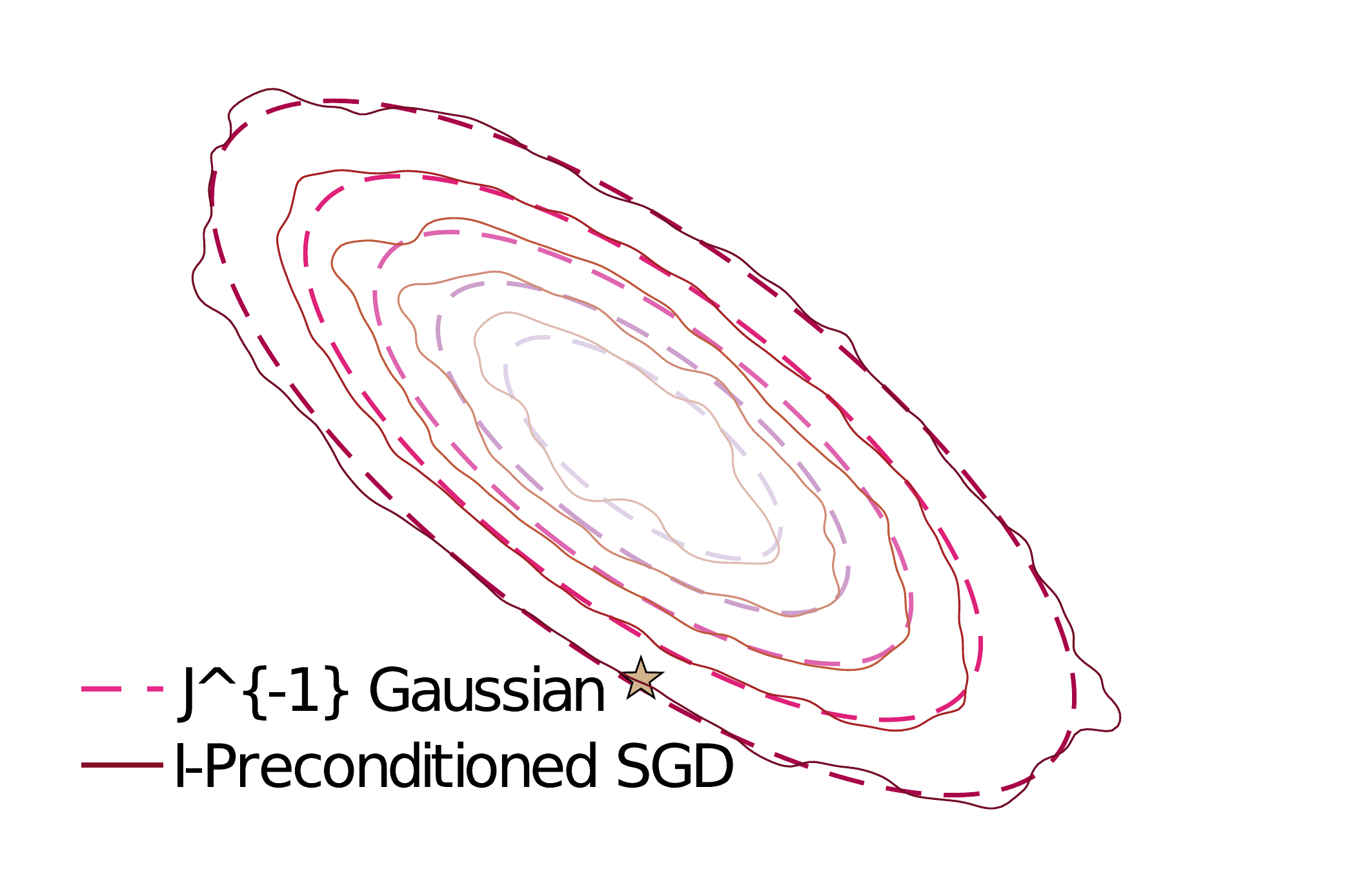}
    \caption{SGD Preconditioned by $\vinfoOpt$}
  \end{subfigure}
  \caption{Joint results of experiment 2: Parameters 1 and 4}
  \label{fig:exp2-results-joint}
\end{figure}

\begin{figure}[ht]
  \centering
  \includegraphics[width=0.5\textwidth]{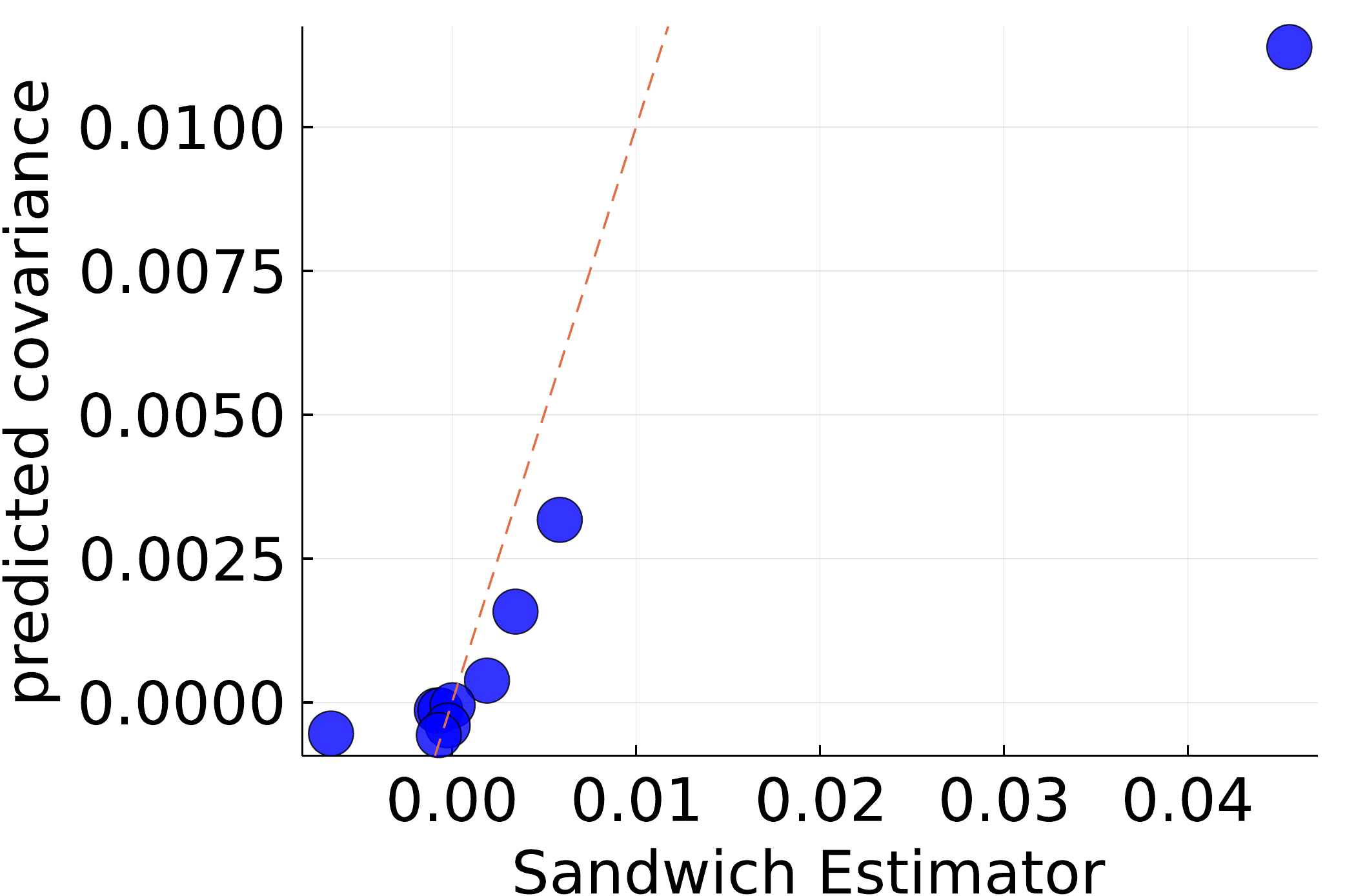}
  \label{fig:exp2plus-results-a}
  \caption{Further result for experiment 2 comparing the scaled sandwich covariance estimator \cref{eq:asymptotic-iteravg-cov} to the predicted values variance-covariance matrix based upon \cref{eq:asymptotic-iteravg-cov-simple} for iterate averages when $\stepsizePow+\batchsizePow=1$.  We see that the higher order correction is material in this case, as expected based upon the theoretical results.}
  \label{fig:exp2-results-supp}
\end{figure}

\end{document}